\documentclass[onefignum,onetabnum]{siamonline171218}

\usepackage{lipsum}
\usepackage{amsfonts}
\usepackage{graphicx}
\usepackage{epstopdf}
\usepackage{algorithmic}
\usepackage{caption}
\usepackage{subcaption}
\usepackage{bm}
\usepackage{bbm}
\usepackage{booktabs}
\usepackage[numbers]{natbib}
\ifpdf
  \DeclareGraphicsExtensions{.eps,.pdf,.png,.jpg}
\else
  \DeclareGraphicsExtensions{.eps}
\fi

\usepackage{enumitem}
\setlist[enumerate]{leftmargin=.5in}
\setlist[itemize]{leftmargin=.5in}

\newsiamremark{remark}{Remark}
\newsiamremark{hypothesis}{Hypothesis}
\crefname{hypothesis}{Hypothesis}{Hypotheses}
\newsiamthm{claim}{Claim}

\headers{Semi-static replication for callable IR derivatives}{J. H. Hoencamp, S. Jain, and B. D. Kandhai}

\title{A semi-static replication approach to efficient hedging and pricing of callable IR derivatives
\thanks{Compiled January, 2022.
\funding{This project has received funding from the NWO under the Industrial Doctorates grant}}}

\author{Jori Hoencamp\thanks{Computational Science Lab, University of Amsterdam, Science Park 904, 1098XH Amsterdam, Netherlands 
  (\email{j.h.hoencamp@uva.nl}).}
\and Shashi Jain\thanks{Department of Management Studies, Indian Institute of Science, Bangalore, India.}
\and Drona Kandhai\footnotemark[2]}

\usepackage{amsopn}

\makeatletter
\newcommand*{\addFileDependency}[1]{
  \typeout{(#1)}
  \@addtofilelist{#1}
  \IfFileExists{#1}{}{\typeout{No file #1.}}
}
\makeatother

\newcommand*{\myexternaldocument}[1]{%
    \externaldocument{#1}%
    \addFileDependency{#1.tex}%
    \addFileDependency{#1.aux}%
}

\ifpdf
\hypersetup{
  pdftitle={A semi-static replication approach to efficient pricing of callable IR derivatives},
  pdfauthor={D. Doe, P. T. Frank, and J. E. Smith}
}
\fi

\myexternaldocument{ex_supplement}

\begin{document}

\maketitle

\begin{abstract}
We present a semi-static hedging algorithm for callable interest rate derivatives under an affine, multi-factor term-structure model. With a traditional dynamic hedge, the replication portfolio needs to be updated continuously through time as the market moves. In contrast, we propose a semi-static hedge that needs rebalancing on just a finite number of instances. We show, taking as an example Bermudan swaptions, that callable interest rate derivatives can be replicated with an options portfolio written on a basket of discount bonds. The static portfolio composition is obtained by regressing the target option's value using an interpretable, artificial neural network. Leveraging on the approximation power of neural networks, we prove that the hedging error can be arbitrarily small for a sufficiently large replication portfolio. A direct, a lower bound, and an upper bound estimator for the risk-neutral Bermudan swaption price is inferred from the hedging algorithm. Additionally, closed-form error margins to the price statistics are determined. We practically demonstrate the hedging and pricing performance through several numerical experiments.
\end{abstract}



\section{Introduction}
To date, the valuation of path-dependent derivatives, such as Bermudan- or American-style options, remains a challenge due to the absence of analytical pricing formulas. This challenge has become particularly pronounced since the implementation of the Basel accords, by which the computation of risk-measures such as value-at-risk, expected shortfall and various value-adjustments has become a central requirement for many financial institutions \cite{gregory2015xva}. As these measures are typically simulation-based quantities, the search for efficient and accurate pricing techniques is as relevant as ever. For callable derivative contracts, practitioners resort to numerical approximation schemes, of which many have been proposed in the literature. Traditional examples include tree-based methods, such as \cite{cox1979option}, \cite{boyle1988lattice}, \cite{broadie1996american}, and PDE-based methods, such as \cite{brennan1977valuation}, \cite{mitchell1999finite}, \cite{haentjens2015adi}.

If the risk factors that drive the product-value are high-dimensional, the traditional methods are no longer feasible. Within such multi-factor frameworks, Monte Carlo methods tend to be a popular alternative. Classic regression-based algorithms have been proposed by \cite{carriere1996valuation}, \cite{tsitsiklis1999optimal}, and \cite{longstaff2001valuing}. Using the iterative dynamic programming formulation of path-dependent options, they approximate the continuation values at each exercise opportunity with ordinary least-square regression. Their method gives rise to an exercise policy that yields a lower bound to the true price. Upper bounds can be generated by considering a dual-formulation of the risk-neutral value, as is proposed by \cite{haugh2004pricing}, and \cite{rogers2002monte}. The estimation of these bounds, however, relies on expensive nested simulations \cite{andersen2004primal}. Regression with systems of basis functions other than polynomials have been studied in \cite{kusuoka2015least}, alongside their rate of convergence and estimation for the corresponding regression error. Although classic regression-based schemes may be accurate for computing today's option price, they can be less precise in generating forward values along the Monte Carlo paths, which is addressed by for example \cite{joshi2016least}.

Apart from pricing, hedging is a challenging aspect of trading path-dependent derivatives. The traditional dynamic hedge is achieved by constructing a replicating portfolio that is rebalanced continuously through time as the market moves. In contrast, a static hedge is a portfolio of assets that mirrors the value of a contract without the need of updating the portfolio-composition. The weights of the portfolio are so to speak \textit{static}. Such a hedge can be favoured over a dynamical strategy for several reasons. Frequently updating a portfolio can be expensive if the transaction costs for the individual assets are high. Another reason is that dynamic hedges can suffer from severe hedging errors in case of sudden market movements as a consequence of discrete rebalancing.

A static hedge formulation for exotic equity options has been proposed in the literature by for example \cite{breeden1978prices}, \cite{carr1994static}, \cite{carr1999static} and \cite{carr2014static}. The main concept is to construct an infinite portfolio of short-dated European options with a continuum of different strike-prices. A different, but comparable approach is proposed in \cite{derman1995static}. Here a portfolio of European options with a continuum of different maturities is constructed to replicate the boundary and terminal conditions of exotic derivatives, such as knock-out options. Static replication of an American-style option is challenging as it involves a time-dependent exercise boundary, giving rise to a free boundary problem. In \cite{chung2009static}, this problem is addressed by combining European options with multiple strikes and maturities, and in \cite{lokeshwar2019neural} a semi-static hedge is constructed using neural network approximations.

Little attention has been given to this topic in the field of interest rate (IR) modelling. Where equity options depend on the realization of a stock, IR derivatives depend on the realization of a full term-structure of interest rates, pushing the complexity of the hedging problem. The articles of \cite{pelsser2003pricing} and \cite{hagan2005convexity} are among the few contributions to the literature, treating static replication of guaranteed annuity options, and CMS swaps, caps and floors respectively with a portfolio of European swaptions. In our work we extend the literature by addressing the replication and pricing problem of path-dependent IR derivatives.


We focus on the replication of Bermudan swaptions under an affine term-structure short-rate model. First, we show that such a contract can be semi-statically hedged by a portfolio of short-maturity IR options, such as discount bond options. While a static hedge is never updated and a dynamic hedge requires continuous updating, our semi-static hedge is updated on a finite number of instances, i.e. at each monitor date of the contract. After rebalancing, the replication portfolio mirrors the contract's value until the next monitor date, after which it terminates. Its pay-off will be similar to that of the contract if it is exercised or will absorb the costs to set up a new hedge in case it is continued. In particular, we show that under a one-factor model, an options portfolio written on a single discount bond suffices to achieve an accurate hedge. Under a $d$-dimensional multi-factor model, we propose that options written on a basket of discount bonds representing $d$ maturities are required to account for the yield curve's temporal complexity.

Secondly, we prove that with a sufficiently large hedge portfolio, any desired level of accuracy in replication can be achieved. We do so by reformulating the optimization related to the hedging problem as a fitting procedure of a shallow, artificial neural network (ANN). This allows us to utilize the tremendous developments that have recently been achieved in machine learning, which has spiked in popularity amongst both researchers and practitioners over the past decade. An often raised drawback to ANN applications in quantitative finance is the lack of interpretability and their ``black-box" nature. However, in our approach, the ANN's structure is chosen deliberately to represent the pay-off of a standard derivative portfolio. Training the free parameters can therefore be interpreted as searching the weights and strikes of this portfolio. Additionally, we suggest specific ANN designs allowing the portfolio to be priced analytically. An essential result is that, not only the time-zero, but any future value approximation of the contract can be obtained in closed-form by simply pricing this portfolio.


Our contribution to the existing literature is threefold. First, we propose semi-static hedging strategies for Bermudan swaptions under a multi-factor short-rate model. In the one-factor case, we argue that replication can be achieved with an options portfolio written on a single discount bond. In the $d$-dimensional case, replication can be achieved with an options portfolio written on a basket of discount bonds. Second, we propose a direct estimator, a lower and an upper bound estimate to the contract's value, which is implied by the semi-static hedge. The lower bound value results from applying a non-optimal exercise strategy on an independent set of Monte Carlo paths. The upper bound is based on the dual formulation of \cite{haugh2004pricing} and \cite{rogers2002monte}, which in contrast to other work can be obtained without resorting to expensive nested simulations. Thirdly, we prove that any desired level of accuracy can be achieved in the hedge and valuation, due to the approximating power of interpretable ANNs. We provide clear, analytic error margins to the price estimates and the related lower and upper bounds. All the mentioned statistics are illustrated and benchmarked in representative numerical examples.

The paper is organised as follows: \Cref{sec: Math formulation} introduces the mathematical setting, describes the modelling framework and provides the problem formulation. \Cref{sec: RLNN} provides a thorough introduction to the algorithm, motivates the use and interpretation of neural networks and treats the regression procedure. \Cref{bounds theory} introduces the lower bound and upper bound estimates to the true option price. In \cref{sec: error analysis} we introduce the error bounds on the direct, lower bound and upper bound estimates brought forth by the algorithm. We finalize the paper by illustrating the method through several numerical examples in \cref{sec: numerical examples} and providing a conclusion in \cref{sec: conclusion}.

\section{Mathematical background\label{sec: Math formulation}}
In this section we describe the general framework for our computations and give a detailed introduction to the Bermudan swaption pricing problem.

\subsection{Model formulation\label{sec: model formulation}}
We fix a finite time-horizon $T>0$ and consider a stochastic economy defined on the time-interval $[0,T]$. Let $(\Omega, \mathcal{F},\mathbb{P})$ denote a complete probability space with $\mathcal{F}$ a sigma-algebra on $\Omega$. We denote by $\mathbb{F}=\left\{\mathcal{F}_t:0\leq t\leq T\right\}$ an augmented filtration on $\Omega$ such that $\mathcal{F}_T:=\mathcal{F}$, representing the information generated by the economy up to time $t$. We assume that the market is complete and free of arbitrage. Following the work of \cite{harrison1979martingales} we denote by $\mathbb{Q}$ the risk-neutral measure equivalent to $\mathbb{P}$, to which the unique no-arbitrage price of any attainable contingent claim is associated. As the num\'eraire related to $\mathbb{Q}$ we consider the bank account $B$ defined as
\[
B(t):=e^{\int_0^tr(u)du},\qquad t\in [0,T]
\]
where the stochastic process $r(t)$ denotes the risk-free instantaneous short-rate. By letting $B(0)=1$, $B(t)$ represents the time-$t$ value of one unit of currency invested in the money-market at time zero. Following the notation of \cite{brigo2007interest}, we denote by $P(T_1,T_2)$ the time-$T_1$ risk-neutral value of a zero-coupon bond maturing at date $T_2$, which is given by
\begin{equation*}
\begin{aligned}
    P(T_1,T_2):=\mathbb{E^Q}\left[\frac{B(T_1)}{B(T_2)}\bigg | \mathcal{F}_{T_1}\right], \quad 0\leq T_1<T_2\leq T
\end{aligned}
\end{equation*}

We assume that the dynamics of the short-rate $r$ are captured by an affine term-structure model, in accordance with the set-up introduced in \cite{duffie1996yield} and \cite{dai2000specification}. The short-rate itself is therefore considered to be an affine function of a - possibly multi-dimensional - latent factor $\mathbf{x}_t$, i.e.
\begin{equation} \begin{aligned}
    r(t)=\omega_1+\mathbf{\omega_2}^\top \mathbf{x}_t \label{eqn: short-rate}
\end{aligned} \end{equation}
with $\omega_1$, $\mathbf{\omega}_2$ denoting a scalar and a vector of  time-dependent coefficients respectively. We furthermore assume that the stochastic process $\left\{\mathbf{x}_t\right\}_{t\in[0,T]}$ is a bounded Markov process that takes values in $\mathbb{R}^d$, which represents all market influences affecting the state of the short-rate. Let the dynamics of $\mathbf{x}_t$ be governed by an SDE of the form \begin{equation} \begin{aligned}
    d\mathbf{x}_t=\mu(t,\mathbf{x}_t)dt+\sigma (t,\mathbf{x}_t)d \mathbf W_t \label{eqn:SDE x1}
\end{aligned} \end{equation}
where $\mathbf W_t$ denotes an $\mathbb{R}^d-$ valued Brownian motion under $\mathbb{Q}$ adapted to the filtration $\mathbb{F}$. The measurable functions $\mu:[0,T]\times\mathbb{R}^d\to\mathbb{R}^d$ and $\sigma:[0,T]\times\mathbb{R}^d\to\mathbb{R}^{d\times d}$ are taken to satisfy the standard regularity conditions by which the SDE in \ref{eqn:SDE x1} admits a strong solution.

In this article we will limit ourselves to the subclass of affine, Gaussian term-structure models. We do so by imposing that the instantaneous drift $\mu$ is an affine function of the latent factors $\mathbf{x}_t$ and that the diffusion is a deterministic function of time. As a result, $\mathbf{x}_t$ yields a Gaussian process and the zero-coupon bond prices have a closed-form expression of the form 
\begin{equation*} \begin{aligned}
    P(T_1,T_2)=\exp\left\{A(T_1,T_2)-B(T_1,T_2)\cdot\mathbf{x}_{T_1}\right\}
\end{aligned} \end{equation*}
where the deterministic coefficients $A(T_1,T_2)\in\mathbb{R}$ and $B(T_1,T_2)\in\mathbb{R}^d$ can be found by solving a system of ODEs, which are of the form of the well-known Ricatti equations; see \cite{duffie1996yield} or \cite{filipovic2009term} for details. We focus on this subclass as it is still rich enough for many risk-related applications, but on the other hand is analytically very tractable. Not only does it yield an explicit expression for the discount bond, but also for European options on discount bonds (see \cite{brigo2007interest}). This will prove to be a convenient property for the techniques introduced in the subsequent sections. 

\subsection{The Bermudan swaption pricing problem}
We consider the pricing problem of a Bermudan swaption. A Bermudan swaption is a contract that gives the holder the right to enter a swap with fixed maturity at a number of predefined monitor dates. Should the holder at any of the monitor dates decide to exercise the option, the holder immediately enters the underlying swap. The lifetime of this swap is assumed to be equal to the time between the exercise date and a fixed maturity date $T_M$.

As an underlying we take a standard interest rate swap that exchanges fixed versus floating cashflows. For simplicity we will assume that the contract is priced in a single-curve framework and that cashflow schemes of both legs coincide, yielding fixing dates $\mathcal{T}_f=\left\{T_0,\ldots,T_{M-1}\right\}$ and payment dates $\mathcal{T}_p=\left\{T_1,\ldots,T_{M}\right\}$. However, we stress that the algorithm is applicable to any industry standard contract specifications and is not limited to the simplifying assumptions that are made here. The time-fraction between two consecutive dates is denoted as $\Delta T_m = T_m-T_{m-1}$. Let $N$ be the notional and $K$ the fixed rate of the swap. Assuming the holder of the option exercises at $T_m$, the payments of the swap will occur at $T_{m+1},\ldots, T_M$.

We consider the class of pricing problems, where the value of the contract is completely determined by the Markov process $\left\{\mathbf{x}_t\right\}_{t\in[0,T]}$ in $\mathbb{R}^d$ as defined in \cref{sec: Math formulation}. Let $h_m:\mathbb{R}^d\to \mathbb{R}$ be the $\mathcal{F}_{T_m}$-measurable function denoting the immediate pay-off of the option if exercised at time $T_m$. Although the methodology holds for any generalization of the functions $h_m$, we will consider those in accordance with the contract specifications described above. This means that the functions $h_m$ are assumed to be given by
\begin{equation*} \begin{aligned}
    h_m(\mathbf{x}_{T_m}):=\delta\cdot N \cdot A_{m,M}\left(T_m\right)\left(S_{m,M}\left(T_m\right)-K\right)
\end{aligned} \end{equation*}
where the indicator $\delta=1$ infers a payer and $\delta=-1$ infers a receiver swaption. 
The swap rate $S_{m,M}$ and the annuity $A_{m,M}$ are defined in the same fashion as \cite{brigo2007interest}, given by the expressions
\[
S_{m,M}(t) = \frac{\sum_{j=m+1}^{M}\Delta T_j P(t,T_j)F(t,T_{j-1},T_j)}{\sum_{j=m+1}^{M}\Delta T_j P(t,T_j)},\quad A_{m,M}(t)=\sum_{j=m+1}^{M}\Delta T_j P(t,T_j)
\]
where the function $F$ denotes the simply compounded forward rate given by the expression
\[
F\left(t,T_{j-1},T_j\right)=\frac{1}{\Delta T_j}\left(\frac{P\left(t,T_{j-1}\right)}{P\left(t,T_j\right)}-1\right)
\]
for any $j\in\{1,\ldots,M\}$. For details we refer to \cite{brigo2007interest}.

Now let $\mathbb{T}$ denote the set of all discrete stopping times with respect to the filtration $\mathbb{F}$, taking values on the grid $\mathcal{T}_f\cup\{\infty\}$. Define the function $h_\tau$ as
\begin{equation} \begin{aligned}
    h_\tau(\mathbf{x}_\tau):=h_{\tau(\omega)}(\mathbf{x}_\tau(\omega))=\begin{cases}h_m\left(\mathbf{x}_{T_m}\right) & \text{if }\tau(\omega)=T_m\\ 0 & \text{if }\tau(\omega)=\infty\end{cases},\qquad \omega\in\Omega \label{eqn: stopping time}
\end{aligned} \end{equation}
In this notation, $\tau(\omega)=\infty$ indicates that the option is not exercised at all. We aim to approximate the time-zero value of the Bermudan swaption, which satisfies the following equation
\begin{equation} \begin{aligned}
    V(0) =\sup_{\tau\in\mathbb{T}} \mathbb{E^Q}\left[\frac{h_\tau(\mathbf{x}_\tau)}{B(\tau)}\bigg|\mathcal{F}_{0}\right] \label{eqn:Berm1}
\end{aligned} \end{equation}
Finding the optimal exercise strategy $\tau$ is typically a non-trivial exercise. Numerical approximations for $V(0)$ can however be computed by considering a dynamical programming formulation as given below, which is shown to be equivalent to \ref{eqn:Berm1} in for example \cite{glasserman2013monte}. Let $t\in(T_m,T_{m+1}]$ for some $m\in \{0,\ldots,M-2\}$ and denote by $V(t)$ the value of the option, conditioned on the fact that it is not yet exercised prior to $t$. This value satisfies the equation (see \cite{glasserman2013monte})

\begin{equation} \begin{aligned}
    V(t)=\begin{cases}
    \max \left\{h_{M-1}\left(\mathbf{x}_{T_{M-1}}\right),\;0\right\} & \text{if }t=T_{M-1}\\
    \max \left\{h_{m}\left(\mathbf{x}_{t}\right),\;B(t)\mathbb{E^Q}\left[\frac{V(T_{m+1})}{B(T_{m+1})}\Big|\mathcal{F}_{t}\right]\right\} & \text{if }t=T_m,\;m\in\{0,\ldots,M-2\}\\
    B(t)\mathbb{E^Q}\left[\frac{V(T_{m+1})}{B(T_{m+1})}\Big|\mathcal{F}_{t}\right] & \text{if }t\in(T_{m},T_{m+1}),\;m\in\{0,\ldots,M-2\}
    \end{cases}\label{eqn: dynamic formulation}
\end{aligned} \end{equation}
We refer to the random variables $C_m(t):=B(t)\mathbb{E^Q}\left[\frac{V(T_{m+1})}{B(T_{m+1})}\Big|\mathcal{F}_{t}\right]$ as the hold or continuation values. It represents the expected value of the contract if it is not being exercised up until $t$, but continues to follow the optimal policy thereafter. Approximations of the dynamic formulation are typically obtained by a backward iteration based on simulations of the underlying risk-factors. Objective is then to determine the continuation values as a function of the state of the risk-factor $\mathbf{x}_{t}$. Popular numerical schemes based on regression have been introduced in for example \cite{carriere1996valuation} and \cite{longstaff2001valuing}.

Based on approximations of the continuation values, the optimal policy $\tau$ can be computed as follows. Assume that for a given scenario $\omega\in\Omega$, the risk-factor takes the values $\mathbf{x}_{T_0}=x_0,\ldots,\mathbf{x}_{T_{M-1}}=x_{M-1}$. Then the holder should continue to hold the option if $C_m(T_m)>h_m(x_m)$ and exercise as soon as $C_m(T_m)\leq h_m(x_m)$. In other words, the exercise strategy can be determined as
\begin{equation*} \begin{aligned}
    \tau(\omega)=\min\left\{T_m\in\mathcal{T}_f\big|C_m(T_m)\leq h_m(x_m)\right\}
\end{aligned} \end{equation*}
Should for some scenario the continuation value be bigger than the immediate pay-off for each monitor date, then $\tau(\omega)=\infty$ and the option expires worthless.

\section{A semi-static hedge for Bermudan swaptions\label{sec: RLNN}}
The main concept of the method that we propose is to construct static hedge portfolios that replicate the dynamical formulation in \ref{eqn: dynamic formulation} between two consecutive monitor dates. In this section we introduce the algorithm for a Bermudan swaption that is priced under a multi-factor affine term-structure model. The methodology utilizes a regress-later technique in which the intermediate option values are regressed against simple IR assets, such as discount bonds. The regression model is chosen deliberately to represent the pay-off of an options portfolio written on these assets. An important consequence is that the hedge can be valued in closed-form. For an introduction to regress-later approaches, we refer to \cite{glasserman2004simulation}.


\subsection{The algorithm\label{sec: algorithm}}
The regress-later algorithm is executed in an iterative manner, backward in time. The outcome is a set of option portfolios $\left\{\Pi_{M-1},\ldots,\Pi_0\right\}$ written on pre-selected IR assets. To be more precise, the algorithm determines the weights and strikes of each portfolio $\Pi_m$, such that it closely mirrors the Bermudan swaption after its composition at $T_{m-1}$ until its expiry at $T_m$. The pay-off of $\Pi_m$ exactly meets the cost of composing the next portfolio $\Pi_{m+1}$ or the Bermudan's pay-off in case it is exercised. The methodology yields a semi-static hedging strategy as the portfolio compositions are constant between two consecutive monitor dates. Hence there is no need for continuous rebalancing, as is the case for a dynamic hedging strategy. The algorithm can roughly be divided into three steps, presented below. \cref{alg: rlnn} summarizes the method.
\paragraph{Sample the independent variables}
We start by sampling $N$ realizations of the risk-factor $\mathbf{x}_t$ on the time grid $\mathcal{T}=\left\{T_0,\ldots,T_{M-1}\right\}$. These realizations will serve as an input for the regression data. We will denote the data points as $\hat x:=\left\{\left(x_{T_0}^n,\ldots,x_{T_{M-1}}^n\right)\right\}_{n=1}^N$. Different sample methodologies could be used, such as:
\begin{itemize}
    \item Take a standard quadrature grid for each monitor date $T_m$, associated with the transition density of the risk-factor. For example, if $\mathbf{x}_t$ has Gaussian dynamics one could consider the Gauss-Hermite quadrature scaled and shifted in accordance with the mean and variance of $\mathbf{x}_t$. See for example \cite{xiu2010numerical}.
    \item Discretize the SDE of the risk-factor and sample by the means of an Euler or Milstein scheme. Make sure that sufficiently coarse time-stepping grid is used, which includes the $M$ monitor dates. See for example \cite{kloeden2013numerical} for details.
\end{itemize}
Secondly we select an asset that will serve as the independent variable for the regression. We will denote this asset as $z_m(t)$. The choice for $z_m$ can be arbitrary, as long as it meets the following conditions:
\begin{itemize}
    \item The asset $z_m(T_m)$ should be a square integrable random variable that is $\mathcal{F}_{T_m}$ measurable, taking values in $\mathbb{R}^d$.
    \item The risk-neutral price of $z_m(t)$ should be only dependent on the current state of the risk-factor and be almost surely unique. That is, the mapping $\mathbf{x}_{T_{m}}\mapsto z_m|\mathbf{x}_{T_{m}}$ should be continuous and injective. This is required to guarantee a well-defined parametrization of the option-value.
\end{itemize}
Examples for $z_m$ would be a zero-coupon bond, a forward Libor rate or a forward swap rate. For each sampled realization of the risk-factor, the corresponding realization of the asset value will be computed and denoted as $\hat z:=\left\{\left(z_{0}^n,\ldots,z_{M-1}^n\right)\right\}_{n=1}^N$. This will serve as the regression data in the following step.
\paragraph{Regress the option value against an IR asset}
In this phase we compose replication portfolios $\Pi_0,\ldots,\Pi_{M-1}$, by fitting $M$ regression functions $G_0,\ldots,G_{M-1}$. We consider functions of the form $G_m:\mathbb{R}^d\to\mathbb{R}$, which assign values in $\mathbb{R}$ to each realization of the selected asset $z_m$. Fitting is done recursively, starting at $T_{M-1}$, moving backwards in time, until the first exercise opportunity $T_0$. Approximations of the Bermudan swaption value at each monitor date serve as dependent variable. At the final monitor date, the value of the contract (given it has not been exercised) is known to be 
\begin{equation*} \begin{aligned}
    V\left(T_{M-1};x_{T_{M-1}}^n\right)=\max\left\{h_{M-1}\left(x_{T_{M-1}}^n\right),\;0\right\}, \qquad n=1,\ldots,N
\end{aligned} \end{equation*} 
Now assume that for some monitor date $T_m$ we have an approximation of the contract value $\tilde V\left(T_m;x_{T_{m}}^n\right)\approx V\left(T_{m};x_{T_{m}}^n\right)$.
Let $\xi_m\in\mathbb{R}^p$ for some $p\in\mathbb{N}$ denote the vector of the unknown regression parameters. The objective is to determine $\xi_m$ such that
\begin{equation*} \begin{aligned}
    G_m\left(z_m(T_m)|\mathbf{x}_{T_m}\right)\approx V\left(T_{m}\right)
\end{aligned} \end{equation*}
with the smallest possible error. This is done by formulating and solving a related optimization problem. In this case we choose to minimize the expected square error, given by
\begin{equation} \begin{aligned}
    \mathbb{E^Q}\left[\left|G_m\left(z_m(T_m)|\mathbf{x}_{T_m}\right)- V\left(T_{m}\right)\right|^2\right]\label{eqn:loss1}
\end{aligned} \end{equation}
There is no exact analytical expression available for the expectation of \ref{eqn:loss1}. However, it can be approximated using the sampled regression data, giving rise to an empirical loss function $L$ given by
\begin{equation} \begin{aligned}
    L(\xi_m|\hat z_m, \hat x_m)=\sum_{n=1}^N\tilde w\left(x_{T_m}^n\right)\cdot \left(G_m\left(z_m^n\right)- \tilde V\left(T_m;x_{T_{m}}^n\right)\right)^2 \label{eqn: loss2}
\end{aligned} \end{equation}
The weight-function $\tilde w$ is implied by the risk-factor sampling method of choice. If $\hat x$ is sampled through Monte Carlo, the weights are constant $\tilde w(x)=\frac{1}{N}$. If $\hat x$ is sampled as a standard quadrature, the function $\tilde w$ denotes the related quadrature weights, see \cite{xiu2010numerical}. The parameters $\xi_m$ are then the result of the fitting procedure, such that
\begin{equation*} \begin{aligned}
    \xi_m \approx \underset{\xi \in \mathbb{R}^p}{\text{argmin}}\;\,L(\xi|\hat z_m, \hat x_m)
\end{aligned} \end{equation*}
If the regression model is chosen accordingly, $G_m(z_m)$ represents the pay-off at $T_m$ of a derivative  portfolio $\Pi_m$ written on the selected asset $z_m$. Details on suggested functional forms of $G_m$, asset selection for $z_m$ and fitting procedures are subject of \cref{sec: architecture}.

\paragraph{Compute the continuation value}
Once the regression is completed, the last step is to compute the continuation value and subsequently the option value at the monitor date preceding to $T_m$. For each scenario $n=1,\ldots,N$ we approximate the continuation value as
\begin{equation} \begin{aligned}
    \tilde C_{m-1}(T_{m-1}) &= B(T_{m-1})\mathbb{E^Q}\left[\frac{\tilde V\left(T_m\right)}{B\left(T_{m}\right)}\bigg|\mathcal{F}_{T_{m-1}}\right]\\
    &\approx B(T_{m-1})\mathbb{E^Q}\left[\frac{G_m\left(z_{m}(T_m)\right)}{B\left(T_{m}\right)}\bigg|\mathcal{F}_{T_{m-1}}\right] \label{eqn: continuation1}
\end{aligned} \end{equation}
As $G_m$ is chosen to represent the pay-off of a derivative portfolio $\Pi_m$ written on $z_m$, we argue that computing $C_{m-1}$ is in fact equivalent to the risk-neutral pricing of $\Pi_m$. In other words we have
\begin{equation*} \begin{aligned}
    \tilde C_{m-1}(T_{m-1})=B(T_{m-1})\mathbb{E^Q}\left[\frac{\Pi_m\left(T_{m}\right)}{B\left(T_{m}\right)}\bigg|\mathcal{F}_{T_{m-1}}\right]:=\Pi_m(T_{m-1})
\end{aligned} \end{equation*}
In \cref{sec: architecture} we treat examples for which $\Pi_m$ can be computed in closed-form.

Finally the option value at the preceding monitor date $T_{m-1}$ is given by
\begin{equation*} \begin{aligned}
    \tilde V\left(T_m;x_{T_m}^n\right)=\max\left\{\tilde C_{m-1}(T_{m-1}),\;h_{m-1}\left(x_{T_{m-1}}^n\right)\right\},\qquad n=1,\ldots,N
\end{aligned} \end{equation*}
The steps are repeated recursively until we have a representation $G_0$ of the option value at the first monitor date. An estimator of the time-zero option value is given by
\begin{equation*} \begin{aligned}
     \tilde V(0)=\mathbb{E^Q}\left[\frac{G_0\left(z_{0}(T_0)\right)}{B\left(T_{0}\right)}\bigg|\mathcal{F}_0\right]
\end{aligned} \end{equation*}
We refer to this approximation as the \textit{direct estimator}.
\begin{algorithm}
\caption{The algorithm for a Bermudan swaption}
\label{alg: rlnn}
\begin{algorithmic}
\STATE{Generate $N$ risk-factor scenarios for $\mathbf{x}_{T_m}$ for $m=0,\ldots,M$}
\STATE{Compute $N$ corresponding asset scenarios for $z_m$ for $m=0,\ldots,M$}
\STATE{$\Tilde V\left(T_{M-1};x_{T_{M-1}}^n\right)\leftarrow \max\left\{h_{M-1}\left(x_{T_{M-1}}^n\right),\,0\right\}$ for $n=1,\ldots,N$}
\STATE{Initialize $G_{M-1}$ parameters $\xi_{M-1}$ from independent uniform distributions}
\FOR{$m=M-1,\ldots,1$}{
\STATE{$\xi_m \leftarrow \underset{\xi \in \mathbb{R}^p}{\text{argmin}}\;\,L(\xi|\hat z_m, \hat x_m)$ minimizing the MSE}
\FOR{$n=1,\ldots,N$}{
\STATE{$\tilde C_{m-1}(T_{m-1}) \leftarrow  B(T_{m-1})\mathbb{E^Q}\left[\frac{G_m\left(z_m(T_m)\right)}{B\left(T_{m}\right)}\Big|\mathcal F_{T_{m-1}}\right]$}
\STATE{$\Tilde V(T_{m-1};x_{T_{m-1}}^n)\leftarrow \max\left\{\tilde C_{m-1}(T_{m-1}),\,h_{m-1}\left(x_{T_{m-1}}^n\right)\right\}$}}
\ENDFOR
\STATE{$\xi_{m-1}\leftarrow \xi_{m}$ initialize weights of $G_{m-1}$}}
\ENDFOR
\STATE{$\xi_0 \leftarrow \underset{\xi \in \mathbb{R}^p}{\text{argmin}}\;\,L(\xi|\hat z_0, \hat x_0)$ minimizing the MSE}
\RETURN{$\mathbb{E^Q}\left[\frac{G_0\left(z_0(T_0)\right)}{B\left(T_{0}\right)}\Big | \mathcal{F}_0\right]$}
\end{algorithmic}
\end{algorithm}

\subsection{A neural network approach to $\mathbf{G_m}$\label{sec: architecture}}
In this section we propose to represent the regression functions $G_m$ as shallow, artificial neural networks. We discuss some approaches to the neural network design and its interpretation as a static hedge portfolio. In particular we specifically propose to utilize networks with a single hidden layer, employing the ReLU activation function to facilitate this interpretation. The choices that are presented here are adapted to a framework of Gaussian risk-factors, such as presented in \cref{sec: Math formulation}. The method however lends itself to be generalized to a broader class of models, by considering an appropriate adjustment to the input or structure of each neural network. Our motivation for a neural network representation of $G_m$ is twofold. First, it allows us to utilize the optimization algorithms that have been developed in this field. Second, by the virtue of the universal approximation theorem \cite{hornik1989multilayer} we can reach any level of accuracy for the replication within a compact domain of the risk-factor (see \cref{sec: error analysis}).

\subsubsection{The 1-factor case}
\begin{figure}
\centering
\includegraphics[width=0.4\linewidth]{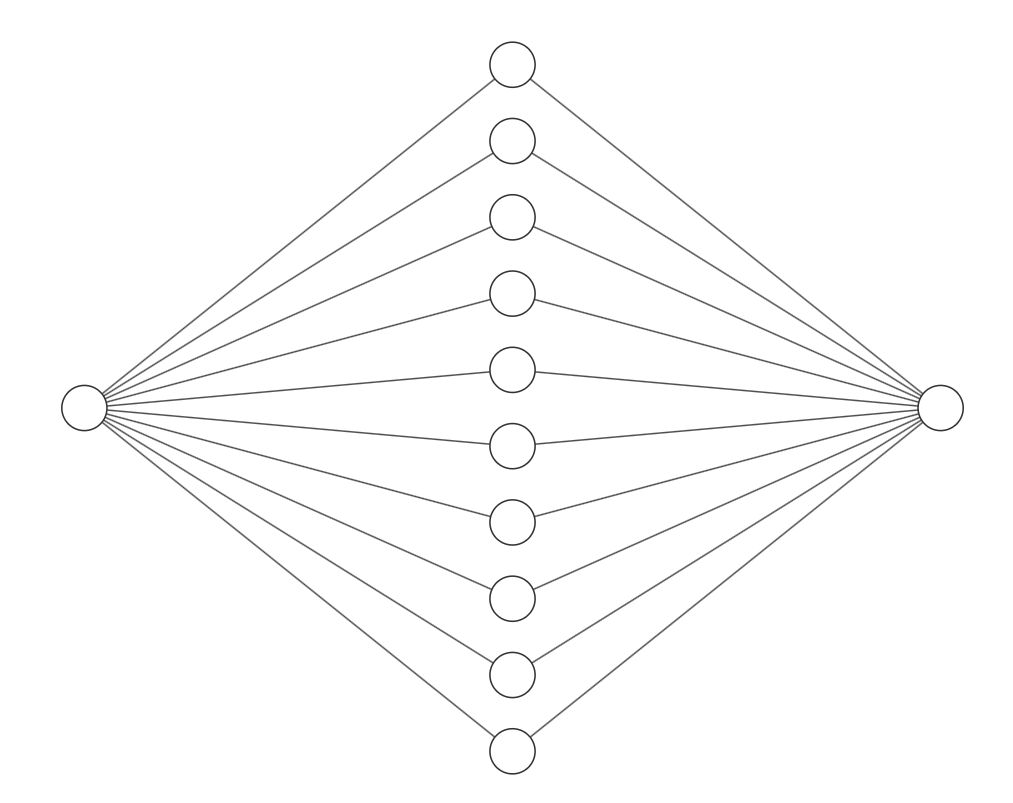}
\caption{Suggested neural network design for $Dim\left(\mathbf{x}_{t}\right)=1$}
\label{fig: single-dim NN}
\end{figure}

First we discuss the case $d=1$. Let $m\in\{0,\ldots,M-1\}$. As a regression function we consider a fully-connected, feed-forward neural network with one hidden layer, denoted as $G_m:\mathbb{R}\to\mathbb{R}$. The design with only a single hidden layer is chosen deliberately to facilitate the network's interpretation. As an input to the network (the asset $z_m$) we select a zero-coupon bond, which pays one unit of currency at $T_M$. Our motivation for this choice is that both discount bonds and European options on discount bonds admit closed-form pricing formulas. The network is given the following structure.
\begin{itemize}
    \item The first layer consists of a single node. As a property of the one-factor model, all the spot-rates at some future instance $\left\{R(T_m,T^\prime),T^\prime\in[T_m,T]\right\}$ are perfectly correlated through the single variable $x_{T_m}$. A 1D input layer taking $P\left(T_m,T_M\right)$ as input is hence rich enough to identify the option value $\tilde V(T_m)$. The hidden layer has $q\in\mathbb{N}$ hidden nodes. The affine transformation acting between the first two layers is denoted $A_1:\mathbb{R}\to\mathbb{R}^q$ and is of the form
    \begin{equation*} \begin{aligned}
        A_1: x \mapsto \mathbf{w}_1 x + \mathbf{b}, \qquad \mathbf{w}_1\in\mathbb{R}^{q\times 1}, \mathbf{b}\in\mathbb{R}^q
    \end{aligned} \end{equation*}
    As an activation function $\varphi:\mathbb{R}^q\to\mathbb{R}^q$ acting on the hidden layer we take the ReLU-function, given by
    \begin{equation*} \begin{aligned}
        \varphi:\left(x_1,\ldots,x_q\right)\mapsto\left(\max\{x_1,\,0\},\ldots,\max\{x_q,\,0\}\right)
    \end{aligned} \end{equation*}
    The ReLU activation is chosen deliberately as it corresponds to the pay-off function of a European option.
    \item The output of the network aims to estimate contract value $\tilde V_m\in\mathbb{R}$ and therefore contains only a single node. We consider a linear transformation acting between the second and last layer $A_2:\mathbb{R}^{q}\to\mathbb{R}$, given by
    \begin{equation*} \begin{aligned}
        A_2:x\mapsto \mathbf{w}_2 x, \qquad \mathbf{w}_2\in\mathbb{R}^{1\times q}
    \end{aligned} \end{equation*}
    On top of that we apply the linear activation, which comes down to an identity function, mapping $x$ to itself.
\end{itemize}
Combined together, the network is specified to satisfy
\begin{equation*} \begin{aligned}
    G_m(\cdot):=A_2\circ\varphi\circ A_1
\end{aligned} \end{equation*}
and the trainable parameters can be presented by the list
\begin{equation*} \begin{aligned}
    \xi_m = \left\{w_{1,1},b_{1,1},\ldots,w_{1,q},b_{1,q}\right\}\cup\left\{w_{2,1},\ldots,w_{2,q}\right\}
\end{aligned} \end{equation*}
The architecture is graphically visualized in \cref{fig: single-dim NN}.
\paragraph{Interpretation of the neural network}
Now that we specified the structure of the neural network, we will discuss how each function $G_m$ can be interpreted as a portfolio $\Pi_m$. In the one-dimensional case, $G_m$ can be expressed as follows
\begin{equation*} \begin{aligned}
    G_m(z_m):=\sum_{j=1}^q w_{2,j}\max\left\{w_{1,j}z_m+b_{j},\;0\right\}
\end{aligned} \end{equation*}
We can regard this as the pay-off of a derivative portfolio $\Pi_m$ written on the asset $z_m$ (in our case the zero-coupon bond $P(T_m, T_M)$). The portfolio contains $q$ derivatives that each have a terminal value equal to $ w_{2,j}\max\left\{w_{1,j}z_m+b_{j},\;0\right\}$. In total we can recognize four types of products, which depend on the signs of $w_{1,j}$ and $b_j$. Keep in mind that $z_m:=P(T_m,T_M)$ is positive by default.
\begin{enumerate}
    \item If $w_{1,j}>0$ and $b_j>0$, we have
    \begin{equation*} \begin{aligned}
        w_{2,j}\max\left\{w_{1,j}z_m+b_{j},\;0\right\}=w_{2,j}w_{1,j}z_m+w_{2,j}b_{j}
    \end{aligned} \end{equation*}
    which is the pay-off of a forward contract on $w_{2,j}w_{1,j}$ units in $z_m$ and $w_{2,j}b_{j}$ units of currency.
    \item If $w_{1,j}>0$ and $b_j<0$, we have
    \begin{equation*} \begin{aligned}
        w_{2,j}\max\left\{w_{1,j}z_m+b_{j},\;0\right\}=w_{2,j}w_{1,j}\max\left\{z_m-\frac{-b_{j}}{w_{1,j}},\;0\right\}
    \end{aligned} \end{equation*}
    which is the pay-off corresponding to $w_{2,j}w_{1,j}$ units of a European call option written on $z_m$, with strike price $\frac{-b_{j}}{w_{1,j}}$.
    \item If $w_{1,j}<0$ and $b_j>0$, we have
    \begin{equation*} \begin{aligned}
        w_{2,j}\max\left\{w_{1,j}z_m+b_{j},\;0\right\}=-w_{2,j}w_{1,j}\max\left\{\frac{b_{j}}{-w_{1,j}}-z_m,\;0\right\}
    \end{aligned} \end{equation*}
    which is the pay-off corresponding to $-w_{2,j}w_{1,j}$ units of a European put option written on $z_m$, with strike price $\frac{b_{j}}{-w_{1,j}}$.
    \item If $w_{1,j}<0$ and $b_j<0$, we have
    \begin{equation*} \begin{aligned}
        w_{2,j}\max\left\{w_{1,j}z_m+b_{j},\;0\right\}=0
    \end{aligned} \end{equation*}
    which clearly represents a worthless contract.
\end{enumerate}
The sign of the coefficient $w_{2,j}$ indicates if one has a short or long position of the product in the portfolio. Hence, under the assumption of a frictionless economy, the absence of arbitrage and the Markov property for $z_m$, the portfolio $\Pi_m$ replicates the original Bermudan contract over the period $(T_{m-1},T_m]$. As the portfolio composition is constant between two consecutive monitor dates, the method described here can be interpreted as a semi-static hedging strategy.


\subsubsection{The multi-factor case\label{sec: d-dim NN design}}
\begin{figure}
\centering
\begin{subfigure}{.5\textwidth}
  \centering
  \includegraphics[width=0.7\linewidth]{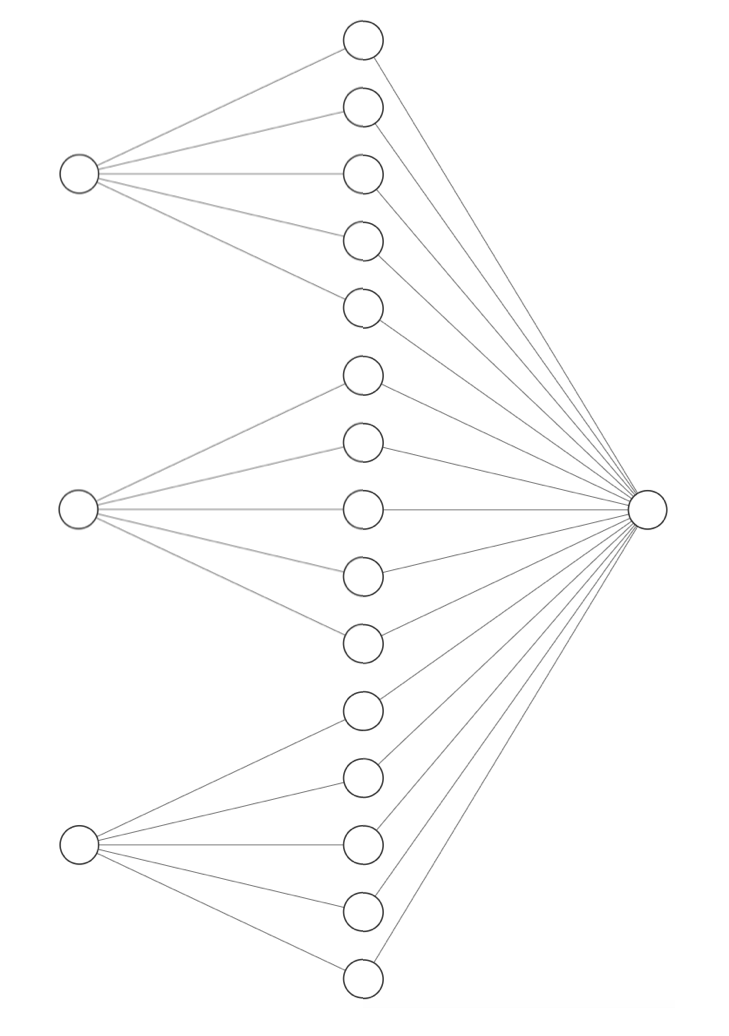}
  \caption{Locally connected neural network}
  \label{fig:Locally conn NN}
\end{subfigure}%
\begin{subfigure}{.5\textwidth}
  \centering
  \includegraphics[width=0.65\linewidth]{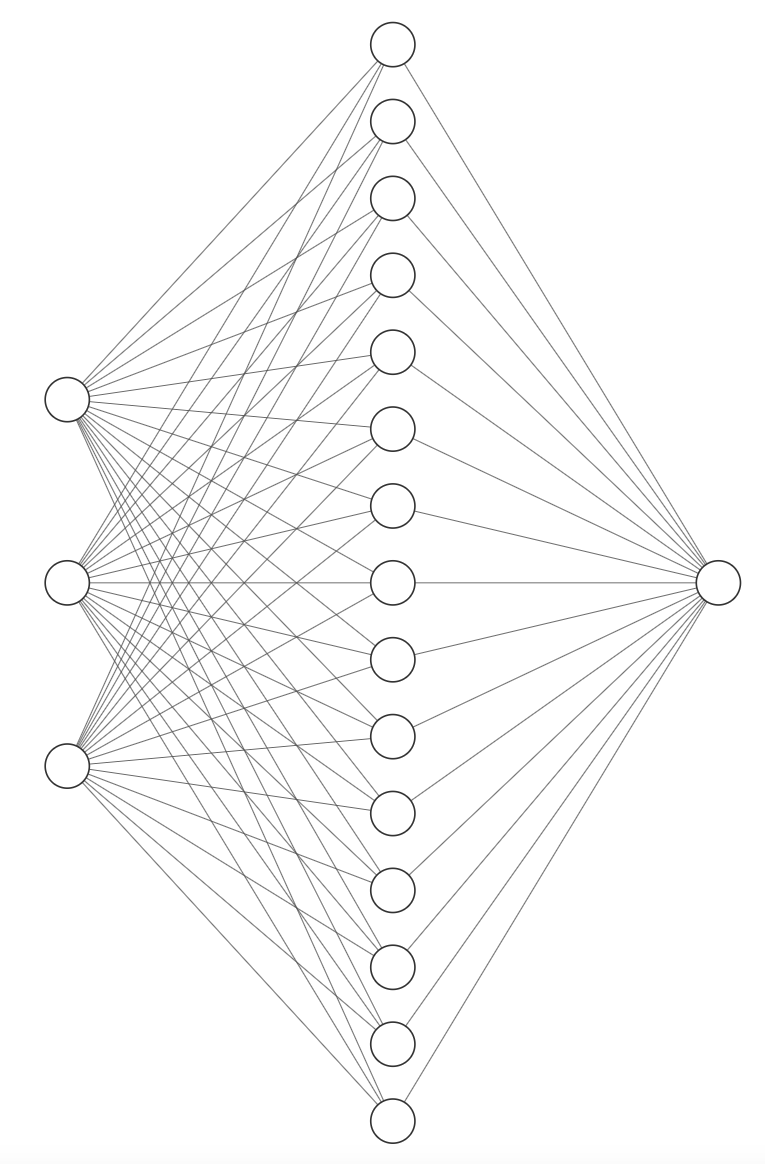}
  \caption{Fully connected neural network}
  \label{fig:Fully conn NN}
\end{subfigure}
\caption{Suggested neural network designs for $Dim\left(\mathbf{x}_{t}\right)\geq2$}
\label{fig: multi-dim NNs}
\end{figure}

In the case $d\geq 2$, the modelled spot-rates are no longer perfectly correlated. For that reason, a single zero-coupon bond does not suffice to identify the option value $\tilde V(T_m)$ for any $m\in\{0,\ldots,M-1\}$. Instead, we propose that a basket of $d$ zero-coupon bonds all maturing at different dates $T_m+\delta_1,\ldots,T_m+\delta_n$ is required as input to the regression. Simply said, if the risk-factor space is $d$-dimensional, it can only be parametrized by an at least $d$-dimensional asset vector.

To see why the above statement is true, simply consider $n$ bonds $P(T_m,T_m+\delta_1),\ldots,$ $P(T_m,T_m+\delta_n)$ and note that the following relation holds
\begin{equation*} \begin{aligned}
    &\begin{pmatrix}P(T_m,T_m+\delta_1)\\\vdots\\P(T_m,T_m+\delta_n)\end{pmatrix}=\begin{pmatrix}\exp\{A(T_m,T_m+\delta_1)-\sum_{j=1}^dB_j(T_m,T_m+\delta_1)x_j(T_m)\}\\\vdots\\\exp\{A(T_m,T_m+\delta_n)-\sum_{j=1}^dB_j(T_m,T_m+\delta_n)x_j(T_m)\}\end{pmatrix}\\
    \implies &  \begin{pmatrix}B_1(T_m,T_{m}+\delta_1)&\hdots&B_d(T_m,T_{m}+\delta_1)\\\vdots&\ddots&\vdots\\B_1(T_m,T_{m}+\delta_n)&\hdots&B_d(T_m,T_{m}+\delta_n)\end{pmatrix}\begin{pmatrix}x_1(T_m)\\\vdots\\x_d(T_m)\end{pmatrix}\\
    &\qquad\qquad\qquad=\begin{pmatrix}A(T_m,T_m+\delta_1)-\log P(T_m,T_m+\delta_1)\\\vdots\\A(T_m,T_m+\delta_d)-\log P(T_m,T_m+\delta_n)\end{pmatrix}\\
    \implies &  \mathbf{B}(T_m)\mathbf{x}_{T_m}=\mathbf{\alpha}
\end{aligned} \end{equation*}
Since we have that $rank(\mathbf{B}(T_m))= \min\{n,d\}$ it follows that if $n<d$ the image of $\mathbf{B}$ does not span the whole risk-factor space, whereas if $n>d$ the image of $\mathbf{B}$ is still equal to the case $n=d$.

Concluding on the argument above, it would be an obvious choice to take a $d-$dimensional vector of bonds as the input and generalize the architecture of $G_m$ by increasing the input-dimension (i.e. the number of nodes in the first layer) from 1 to $d$. A problem that would occur however, is that $\Pi_m$ then represents a derivatives portfolio written on a basket of bonds, by which the tractability of pricing $\Pi_m$ would be lost. Therefore we suggest two alternatives to the design of $G_m$, intended to preserve the analytical valuation potential of $\Pi_m$.

The basic specifications of the neural network will remain similar to the 1-factor case. We consider a feed-forward neural network with one hidden layer of the form $G_m:\mathbb{R}^d\to\mathbb{R}$.
\begin{itemize}
    \item The first layer consists of $d$ nodes and the hidden layer has $q\in\mathbb{N}$ hidden nodes. The affine transformation and activation acting between the first two layers are denoted $A_1:\mathbb{R}^d\to\mathbb{R}^q$ and $\varphi:\mathbb{R}^q\to\mathbb{R}^q$ respectively given by
    \begin{equation*} \begin{aligned}
        &A_1: x \mapsto \mathbf{w}_1 x + \mathbf{b}, \qquad \mathbf{w}_1\in\mathbb{R}^{q\times d}, \mathbf{b}\in\mathbb{R}^q\\
        &\varphi:\left(x_1,\ldots,x_q\right)\mapsto\left(\max\{x_1,\,0\},\ldots,\max\{x_q,\,0\}\right)
    \end{aligned} \end{equation*}
    \item The output contains a single node. A linear transformation acts between the second and last layer $A_2:\mathbb{R}^{q}\to\mathbb{R}$, together with the linear activation, given by
    \begin{equation*} \begin{aligned}
        A_2:x\mapsto \mathbf{w}_2 x, \qquad \mathbf{w}_2\in\mathbb{R}^{1\times q}
    \end{aligned} \end{equation*}
    \item The network is given by $G_m(\cdot):=A_2\circ\varphi\circ A_1$
\end{itemize}

\paragraph{Suggestion 1: A locally connected neural network}
The outcome of each node in the hidden layer represents the terminal value of a derivative written on the asset $\mathbf{z}_m$, which together compose the portfolio $\Pi_m$. In the $d-$dimensional case the outcome of the $j^{th}$ node $\nu_j$ can be expressed as
\begin{equation*} \begin{aligned}
    \nu_j(\mathbf{z})=\max\left\{\sum_{k=1}^d w_{jk}z_k+b_j,\;0\right\}
\end{aligned} \end{equation*}
which corresponds to the pay-off of an arithmetic basket option with weights $w_{j1},\ldots w_{jd}$ and strike price $b_j$. Such an exotic option is difficult to price. To overcome this issue we constrain the matrix $\mathbf{w}_1$ to only admit a single non-zero value in each row. Let the number of hidden nodes be a multiple of the input dimension, i.e. $q=n\cdot d$ for some $n\in\mathbb{N}$. The matrix $\mathbf{w}_1$ is set to be of the form
\begin{equation*} \begin{aligned}
    \mathbf{w}_1=\begin{pmatrix}w_{1,1}&0&0&\cdots&0&0\\
    \vdots&\vdots&\vdots&&\vdots&\vdots\\
    w_{1,n}&0&0&\cdots&0&0\\
    0&w_{2,n+1}&0&\cdots&0&0\\
    \vdots&\vdots&\vdots&&\vdots&\vdots\\
    0&w_{2,2n}&0&\cdots&0&0\\
    \vdots&\vdots&\vdots&&\vdots&\vdots\\
    0&0&0&\cdots&0&w_{d,d\cdot n}\end{pmatrix}
\end{aligned} \end{equation*}
The architecture is graphically depicted in \cref{fig:Locally conn NN}. As a result, the outcome of each node $\nu_j$ again represents a European option or forward written on a single bond, which can be priced in closed-form (see \cref{sec: eval local NN}).

We can recognize two drawbacks to this approach. First, the number of trainable parameters for a fixed number of hidden nodes is much lower compared to the fully connected case. This can simply be overcome by increasing $q$. Second, as the network is not fully connected, the universal approximation theorem no longer applies to $G_m$. Therefore we have no guarantee the approximation errors can be reduced to any desirable level. Our numerical experiments however indicate that the approximation accuracy of this design is not inferior to that of a fully connected counterpart of the same dimensions; see \cref{sec: numerical examples}.

\paragraph{Suggestion 2: A fully connected neural network}
Our second approach does not entail altering the structure or weights of the network, but suggests to take a different input. We hence consider a fully connected, feed-forward neural network with one hidden layer of the form $G_m:\mathbb{R}^d\to\mathbb{R}$. The architecture is graphically depicted in \cref{fig:Fully conn NN}. However as an input, we use the log of $n$ bonds, i.e.
\begin{equation*} \begin{aligned}
    \mathbf{z}_m:=\left(\log P(T_m,T_m+\delta_1),\ldots,\log P(T_m,T_m+\delta_n)\right)^\top
\end{aligned} \end{equation*}
As a result, each node $\nu_j$ can be compared to the payoff of a geometric basket option written on $n$ assets $\mathbf{z}_m$ equal to log of $P(t,T_m+\delta_j)$. Under the assumption that the dynamics of the risk-factor $\mathbf{x}_t$ are Gaussian, these options can be priced explicitly as we will show in \cref{sec: eval fully NN}.

An advantage of this approach is that it employs a fully-connected network which by the virtue of the universal approximation theorem \cite{hornik1989multilayer} can yield any desired level of accuracy. A drawback is that the financial interpretation of the network as a replicating portfolio is not as strong as in suggestion 1, due to the required $\log$ in the payoff.

\subsection{Training of the neural networks}
In this section we specify some of the main considerations related to the fitting procedure of the algorithm. The method requires the training of $M$ shallow feed-forward networks as specified in \cref{sec: architecture}, which we denote $G_0,\dots,G_{M-1}$. Our numerical experiments indicated that normalization of the training-set strongly improved the networks' fitting accuracy. Details for pre-processing the regression data are treated in \cref{sec: preprocessing the data}.

\paragraph{Optimization}
The training of each network is done in an iterative process, starting with $G_{M-1}$ working backwards until $G_0$. The effectiveness of the process depends on several standard choices related to neural network optimization, of which some are listed below.
\begin{itemize}
    \item As optimizer we apply AdaMax \cite{kingma2014adam}, a variation to the commonly used Adam algorithm. This is a stochastic, first-order, gradient-based optimizer that updates weights inversely proportional to the $L_\infty$-norm of their current and past gradients; whereas Adam is based on the $L_2$-norm. Our experiments indicate that AdaMax slightly outperforms comparable algorithms in the scope of our objectives.
    \item The batch size, i.e. the number of training points used per weight update, is set to a standard 32. The learning rate, which scales the step size of each update, is kept in the range 0.0001-0.0005.
    \item For the initial network, $G_{M-1}$, we use random initialization of the parameters. If the considered contract is a payer Bermudan swaption, we initialize the (non-zero) entries of $\mathbf{w}_1$ i.i.d. unif$(0,1)$ and the biases $\mathbf{b}$ i.i.d. unif$(-1,0)$. In the case of a receiver contract, it's the other way around. The weights $\mathbf{w}_2$ are initialized i.i.d. unif$(-1,1)$.
    \item For the subsequent networks, $G_{M-2},\ldots,G_0$, each network $G_m$ is initialized with the final set of weights of the previous network $G_{m+1}$.
    \item As a training set for the optimizer we use a collection of 20,000 data-points.
\end{itemize}
Some specific choices for the hyperparameters are motivated by a convergence-analysis presented in \cref{sec: hyper-parameters}.

\section{Lower and upper bound estimates\label{bounds theory}}
The algorithm described in \cref{sec: algorithm}, gives rise to a direct estimator of the true option price $V$. The accuracy of this estimator depends on the approximation performance of the neural networks at each monitor date. Should each regression yield a perfect fit, then the estimation error would automatically be zero. In practice however, the loss function, defined in (\ref{eqn: loss2}), never fully converges to zero. As the networks are trained to closed-form exercise and continuation values, error measures such as MSE and MAE can be easily obtained. Especially the mean absolute errors provide a strong indication of the error bounds on the direct estimator (see \cref{sec: error analysis}).

Although convergence errors put solid bounds on the accuracy of the estimator, they are typically quite loose. Therefore they give rise to non-tight confidence bounds. To overcome this issue we introduce a numerical approximation to a tight lower and upper bound to the true price, in the same spirit as \cite{lokeshwar2019neural}. These should provide a better indication of the quality of the estimate.

\subsection{The lower bound\label{sec: lower bound}}
We compute a lower bound approximation by considering the non-optimal exercise strategy $\tilde\tau$ implied by the continuation values estimates introduced in \cref{sec: algorithm}. We define $\tilde\tau$ as
\begin{equation} \begin{aligned}
    \tilde\tau(\omega)=\min\left\{T_m\in\mathcal{T}_f\big|\tilde C_m\left(T_m\right)\leq h_m\left(\mathbf{x}_{T_m}\right)\right\}\label{eq: stopping time}
\end{aligned} \end{equation}
where $\tilde C_m$ refers to the approximated continuation value given in \ref{eqn: continuation1}. A strict lower bound is now given by
\begin{equation} \begin{aligned}
    L(0)= \mathbb{E^Q}\left[\frac{h_{\tilde\tau}(\mathbf{x}_{\tilde\tau})}{B(\tilde\tau)}\bigg|\mathcal{F}_{0}\right]=P(0,T_M)\mathbb{E}^{T_M}\left[\frac{h_{\tilde\tau}(\mathbf{x}_{\tilde\tau})}{P(\tilde\tau,T_M)}\bigg|\mathcal{F}_{0}\right] \label{eqn: lower bound}
\end{aligned} \end{equation}
where $h_{\tilde\tau}$ corresponds to the definition given in \ref{eqn: stopping time}. The term on the right is obtained by changing the measure from $\mathbb{Q}$ to the $T_M-$forward measure $\mathbb{Q}^{T_M}$\cite{geman1995changes}. Under the risk-neutral measure the lower bound can be estimated by simulating a fresh set of scenarios of the risk-factor $\hat x:=\left\{\left(x_{t_1}^n,x_{t_2}^n,\ldots,x_{T_M}^n\right)\Big|n=1,\ldots,N\right\}$. Let $\hat r_j^n$ denote the corresponding realizations of the short-rate. An approximation of the discounting term is obtained as
\begin{equation} \begin{aligned}
    B^n(t_j):= \exp\left\{\int_0^{t_j}r^n(u)du\right\}\approx \exp\left\{\sum_{i=1}^j\frac{1}{2}(\hat r_{i-1}^n+\hat r_i^n)\cdot\left(t_i-t_{i-1}\right)\right\}:=\tilde B^n(t_j) \label{eqn: discount factor}
\end{aligned} \end{equation}
The estimate of $L$ is then given by
\begin{equation*} \begin{aligned}
    \tilde L(0)=\frac{1}{N}\sum_{n=1}^N\frac{h_{\tilde\tau^n}\left(x^n_{\tilde\tau^n}\right)}{B^n(\tilde\tau^n)}
\end{aligned} \end{equation*}
For accurate approximation of the discount term, a coarse discretization of the simulation grid is required (see \cite{brigo2007interest}), which is computationally demanding. An alternative is to simulate the fresh set of scenarios under the $T_M$-forward measure instead. Denote by $P^n(t,T_M)$ the zero-coupon bond realization corresponding to $x^n_t$ and compute the approximation as
\begin{equation*} \begin{aligned}
    \tilde L(0) = \frac{P(0,T_M)}{N}\sum_{n=1}^N\frac{h_{\tilde\tau^n}\left(x^n_{\tilde\tau^n}\right)}{P^n(\tilde\tau^n,T_M)}
\end{aligned} \end{equation*}
As a consequence there is no need to simulate the num\'eraire as there is an analytical expression available for $P^n(\tilde\tau^n,T_M)$.

\subsection{The upper bound\label{sec: Upper bound}}
We compute an upper bound by considering a dual formulation of the price expression \ref{eqn:Berm1} as proposed in \cite{haugh2004pricing} and \cite{rogers2002monte}. Let $\mathcal{M}$ denote the set of all Martingales $M_t$ adapted to $\mathbb{F}$ such that $\sup_{t\in[0,T]}\left|M_t\right|<\infty$. An upper bound $U(0)$ to the true price $V(0)$ is obtained by observing that the following inequality holds (see \cite{haugh2004pricing}),
\begin{equation} \begin{aligned}
    V(0)\leq M_0+\mathbb{E^Q}\left[\max_{T_m\in\mathcal{T}_f}\left\{\frac{h_m(\mathbf{x}_{T_m})}{B(T_m)}-M_{T_m}\right\}\bigg|\mathcal{F}_0\right]:=U(0)\label{eqn: martingale ineq}
\end{aligned} \end{equation}
for any $M_t\in\mathcal{M}$. To find a suitable Martingale that yields a tight bound, we consider the Doob-Meyer decomposition of the true discounted option price process $\frac{V(t)}{B(t)}$. As the price process is a supermartingale, we can write \begin{equation*} \begin{aligned}
    \frac{V(t)}{B(t)}:=Y_t+Z_t
\end{aligned} \end{equation*}
where $Y_t$ denotes a Martingale and $Z_t$ a predictable, strictly decreasing process such that $Z_0=0$. Note that equation \ref{eqn: martingale ineq} attains an equality if we set $M_t=Y_t$, i.e. the Martingale part of the option price process. The bound will hence be tight if we consider a Martingale $M_t$ that is close to the unknown $Y_t$. Let $G_m(\cdot)$ denote the neural networks induced by the algorithm. In the spirit of \cite{andersen2004primal} and \cite{lokeshwar2019neural} we construct a Martingale on the discrete time grid $\left\{0,T_0,\ldots,T_{M-1}\right\}$ as follows.
\begin{equation} \begin{aligned}
    & M_0=\mathbb{E^Q}\left[\frac{G_0(z_0(T_0))}{B(T_0)}\bigg|\mathcal{F}_{0}\right],\; M_{T_0}=\frac{G_0(z_0(T_0))}{B(T_0)}\\
    & M_{T_{m}}=M_{T_{m-1}}+\frac{G_m(z_m(T_m))}{B(T_m)}-\mathbb{E^Q}\left[\frac{G_m(z_m(T_m))}{B(T_m)}\bigg|\mathcal{F}_{T_{m-1}}\right],\quad m=1,\ldots,M-1 \label{eqn:martingale1}
\end{aligned} \end{equation}
Clearly, the process $\left\{M_{T_m}\right\}_{m=0}^{M-1}$ yields a discrete Martingale as
\begin{equation*} \begin{aligned}
    \mathbb{E^Q}\left[M_{T_m}\big|\mathcal{F}_{T_{m-1}}\right]&=\mathbb{E^Q}\left[M_{T_{m-1}}+\frac{G_m(z_m(T_m))}{B(T_m)}-\mathbb{E^Q}\left[\frac{G_m(z_m(T_m))}{B(T_m)}\bigg|\mathcal{F}_{T_{m-1}}\right]\bigg|\mathcal{F}_{T_{m-1}}\right]\\
    &=\mathbb{E^Q}\left[M_{T_{m-1}}\big|\mathcal{F}_{T_{m-1}}\right]+\mathbb{E^Q}\left[\frac{G_m(z_m(T_m))}{B(T_m)}-\frac{G_m(z_m(T_m))}{B(T_m)}\bigg|\mathcal{F}_{T_{m-1}}\right]\\
    &=M_{T_{m-1}}
\end{aligned} \end{equation*}
Furthermore, the process $M_t$ as defined above will coincide with $Y_t$ if the approximation errors in $G_m(\cdot)$ equal zero, hence yielding an equality in \ref{eqn: martingale ineq}. Note that the recursive relation in \ref{eqn:martingale1} can be rewritten as
\begin{equation} \begin{aligned}
    M_{T_m}=\frac{G_0(z_0(T_0))}{B(T_0)}+\sum_{j=1}^{m} \left(\frac{G_j(z_j(T_j))}{B(T_j)}-\mathbb{E^Q}\left[\frac{G_j(z_j(T_j))}{B(T_j)}\bigg|\mathcal{F}_{T_{j-1}}\right]\right)\label{eqn:martingale2}
\end{aligned} \end{equation}
We can now estimate the upper bound by again simulating a set of scenarios of the risk-factor $\left\{\left(x_{t_1}^n,x_{t_2}^n,\ldots,x_{T_M}^n\right)\Big|n=1,\ldots,N\right\}$ and approximate $U(0)$ under the risk-neutral measure as
\begin{equation*} \begin{aligned}
    \tilde U(0)=M_0+\frac{1}{N}\sum_{n=1}^N\max_{T_m\in\mathcal{T}_f}\left\{\frac{h_{T_m}\left(x^n_{T_m}\right)}{ B^n(T_m)}-M^n_{T_m}\right\}
\end{aligned} \end{equation*}
where the discounting term is estimated as in \ref{eqn: discount factor}. In a similar fashion, the upper bound can be approximated under the $T_M-$forward measure. In that case the risk-factor should be simulated under $\mathbb{Q}^{T_M}$ and the num\'eraire $B(t)$ be replaced by $P(t,T_M)$. By doing so we again avoid the need to approximate the num\'eraire on a coarse simulation grid.

Note that by the deliberate choice of $G_m(\cdot)$, all the conditional expectations appearing in \ref{eqn:martingale2} can be computed in closed-form (see \cref{sec: eval cond exp}). Hence, there is no need to resort to nested simulations, in contrast to for example \cite{andersen2004primal} and \cite{becker2020pricing}. Especially if simulations are performed under the $T_M-$forward measure, both lower and upper bound estimations can be obtained at minimal additional computational cost.

\section{Error analysis\label{sec: error analysis}}
In this section we analyze the errors of the semi-static hedge, the direct estimator, the lower bound estimator and the upper bound estimator, which are induced by the imprecision of the regression functions $G_0,\ldots,G_{M-1}$. We show that for a sufficiently large hedging portfolio, the replication error will be arbitrarily small. Furthermore, we will provide error margins for the price estimators in terms of the regression imprecision. We thereby show that the direct estimator, lower bound and upper bound will converge to the true option price as the accuracy of the regressions increases. Cornerstone to the subsequent theorems is the universal approximation theorem, as presented in for example \cite{hornik1989multilayer}. Given that $\tilde V$ is a continuous function on the compact set $\mathcal I_d$, it guarantees that for each $m\in \{0,\ldots,M-1\}$ there exists a neural network $G_m$ such that
\begin{equation*} \begin{aligned}
    \sup_{x\in \mathcal{I}_d} B^{-1}(T_m)\left|\tilde V\left(T_m;x\right)-G_m\left(z_m(T_m)|x\right)\right|<\varepsilon
\end{aligned} \end{equation*}
for arbitrary $\varepsilon>0$. In other words, the regression error can be kept arbitrarily small on any compact domain of the risk-factor.

\subsection{Accuracy of the semi-static hedge}
In the theorem below we prove that the semi-static hedge can reach any desired level of accuracy. Let $\mathcal{T}_f=\left\{T_0,\dots,T_{M-1}\right\}$ denote the set of monitor dates. Recall that it is assumed that the risk-factors follow a Markovian process and that the market is free of arbitrage and frictionless. For the following theorem we additionally assume that $\mathbf{x}_t\in\mathcal{I}_d$ for some compact set $\mathcal{I}_d\subset\mathbb{R}^d$. As $\mathcal{I}_d$ can be arbitrarily large, this assumption is loose enough to account for a vast majority of the risk-factor scenarios in a standard Monte Carlo sample. On top of that, $\mathcal{I}_d$ can be chosen sufficiently large such that $\mathbb{E^Q}\left[\left|\tilde V\left(T_m\right)-G_m\left(z_m\right)\right|\mathbbm{1}_{\left\{\mathbf{x}_{T_m}\not\in\mathcal{I}_d\right\}}\Big| \mathcal{F}_0\right]$ approaches zero.


\begin{theorem}
Let $\varepsilon>0$ and $|\mathcal{T}_f|=M$. Denote by $\tilde V\left(t\right) $ the value of the replication portfolio for a Bermudan swaption, conditional on the fact it is not exercised prior to time $t$. Assume that there exist $M$ networks $G_m(\cdot)$, such that
\begin{equation*} \begin{aligned}
    \sup_{x\in \mathcal{I}_d} B^{-1}\left(T_m\right) \left|\tilde V\left(T_m;x\right)-G_m(z_m(T_m)|x)\right|<\varepsilon,\qquad \forall_{m\in\{0,\ldots,M-1\}}
\end{aligned} \end{equation*}
Then for any $t\in\left[0,T_{M-1}\right]$ we have that
\begin{equation*} \begin{aligned}
    \sup_{x \in \mathcal{I}_d}B^{-1}\left(t\right)\left|V\left(t;x\right)-\tilde V\left(t;x\right)\right|<M\varepsilon
\end{aligned} \end{equation*}
\label{theorem: sup error}
\end{theorem}
\begin{proof}
We prove by induction on $m$. At the last exercise date of the Bermudan, i.e. $t=T_{M-1}$, we have $V\left(T_{M-1};x\right)=\tilde V\left(T_{M-1};x\right):=\max \left\{h_{M-1}\left(x\right),0\right\}$, representing the  final pay-off of the contract, which at $T_{M-1}$ is exactly known. Hence it should be obvious that  
\begin{equation*} \begin{aligned}
    \sup_{x\in \mathcal{I}_d} B^{-1}\left(T_{M-1}\right)\left| V\left(T_{M-1};x\right)-\tilde V(T_{M-1};x)\right|=0
\end{aligned} \end{equation*}
For the inductive step, assume that for some $T_{m+1}\in\mathcal{T}_f$, an approximation $\tilde V(T_{m+1})$ of the price is given, satisfying
\begin{equation*} \begin{aligned}
    \sup_{x \in \mathcal{I}_d}B^{-1}\left(T_{m+1}\right)\left|V\left(T_{m+1};x\right)-\tilde V\left(T_{m+1};x\right)\right|<k \varepsilon
\end{aligned} \end{equation*}
We will show that it follows that for all $t\in[T_m,T_{m+1})$
\begin{equation*} \begin{aligned}
    \sup_{x \in \mathcal{I}_d}B^{-1}\left(t\right)\left|V\left(t;x\right)-\tilde V\left(t;x\right)\right|<(k+1)\varepsilon
\end{aligned} \end{equation*}
First consider the case $t\in\left(T_{m},T_{m+1}\right)$. It follows that
\begin{equation*} \begin{aligned}
    \sup_{x \in \mathcal{I}_d}\left|\frac{V(t;x)-\tilde V(t;x)}{B\left(t\right)}\right|&=\sup_{x \in \mathcal{I}_d}\left|\frac{C_m(t;x)-\tilde C_m(t;x)}{B\left(t\right)}\right|\\
    &=\sup_{x \in \mathcal{I}_d}\left|\mathbb{E^Q}\left[\frac{V\left(T_{m+1}\right)}{B\left(T_{m+1}\right)}\bigg|\mathbf{x}_t=x\right]-\mathbb{E^Q}\left[\frac{G_{m+1}\left(z_{m+1}\right)}{B\left(T_{m+1}\right)}\bigg|\mathbf{x}_t=x\right]\right|\\
    &\leq\sup_{x \in \mathcal{I}_d}\mathbb{E^Q}\left[B^{-1}\left(T_{m+1}\right)\left|V\left(T_{m+1}\right)-G_{m+1}(z_{m+1})\right|\Big|\mathbf{x}_{t}=x\right]\\
    &=\sup_{x \in \mathcal{I}_d}\mathbb{E^Q}\left[\vphantom{\Big|}B^{-1}\left(T_{m+1}\right)\left|V\left(T_{m+1}\right)-\tilde V\left(T_{m+1}\right)\right.\right.\\
    &\qquad\qquad \left.\left.+\tilde V\left(T_{m+1}\right)-G_{m+1}(z_{m+1})\right|\Big|\mathbf{x}_{t}=x\right]\\
    &\leq\sup_{x \in \mathcal{I}_d}\left(\mathbb{E^Q}\left[\vphantom{\Big|}B^{-1}\left(T_{m+1}\right)\left|V\left(T_{m+1}\right)-\tilde V\left(T_{m+1}\right)\right|\Big|\mathbf{x}_{t}=x\right]\right.\\
    &\qquad\qquad \left.+\mathbb{E^Q}\left[B^{-1}\left(T_{m+1}\right)\left|\tilde V\left(T_{m+1}\right)-G_{m+1}(z_{m+1})\right|\Big|\mathbf{x}_{t}=x\right]\right)\\
\end{aligned} \end{equation*}
In the last expression above, the first term is bounded due to the induction hypothesis, i.e. $B^{-1}\left(T_{m+1}\right)\left|V\left(T_{m+1}\right)-\tilde V\left(T_{m+1}\right)\right|<k\varepsilon$. The second term is bounded by assumption, i.e. there exists a network $G_{m+1}(\cdot)$, such that $B^{-1}\left(T_{m+1}\right)\left|\tilde V\left(T_{m+1}\right)-G_{m+1}(z_{m+1})\right|<\varepsilon$. We hence conclude that
\begin{equation*} \begin{aligned}
    \sup_{x \in \mathcal{I}_d}B^{-1}\left(t\right)\left|V(t;x)-\tilde V(t;x)\right|<(k+1)\varepsilon,\qquad \forall_{t\in\left(T_m,T_{m+1}\right)}
\end{aligned} \end{equation*}
If on the other hand $t=T_{m}$ we have that
\begin{equation*} \begin{aligned}
    \sup_{x \in \mathcal{I}_d}\left|\frac{V(t;x)-\tilde V(t;x)}{B\left(t\right)}\right|&=\sup_{x \in \mathcal{I}_d}\left|\frac{\max\left\{C_m(t;x),\,h_m(x)\right\}-\max\left\{\tilde C_m(t;x),\,h_m(x)\right\}}{B\left(t\right)}\right|
\end{aligned} \end{equation*}
Denoting $H(x):=B^{-1}\left(t\right)\left|\max\left\{C_m(t;x),\,h_m(x)\right\}-\max\left\{\tilde C_m(t;x),\,h_m(x)\right\}\right|$ in the expression above, we can distinguish 4 cases for each $x \in \mathcal{I}_d$, which are
\begin{itemize}
    \item $C_m(t;x),\tilde C_m(t;x)>h_m(x)$, then $H(x)=B^{-1}\left(t\right)\left|C_m(t;x)-\tilde C_m(t;x)\right|<(k+1)\varepsilon$
    \item $C_m(t;x),\tilde C_m(t;x)<h_m(x)$, then $H(x)=B^{-1}\left(t\right)\left|h_m(x)- h_m(x)\right|=0<(k+1)\varepsilon$
    \item $C_m(t;x)<h_m(x)< \tilde C_m(t;x)$, then $H(x)=B^{-1}\left(t\right)\left|h_m(x)-\tilde C_m(t;x)\right|$\\$<B^{-1}\left(t\right)\left|C_m(t;x)-\tilde C_m(t;x)\right|<(k+1)\varepsilon$
    \item $\tilde C_m(t;x)<h_m(x)< C_m(t;x)$, then $H(x)=B^{-1}\left(t\right)\left|C_m(t;x)-h_m(x)\right|$\\$<B^{-1}\left(t\right)\left|C_m(t;x)-\tilde C_m(t;x)\right|<(k+1)\varepsilon$
\end{itemize}
From all the cases we can induce that
\begin{equation*} \begin{aligned}
    \sup_{x \in \mathcal{I}_d}B^{-1}\left(t\right)\left|V(t;x)-\tilde V(t;x)\right| \leq (k+1) \varepsilon
\end{aligned} \end{equation*}
We conclude that by induction on $m=M-1,\ldots,0$ that
\begin{equation*} \begin{aligned}
    \sup_{x \in \mathcal{I}_d}B^{-1}\left(t\right)\left|V\left(t;x\right)-\tilde V\left(t;x\right)\right|<M\varepsilon
\end{aligned} \end{equation*}
for all $t\in\left[0,T_{M-1}\right]$.
\end{proof}

\subsection{Error of the direct estimator}
\Cref{theorem: sup error} bounds the hedging-error of the semi-static hedge in terms of the maximum regression errors. This implicitly provides an error margin to the direct estimator under the aforementioned assumptions. Although universal approximation theorem guarantees that the supremum-errors can be kept at any desired level, in practice they are substantially higher than for example the MSEs or MAEs of the regression function. This is due to inevitable fitting imprecision outside or near the boundaries of the finite training sets. In the following theorem we propose that the error of the direct estimator can be bounded in terms of the discounted MAEs of the neural networks. These quantities are generally much tighter than the supremum-errors and are typically easier to estimate.

The proof of the theorem follows a similar line of thought as the proof of \cref{theorem: sup error}. As the direct estimator at time-zero depends on the expectation of the continuation value at $T_0$, we can show by an iterative argument that the overall error is bounded by the sum of the mean absolute fitting errors at each monitor date. The error-bound in the direct estimator therefore scales linearly with the number of exercise opportunities. For a complete proof we refer to \cref{proof: direct error}.

\begin{theorem}\label{theorem: MAE error}
Let $\varepsilon >0$ and assume $|\mathcal{T}_f|=M$. Denote by $\tilde V $ the time-zero direct estimator for the price of a Bermudan swaption $V$. Assume that for each $T_m\in \{T_0,\ldots,T_{M-1}\}$ there is a neural network approximation $G_m(\cdot)$, such that
\begin{equation*} \begin{aligned}
    \mathbb{E^Q}\left[B^{-1}\left(T_m\right)\left|\tilde V\left(T_m\right)-G_m\left(z_m\right)\right|\bigg| \mathcal{F}_0\right]<\varepsilon
\end{aligned} \end{equation*}
where $\tilde V\left(T_m\right):=\max\left\{B(T_m)\mathbb{E^Q}\left[\frac{G_{m+1}\left(z_{m+1}\right)}{B\left(T_{m+1}\right)}\bigg|\mathcal{F}_{T_m}\right], h_m\left(\mathbf{x}_{T_m}\right)\right\}$ denotes the estimator at date $T_m$. Then the error in $\tilde V$ is bounded as given below
\begin{equation*} \begin{aligned}
    \left|V(0)-\tilde V(0)\right|<M\varepsilon
\end{aligned} \end{equation*}
\end{theorem}

\subsection{Tightness of the lower bound estimate}
A lower bound $L(t)$ to the true price can be computed by considering the non-optimal exercise strategy, implied by the direct estimator (see \cref{sec: lower bound}). This relies on the stopping time
\begin{equation} \begin{aligned}
    \tilde\tau(\omega)=\min\left\{T_m\in\mathcal{T}_f\big|\tilde C_m\left(T_m\right)\leq h_m\left(\mathbf{x}_{T_m}\right)\right\}\label{eqn: stopping time lower bound}
\end{aligned} \end{equation}
In the following theorem we propose that the tightness of $L(0)$ can be bounded by the discounted MAEs of neural network approximations.

The proof of the theorem relies on the fact that conditioned on any realization of $\tilde\tau$ and $\tau$, the expected difference between $L(0)$ and $V(0)$ is bounded by the sum of the mean absolute fitting errors at the monitor dates between $\tilde\tau$ and $\tau$. In the proof we therefore distinguish between the events $\tilde\tau<\tau$ and $\tilde\tau>\tau$. Then, by an inductive argument we can show that the bound on the spread between $L(0)$ and the true price scales linearly with the number of exercise opportunities. For a complete proof we refer to \cref{proof: lower bound}.

\begin{theorem}\label{theorem: lower bound}
Let $\varepsilon >0$ and assume $|\mathcal{T}_f|=M$. Denote by $L(0)$ the lower bound on the true Bermudan swaption price as defined in \ref{eqn: lower bound}. Assume that for each $T_m\in \{T_0,\ldots,T_{M-1}\}$ there is a neural network approximation $G_m(\cdot)$, such that
\begin{equation*} \begin{aligned}
    \mathbb{E^Q}\left[B^{-1}\left(T_m\right)\left|\tilde V\left(T_m\right)-G_m\left(z_m\right)\right|\bigg| \mathcal{F}_0\right]<\varepsilon
\end{aligned} \end{equation*}
where $\tilde V\left(T_m\right):=\max\left\{B(T_m)\mathbb{E^Q}\left[\frac{G_{m+1}\left(z_{m+1}\right)}{B\left(T_{m+1}\right)}\Big|\mathcal{F}_{T_m}\right], h_m\left(\mathbf{x}_{T_m}\right)\right\}$ denotes the estimator at date $T_m$. Then the spread between $V(0)$ and $L(0)$ is bounded as given below
\begin{equation*} \begin{aligned}
    \left|V(0)-L(0)\right|<2(M-1)\varepsilon
\end{aligned} \end{equation*}
\end{theorem}

\subsection{Tightness of the upper bound estimate}
An upper bound $U(t)$ to the true price can be computed by considering a dual formulation of the dynamic pricing equation \cite{haugh2004pricing}, see \cref{sec: Upper bound}. From a practical point of view, the difference between the upper bound and the true price can be interpreted as the maximum loss an investor would incur due to hedging imprecision resulting from the algorithm \cite{lokeshwar2019neural}. The overall hedging error at some monitor date $T_m$ is the result of all incremental hedging errors occurring from rebalancing the portfolio at preceding monitor dates. As the incremental hedging errors can be bounded by the sum of the expected absolute fitting errors, we propose that the tightness of $U(t)$ can be bounded by the discounted MAEs of the neural networks and scales at most quadratically with the number of exercise opportunities.

The proof follows a similar line of thought as presented in \cite{andersen2004primal}. There it is noted that the difference between the dual formulation of the option and its true price is difficult to bound. Here we make a similar remark and propose a  theoretical maximum spread between $U(0)$ and $V(0)$ that is relatively loose. Our numerical experiments however indicate that the upper bound estimate is much tighter in practice. For a complete proof we refer to \cref{proof: upper bound}.
\begin{theorem}\label{theorem: upper bound}
Let $\varepsilon >0$ and assume $|\mathcal{T}_f|=M$. Denote by $U(0)$ the upper bound on the true Bermudan swaption price as defined in \ref{eqn: martingale ineq}. Assume that for each $T_m\in \{T_0,\ldots,T_{M-1}\}$ there is a neural network approximation $G_m(\cdot)$, such that
\begin{equation*} \begin{aligned}
    \mathbb{E^Q}\left[B^{-1}\left(T_m\right)\left|\tilde V\left(T_m\right)-G_m\left(z_m\right)\right|\bigg| \mathcal{F}_0\right]<\varepsilon
\end{aligned} \end{equation*}
where $\tilde V\left(T_m\right):=\max\left\{B(T_m)\mathbb{E^Q}\left[\frac{G_{m+1}\left(z_{m+1}\right)}{B\left(T_{m+1}\right)}\Big|\mathcal{F}_{T_m}\right], h_m\left(\mathbf{x}_{T_m}\right)\right\}$ denotes the estimator at date $T_m$. Then the spread between $V(0)$ and $U(0)$ is bounded as given below
\begin{equation*} \begin{aligned}
    \left|U(0)-V(0)\right|<M(M-1)\varepsilon
\end{aligned} \end{equation*}
\end{theorem}

\section{Numerical experiments\label{sec: numerical examples}}
In this section we treat several numerical examples to illustrate the convergence, pricing and hedging performance of our proposed method. We will start by considering the price estimate of a vanilla swaption contract in a 1-factor model. This is a toy-example by which we can demonstrate the accuracy of the direct estimator in comparison to exact benchmarks. We continue with price estimates of Bermudan swaption contracts in a 1-factor and a 2-factor framework. The performance of the direct estimator will be compared to the established least-square regression method (LSM) introduced in \cite{longstaff2001valuing}, fine-tuned to an interest rate setting as described in \cite{feng2016efficient}. Additionally we will approximate the lower- and upper bound estimates as described in \cref{bounds theory} and show that they are well inside the error-margins introduced in \cref{sec: error analysis}. Finally we will illustrate the performance of the static hedge for a swaption in a 1-factor model and a Bermudan swaption in a 2-factor model. For the 1-factor case we can benchmark the performance by the analytic delta-hedge for a swaption, provided in \cite{henrard2003explicit}.

A $T_0\times T_M$ contract (either European swaption or Bermudan swaption) refers to an option written on a swap with a notional amount of 100 and a lifetime between $T_0$ and $T_M$. This means that $T_0$ and $T_{M-1}$ are the first and last monitor date respectively in case of a Bermudan. The underlying swaps are set to exchange annual payments, yielding year-fractions of 1 and annual exercise opportunities. All examples that are illustrated here have been implemented in python, using the Quant-Lib library \cite{Ame2003} for standard pricing routines and Keras with Tensorflow backend \cite{chollet2015keras} for constructing, fitting and evaluating the neural networks.

\subsection{1-factor swaption}
We start by considering a swaption contract under a one dimensional risk-factor setting. The direct estimator of the true $V(0)$ swaption price is computed similar to a Bermudan swaption, but with only a single exercise possibility at $T_0$. Therefore only a single neural network per option needs to be trained to compute the option price. We have used 64 hidden nodes and 20,000 training-points, generated through Monte Carlo sampling. We assume the risk-factor to be captured by the Hull-White model with constant mean-reversion parameter $a$ and constant volatility $\sigma$. The dynamics of the shifted mean-zero process \cite{brigo2007interest} are hence given by
\begin{equation} \begin{aligned}
    dx(t)=-ax(t)dt+\sigma dW(t),\qquad x(0)=0\label{eqn: hull-white}
\end{aligned} \end{equation}
For simplicity we consider a flat time-zero instantaneous forward rate $f(0,t)$. The risk-neutral scenarios are generated using a discrete Euler scheme of the process above. Parameter values that were used in the numerical experiments are summarized in \cref{table:1F parameters}.
\begin{table}
\centering
\begin{tabular}{ |c|c|c|c| } 
 \hline
 \textbf{Parameter} & $a$ & $\sigma$ & $f(0,t)$ \\ 
 \hline
 \textbf{Value} & 0.01 & 0.01 & 0.03 \\ 
 \hline
\end{tabular}
\caption{Parameters 1F Hull-White model}
\label{table:1F parameters}
\end{table}

Figures \ref{fig:swaption result 5y10y} and \ref{fig:swaption result 10y5y} show the time-zero option values in basis points (0.01\%) of the notional for a $5Y\times10Y$ and a $10Y\times5Y$ payer swaption as a function of the moneyness. The moneyness is defined as $\frac{S}{K}$, where $K$ denotes the fixed strike and $S$ the time-zero swap rate associated with the underlying swap. The exact benchmarks are computed by an application of Jamshidian's decomposition \cite{jamshidian1989exact}. The relative estimate errors are shown in Figures \ref{fig:swaption error 5y10y} and \ref{fig:swaption error 10y5y}. We observe a close agreement between the estimates and the reference prices. The errors are in the order of several basis points of the true option price. In the current setting the results presented serve mostly as a validation of the estimator. We however point out that this algorithm for swaptions is applicable in general frameworks, such as multi-factor, dual-curve or non-overlapping payment schemes, for which exact routines are no longer available.

\begin{figure}
\centering
\begin{subfigure}{.5\textwidth}
  \centering
  \includegraphics[width=0.93\linewidth]{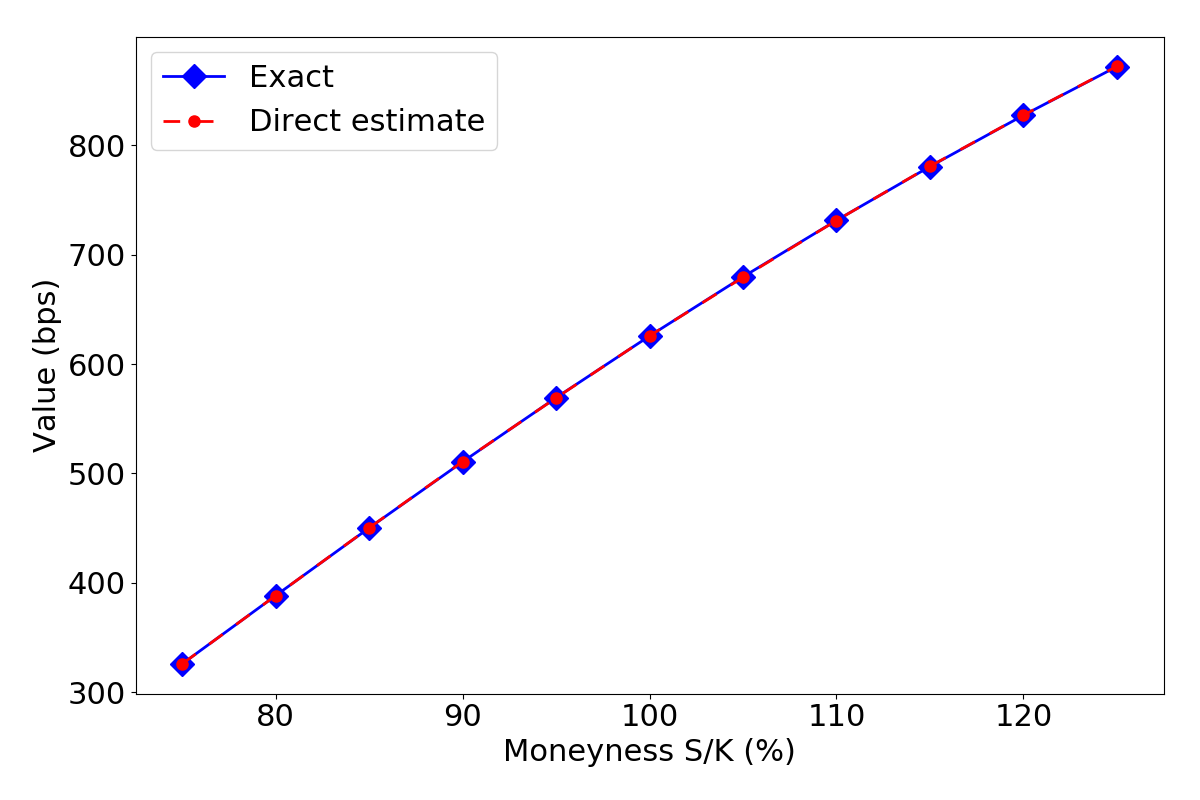}
  \caption{5Yx10Y swaption valuation}
  \label{fig:swaption result 5y10y}
\end{subfigure}%
\begin{subfigure}{.5\textwidth}
  \centering
  \includegraphics[width=0.93\linewidth]{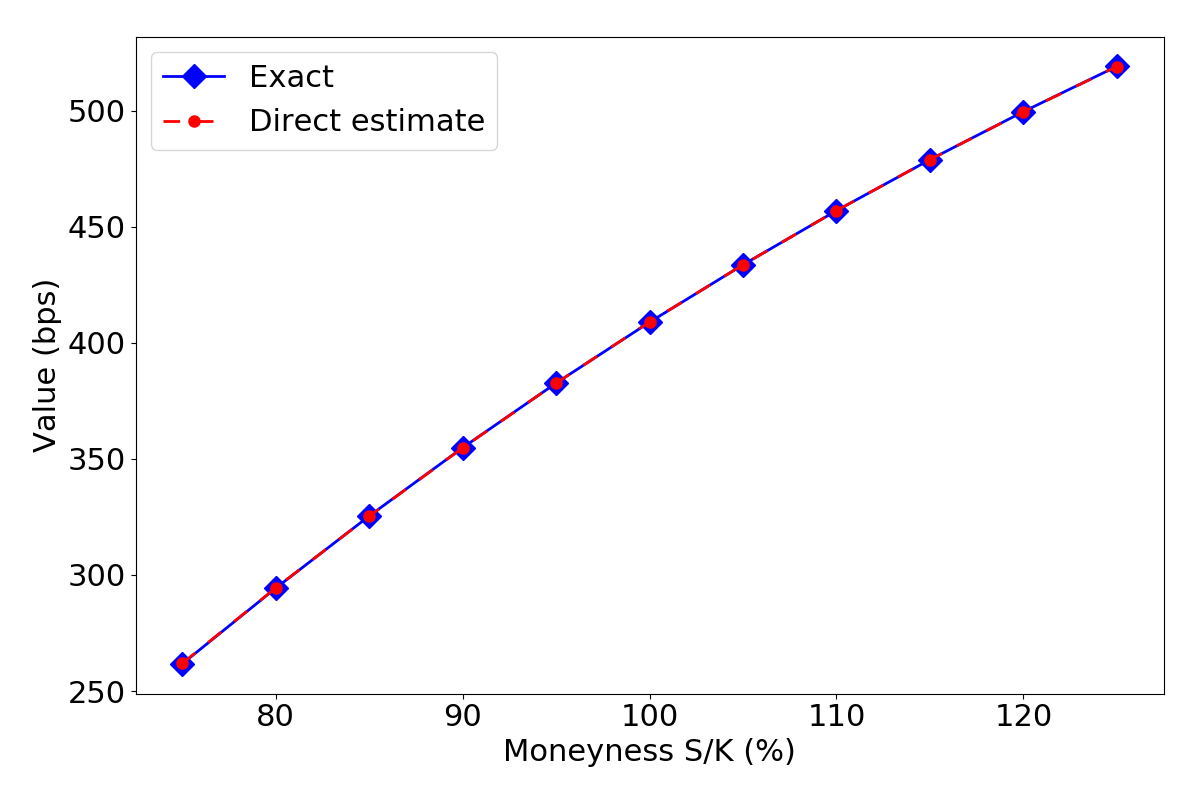}
  \caption{10Yx5Y swaption valuation}
  \label{fig:swaption result 10y5y}
\end{subfigure}\\
\begin{subfigure}{.5\textwidth}
  \centering
  \includegraphics[width=0.95\linewidth]{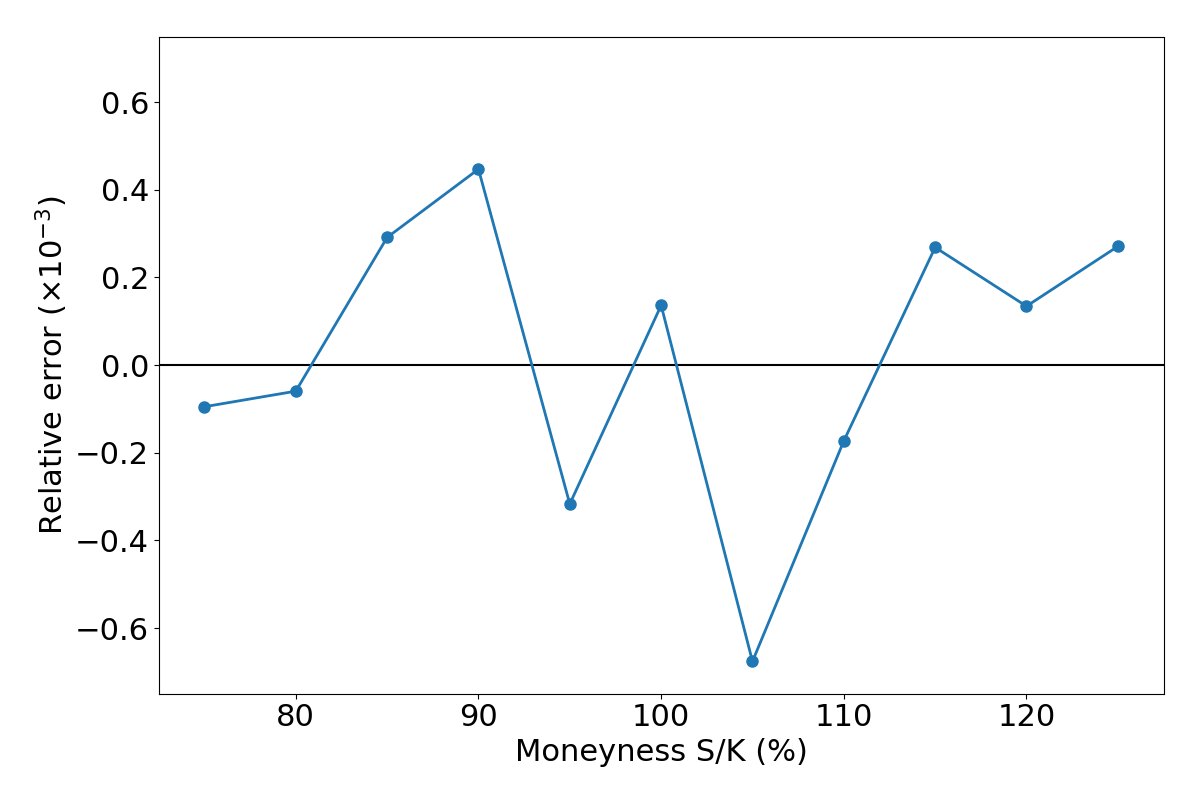}
  \caption{Error 5Yx10Y swaption estimates}
  \label{fig:swaption error 5y10y}
\end{subfigure}%
\begin{subfigure}{.5\textwidth}
  \centering
  \includegraphics[width=0.95\linewidth]{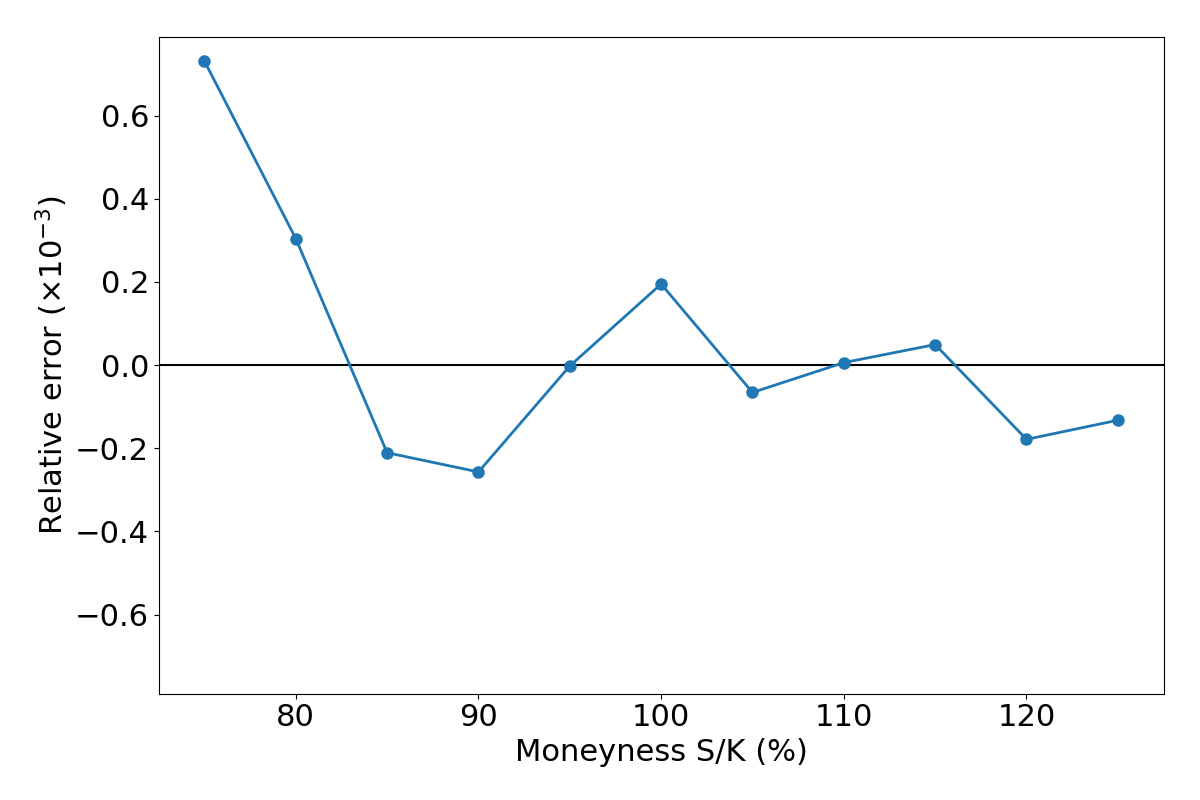}
  \caption{Error 10Yx5Y swaption estimates}
  \label{fig:swaption error 10y5y}
\end{subfigure}
\caption{Accuracy of the direct estimator for vanilla swaptions. $S_{5Y\times10Y}\approx S_{10Y\times5Y}\approx0.0305$.}
\label{fig: swaption results}
\end{figure}

\subsection{1-factor Bermudan swaption\label{sec Num ex swaption}}
As a second example we consider a Bermudan swaption contract. The same dynamics for the underlying risk-factor are assumed as discussed in the previous paragraph, using the parameter settings of \cref{table:1F parameters}. Monte Carlo scenarios are generated based on a discretised Euler scheme associated to the SDE in \ref{eqn: hull-white}, taking weekly time-steps.

We first demonstrate the convergence property of the direct estimator, that is implied by the replication portfolio. We consider a $1Y\times 5Y$ Bermudan swaption with strike $K=0.03$. For this analysis, the neural networks were trained to a set of 2000 Monte Carlo generated training points. \cref{fig: 1F convergence} shows the direct estimator as a function of the number of hidden nodes in each neural network, alongside an LSM-based benchmark. In \cref{fig: 1F convergence error} the error with respect to the LSM estimate is shown on a logscale. We observe that the direct estimator converges to the LSM confidence interval or slightly above, which is in accordance with the fact that LSM is biased low by definition. The analysis indicates that a portfolio of 16 discount bond options is sufficient to achieve a replication of similar accuracy as the LSM benchmark.

\Cref{table:1F results} depicts numerical pricing results for a $1Y\times 5Y$, $3Y\times7Y$ and $1Y\times10Y$ receiver Bermudan swaption. For each contract we consider different levels of moneyness, setting the fixed rate $K$ of the underlying swap to respectively 80\%, 100\% and 120\% of the time-zero swap rate. The estimations of the direct, the upper bound and the lower bound statistics are again reported alongside LSM-based benchmarks. Here, the neural networks have 64 hidden nodes and are fitted using a training set of 20,000 points. The lower and upper bound estimates, as well as the LSM estimates, are based on simulation runs of 200,000 paths each. The given lower and upper bounds are Monte Carlo estimates of the statistics defined in \ref{eqn: lower bound} and \ref{eqn: martingale ineq} and are therefore subject to standard errors, which are reported in parentheses. The reference LSM results have been generated using $\left\{1,x,x^2\right\}$ as regression basis functions for approximating the continuation values. The standard errors and confidence intervals are obtained from ten independent Monte Carlo runs.  The choice for hyperparameter settings is motivated by the analysis of \cref{sec: hyper-parameters}.

The spreads between the lower and upper bound estimates provide a good indication of the accuracy of the method. For the current setting we obtain spreads in the order of several basis points up a few dozen of basis points. The lower bound estimate is typically very close to the LSM estimate, which itself is also biased low. Their standard errors are of the same order of magnitude. The upper bound estimates prove to be very stable and show a variance that is roughly two orders of magnitude smaller compared to that of the lower bound. The direct estimate is occasionally slightly less accurate. This can be explained by the fact that it depends on the accuracy of the regression over the full domain of the risk-factor, whereas for the lower bound only a high accuracy near the exercise boundaries is required. \Cref{fig: Fitting errors 1F} presents the mean absolute error of each neural network after fitting as a function of the network's index. The errors are displayed in basis points of the notional. We observe that the errors are the smallest at maturity and tend to increase with each iteration backward in time. That the errors at the final monitor date are virtually zero can be explained by the fact that the pay-off at $T_{M-1}$ is given by
\begin{equation*} \begin{aligned}
    \max\left\{h_{M-1}(\mathbf{x}_{T_{M-1}}),\,0\right\}&=N \cdot \max\left\{A_{M-1,M}\left(T_{M-1}\right)\cdot\left(K-S_{M-1,M}\left(T_{M-1}\right)\right),\,0\right\}\\
    &=N\cdot\max\left\{(\Delta T_M K+1)P(T_{M-1},T_M)-1,\,0\right\}\\
    &\simeq w_2\varphi(w_1z-b)
\end{aligned} \end{equation*}
which can be exactly captured by a network with only a single hidden node. With each step backwards, the target function is harder to fit, yielding larger errors. We observe MAEs up to one basis point of the notional amount. The empirical lower-upper bound spreads remain well within the theoretical error margins provided in subsections \ref{sec: lower bound} and \ref{sec: Upper bound}. The spreads are mostly much lower than the sum of the MAEs, indicating that the bound estimates are in practice significantly tighter than their theoretical maximum spread.

\begin{figure}
\centering
\begin{subfigure}{.5\textwidth}
  \centering
  \includegraphics[width=0.93\linewidth]{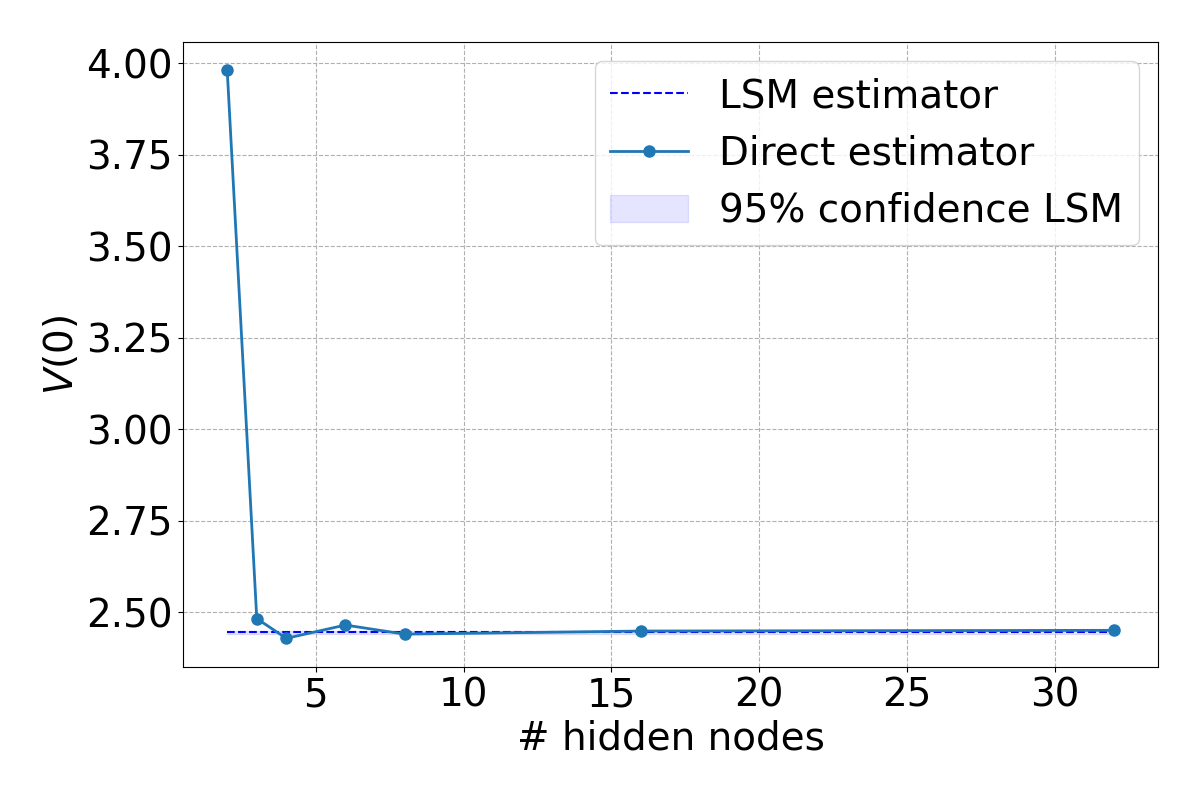}
  \caption{Convergence price}
  \label{fig: 1F convergence}
\end{subfigure}%
\begin{subfigure}{.5\textwidth}
  \centering
  \includegraphics[width=0.93\linewidth]{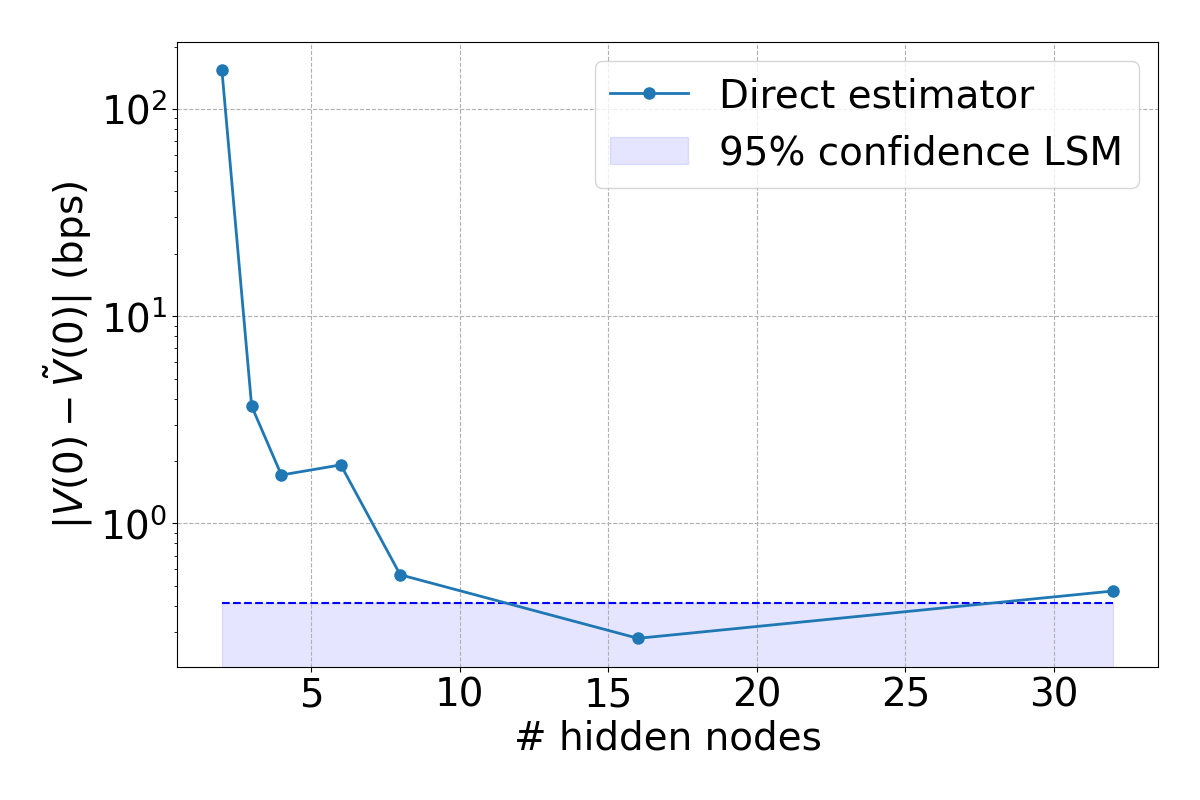}
  \caption{Convergence pricing error}
  \label{fig: 1F convergence error}
\end{subfigure}
\caption{Convergence of the direct estimator for the 1Yx5Y Bermudan swaption price as a function of hidden node count, with respect to the LSM benchmark under a 1-factor model.}
\label{fig: 1F convergence results}
\end{figure}

\begin{table}
\centering
\resizebox{\textwidth}{!}{
\begin{tabular}{c c c c c c c c} \toprule
    \textbf{Type} & \textbf{K/S} & \textbf{Dir.est.} & \textbf{Lower bnd} & \textbf{Upper bnd} & \textbf{UB-LB} & \textbf{LSM est.} & \textbf{LSM 95\% CI}  \\ \midrule
    1Y$\times$5Y  & 80\% & 1.527 & 1.521(0.001) & 1.528(0.000) & 0.007 & 1.521(0.001) & [1.518, 1.523] \\
      & 100\%  &  2.543  & 2.534(0.002) & 2.542(0.000) & 0.008 & 2.534(0.002) & [2.531, 2.538]\\
      & 120\%  & 4.015 & 4.016(0.002)  & 4.018(0.000) & 0.002 & 4.016(0.002) & [4.012, 4.021] \\ \midrule
    3Y$\times$7Y  & 80\% & 3.296 & 3.293(0.002) & 3.295(0.000) & 0.002 & 3.293(0.002) & [3.290, 3.296] \\
      & 100\%  &  4.767  & 4.755(0.004) & 4.761(0.000) & 0.006 & 4.755(0.004) & [4.747, 4.762]\\
      & 120\%  & 6.625 & 6.629(0.004)  & 6.631(0.000) & 0.002 & 6.629(0.004) & [6.621, 6.638] \\ \midrule
    1Y$\times$10Y  & 80\% & 3.950 &  3.945(0.005) & 3.960(0.000) & 0.015 & 3.945(0.005) & [3.935, 3.955] \\
      & 100\%  &  5.818  & 5.811(0.003) & 5.818(0.000) & 0.007 & 5.811(0.003) & [5.805, 5.816]\\
      & 120\%  & 8.346 & 8.354(0.005)  & 8.360(0.000) & 0.006 & 8.353(0.005) & [8.344, 8.362] \\ \bottomrule
\end{tabular}}
\caption{Results 1-factor model. $S_{1Y\times5Y}\approx S_{3Y\times7Y}\approx S_{1Y\times10Y}\approx0.0305$. Standard errors are in parentheses, based on 10 independent MC runs of $2\times10^5$ paths each.}
\label{table:1F results}
\end{table}

\begin{figure}
\centering
\begin{subfigure}{.3\textwidth}
  \centering
  \includegraphics[width=1.\linewidth]{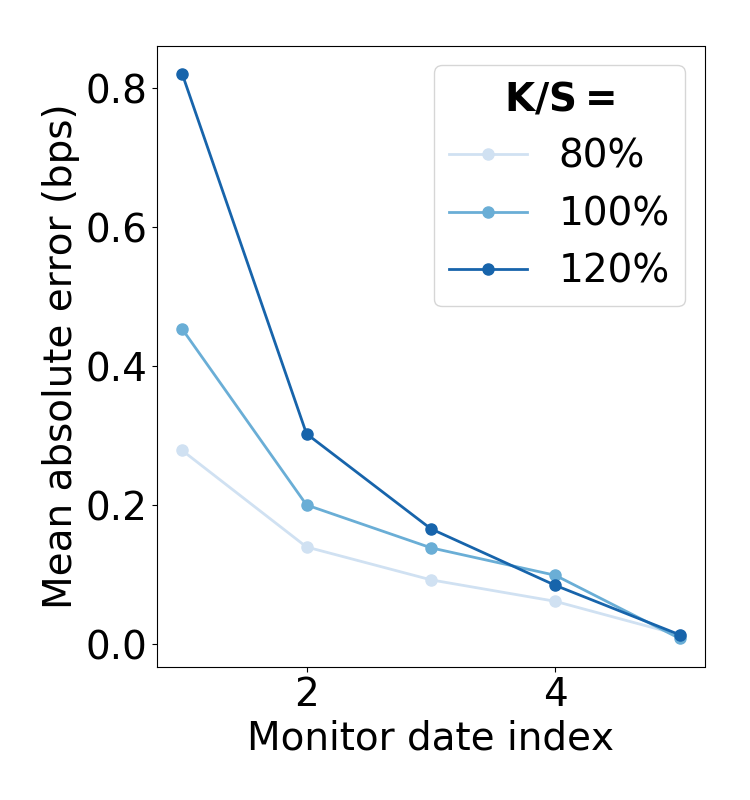}
  \caption{1Yx5Y Bermudan}
  \label{fig: error 1F 1y5y}
\end{subfigure}$\,$
\begin{subfigure}{.42\textwidth}
  \centering
  \includegraphics[width=1.\linewidth]{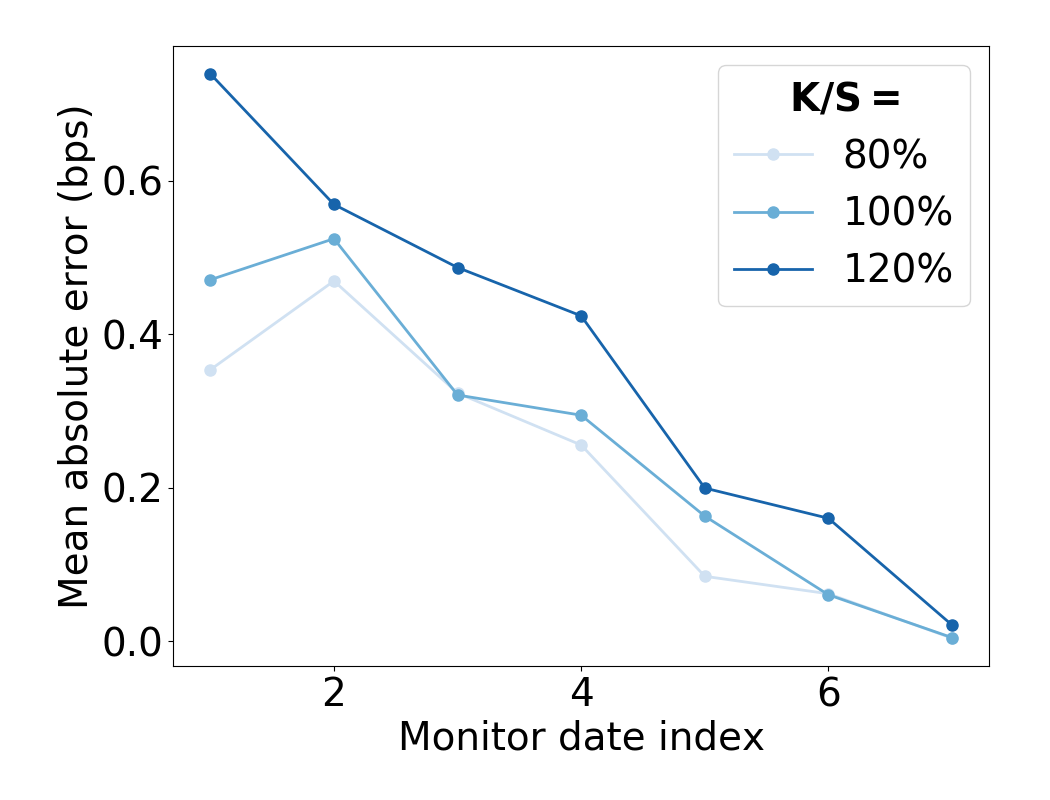}
  \caption{3Yx7Y Bermudan}
  \label{fig: error 1F 3y7y}
\end{subfigure}
\begin{subfigure}{.6\textwidth}
  \centering
  \includegraphics[width=1.0\linewidth]{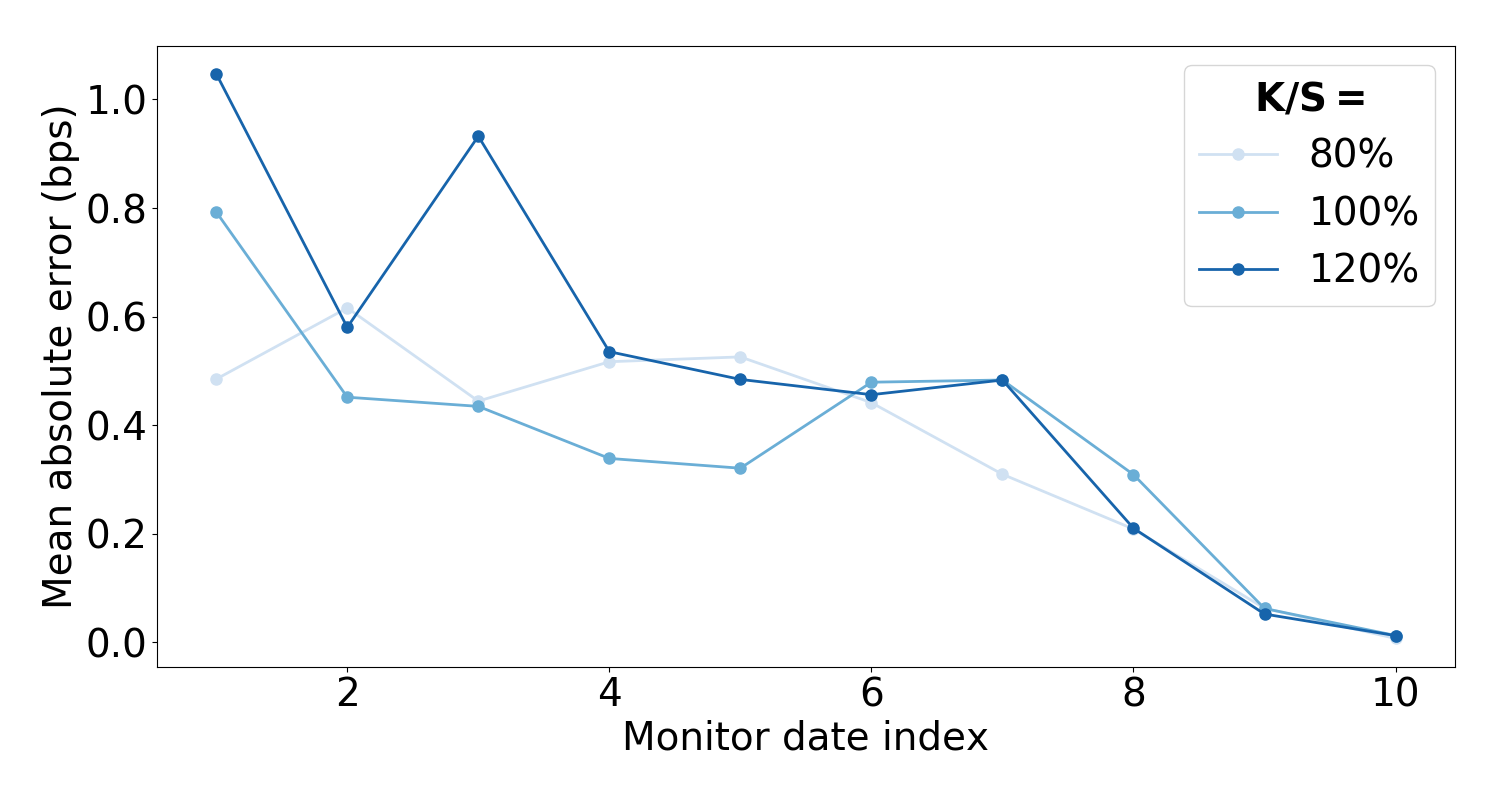}
  \caption{1Yx10Y Bermudan}
  \label{fig: error 1F 1y10y}
\end{subfigure}%
\caption{Mean absolute errors of neural network fit per monitor date under a 1-factor model.}
\label{fig: Fitting errors 1F}
\end{figure}

\subsection{2-factor Bermudan swaption\label{sec: num ex 2fac berm}}
As a final pricing example we consider a Bermudan swaption contract under a 2-factor model. The dynamics of the underlying risk-factors are assumed to follow a G2++ model \cite{brigo2007interest}. Monte Carlo scenarios are generated based on a discretised Euler scheme, taking weekly time-steps, based on the SDE below.
\begin{equation*} \begin{aligned}
    dx_1(t)&=-a_1x_1(t)dt+\sigma_1 dW_1(t),\qquad x_1(0)=0\\
    dx_2(t)&=-a_2x_2(t)dt+\sigma_2 dW_2(t),\qquad x_2(0)=0
\end{aligned} \end{equation*}
where $W_1$ and $W_2$ are correlated Brownian motions with $d\left<W_1,W_2\right>_t=\rho dt$. Parameter values that were used in the numerical experiments are summarized in \cref{table:2F parameters}.

We again start by demonstrating the convergence property of the direct estimator, for both the locally-connected and the fully-connected neural network designs as specified in \cref{sec: d-dim NN design}. The same $1Y\times 5Y$ Bermudan swaption with strike $K=0.03$ is used and the networks are each fitted to a set of 6400 training points. \cref{fig: 2F convergence} shows the direct estimator as a function of the number of hidden nodes in each neural network, alongside an LSM-based benchmark. In \cref{fig: 2F convergence error} the error with respect to the LSM estimate is shown on a logscale. We observe a similar convergence behaviour, where the direct estimators approach the LSM benchmark within the 95\% confidence range. Here it is noted that a portfolio of 8 discount bond options is already sufficient to achieve a replication of similar accuracy as the LSM estimator.

\Cref{table:2F results} depicts numerical results for a $1Y\times 5Y$, $3Y\times7Y$ and $1Y\times10Y$ receiver Bermudan swaption, for different levels of moneyness. We again report the direct, the upper bound and the lower bound estimates for both neural network designs. In this case, all networks have 64 hidden nodes and are fitted to training sets of 20,000 points. As before, the lower bound, the upper bound and the LSM estimates are the result of 10 independent Monte Carlo simulations of 200,000 scenarios.
\begin{table}
\centering
\begin{tabular}{ |c|c|c|c|c|c|c| } 
 \hline
 \textbf{Parameter} & $a_1$ & $a_2$ & $\sigma_1$ & $\sigma_2$ & $\rho$ & $f(0,t)$ \\ 
 \hline
 \textbf{Value} & 0.07 & 0.08 & 0.015 & 0.008 & -0.6 & 0.03 \\ 
 \hline
\end{tabular}
\caption{Parameters 2F G2++ model}
\label{table:2F parameters}
\end{table}
For the LSM algorithm we used $\left\{1,x_1,x_2,x_1^2,x_1x_2,x_2^2\right\}$ as basis-functions. Note that the number of monomials grows quadratically with the dimension of the state space and with that the number of free parameters. For our method, this number grows at a linear rate. Choices for the hyperparameters are again based on the analysis of \cref{sec: hyper-parameters}. The results under the 2-factor case share several features with the 1-factor results. We observe spreads between the lower and upper bounds ranging from several basis-points up to a few dozen basis points of the option price. The lower bound estimates turn out very close to the LSM estimates and the same holds for their standard errors. The upper bounds are again very stable with low standard errors and the direct estimator appears slightly less accurate. If we compare the locally-connected to the fully-connected case, we observe that the results are overall in close agreement, especially the lower and upper bound estimates. This is remarkable given that the fully-connected case gives rise to more trainable parameters, by which we would expect a higher approximation accuracy. In the 2-factor setting, the ratio of free parameters for the two designs is $3:4$.

\Cref{fig: Fitting errors 2F} shows the mean absolute errors of the neural networks after fitting. The MAEs for the locally-connected networks are in blue; the fully-connected are in red. All are represented in basis points of the notional amount. We observe that the errors are mostly in the same order of magnitude as the one-dimensional case. The figures indicate that the locally-connected networks slightly outperform the fully-connected networks in terms of accuracy, although this does not appear to materialize in tighter estimates of the lower and upper bounds. For the locally-connected case we again observe that the errors are virtually zero at the last monitor date, for the same reasons as in the 1-factor setting. In the fully-connected representation, an exact replication might not exist, resulting in larger errors. We conjecture that this effect partially carries over to the networks at preceding monitor dates. The empirical lower-upper bound spreads remain well within the theoretical error margins, as the spreads are in all cases lower than the sum of the MAEs. Hence also for the two-factor setting we find that the bound estimates are tighter in practice than their theoretical maximum spreads.

\begin{figure}
\centering
\begin{subfigure}{.5\textwidth}
  \centering
  \includegraphics[width=0.93\linewidth]{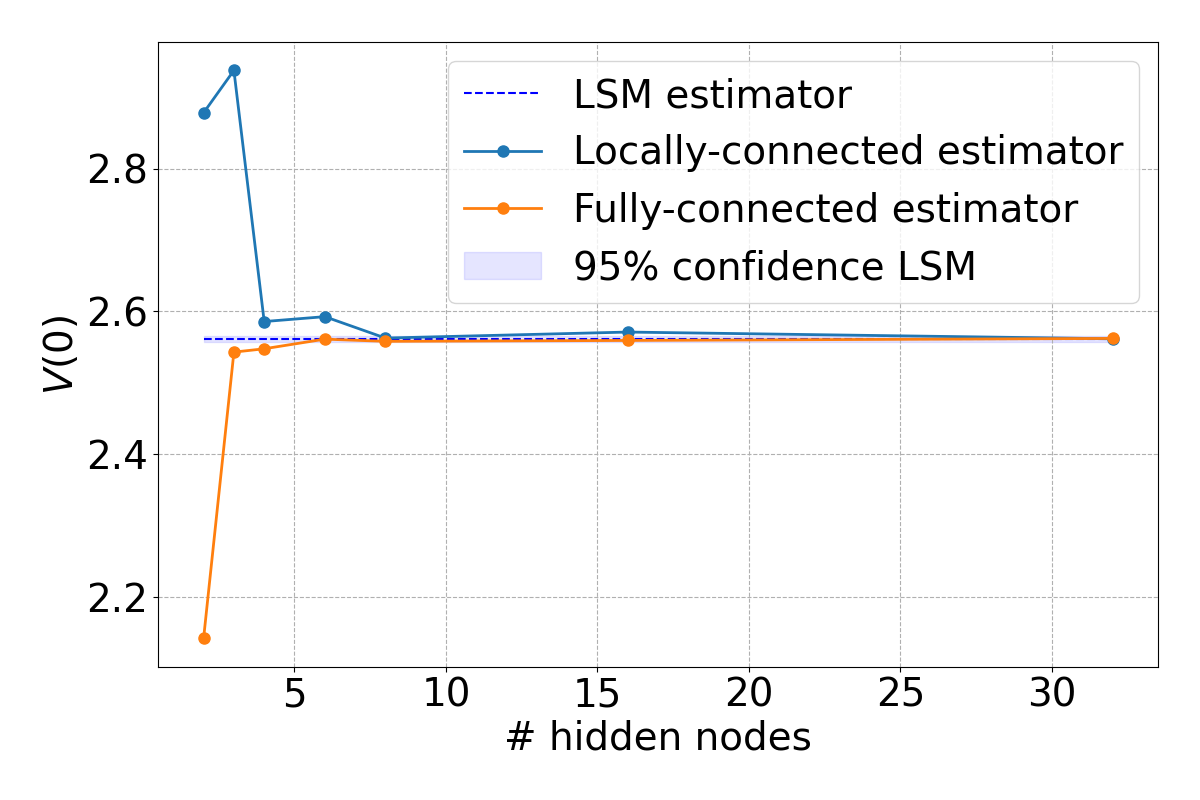}
  \caption{Convergence price}
  \label{fig: 2F convergence}
\end{subfigure}%
\begin{subfigure}{.5\textwidth}
  \centering
  \includegraphics[width=0.93\linewidth]{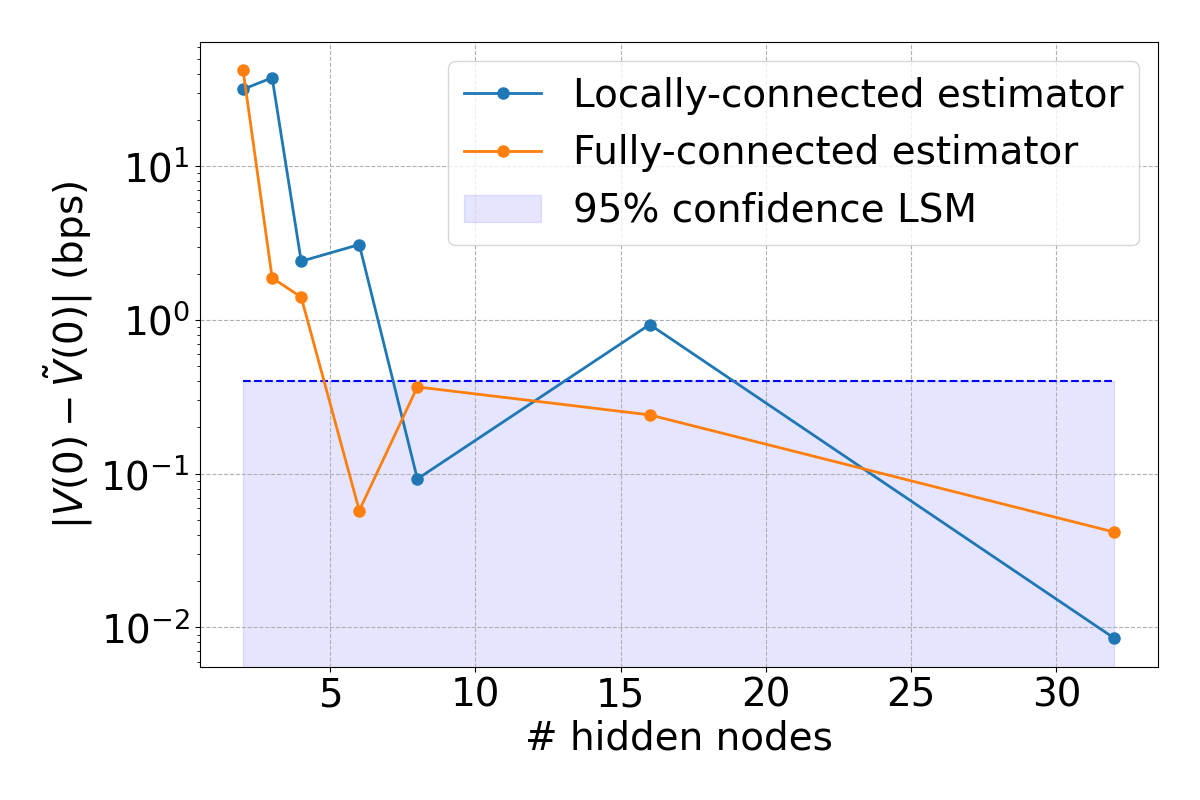}
  \caption{Convergence pricing error}
  \label{fig: 2F convergence error}
\end{subfigure}
\caption{Convergence of the direct estimator for the 1Yx5Y Bermudan swaption price as a function of hidden node count, with respect to the LSM benchmark under a 2-factor model.}
\label{fig: 2F convergence results}
\end{figure}

\begin{table}
\centering
\resizebox{\textwidth}{!}{
\begin{tabular}{cccccccc} 
    \toprule
    \multicolumn{8}{c}{\textsc{Locally-connected neural networks}} \\ \midrule
    \textbf{Type} & \textbf{K/S} & \textbf{Dir.est.} & \textbf{Lower bnd} & \textbf{Upper bnd} & \textbf{UB-LB} & \textbf{LSM est.} & \textbf{LSM 95\% CI}  \\ \midrule
    1Y$\times$5Y  & 80\% & 1.617 &  1.617(0.002) & 1.619(0.000) & 0.002 & 1.617(0.002) & [1.614, 1.621] \\
      & 100\%  &  2.652  & 2.650(0.002) & 2.654(0.000) & 0.004 & 2.650(0.002) & [2.646, 2.654]\\
      & 120\%  & 4.128 & 4.127(0.003)  & 4.131(0.000) & 0.004 & 4.127(0.003) & [4.121, 4.132] \\ \midrule
    3Y$\times$7Y  & 80\% & 3.073 &  3.076(0.004) & 3.078(0.000) & 0.002 & 3.077(0.004) & [3.069, 3.085] \\
      & 100\%  &  4.554  & 4.553(0.004) & 4.553(0.000) & 0.000 & 4.552(0.004) & [4.545, 4.559]\\
      & 120\%  & 6.444 & 6.448(0.004)  & 6.451(0.000) & 0.003 & 6.446(0.005) & [6.435, 6.456] \\ \midrule
    1Y$\times$10Y  & 80\% & 3.616 &  3.624(0.002) & 3.626(0.000) & 0.002 & 3.622(0.002) & [3.618, 3.627] \\
      & 100\%  &  5.508 & 5.509(0.002) & 5.514(0.000) & 0.005 & 5.508(0.002) & [5.503, 5.512]\\
      & 120\%  & 8.128 & 8.123(0.005)  & 8.130(0.000) & 0.007 & 8.121(0.005) & [8.110, 8.132] \\ \midrule
      \multicolumn{8}{c}{\textsc{Fully-connected neural networks}} \\ \midrule
    \textbf{Type} & \textbf{K/S} & \textbf{Dir.est.} & \textbf{Lower bnd} & \textbf{Upper bnd} & \textbf{UB-LB} & \textbf{LSM est.} & \textbf{LSM 95\% CI}  \\ \midrule
    1Y$\times$5Y  & 80\% & 1.617 &  1.617(0.002) & 1.619(0.000) & 0.002 & 1.617(0.002) & [1.614, 1.621] \\
      & 100\%  &  2.651  & 2.650(0.002) & 2.654(0.000) & 0.004 & 2.650(0.002) & [2.646, 2.654]\\
      & 120\%  & 4.129 & 4.127(0.003)  & 4.131(0.000) & 0.004 & 4.127(0.003) & [4.121, 4.132] \\ \midrule
    3Y$\times$7Y  & 80\% & 3.076 &  3.077(0.004) & 3.078(0.000) & 0.001 & 3.077(0.004) & [3.069, 3.085] \\
      & 100\%  &  4.553  & 4.553(0.004) & 4.554(0.000) & 0.001 & 4.552(0.004) & [4.545, 4.559]\\
      & 120\%  & 6.451 & 6.447(0.005)  & 6.451(0.000) & 0.004 & 6.446(0.005) & [6.435, 6.456] \\ \midrule
    1Y$\times$10Y  & 80\% & 3.616 &  3.624(0.002) & 3.626(0.000) & 0.002 & 3.622(0.002) & [3.618, 3.627] \\
      & 100\%  &  5.506  & 5.509(0.002) & 5.514(0.000) & 0.005 & 5.508(0.002) & [5.503, 5.512]\\
      & 120\%  & 8.124 & 8.123(0.005)  & 8.130(0.000) & 0.007 & 8.121(0.005) & [8.110, 8.132] \\ \bottomrule
\end{tabular}}
\caption{Results 2-factor model for the locally-connected and fully-connected neural network cases. $S_{1Y\times5Y}\approx S_{3Y\times7Y}\approx S_{1Y\times10Y}\approx0.0305$. Standard errors are in parentheses, based on 10 independent MC runs of $2\times10^5$ paths each.}
\label{table:2F results}
\end{table}

\begin{figure}
\centering
\begin{subfigure}{.3\textwidth}
  \centering
  \includegraphics[width=1.\linewidth]{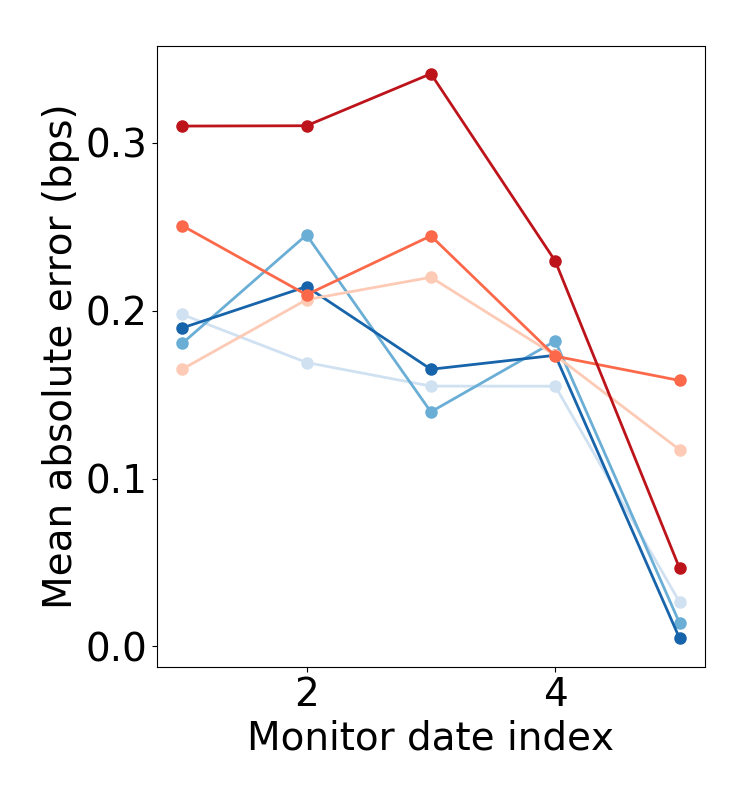}
  \caption{1Yx5Y Bermudan}
  \label{fig: error 2F 1y5y}
\end{subfigure}
\begin{subfigure}{.42\textwidth}
  \centering
  \includegraphics[width=1.\linewidth]{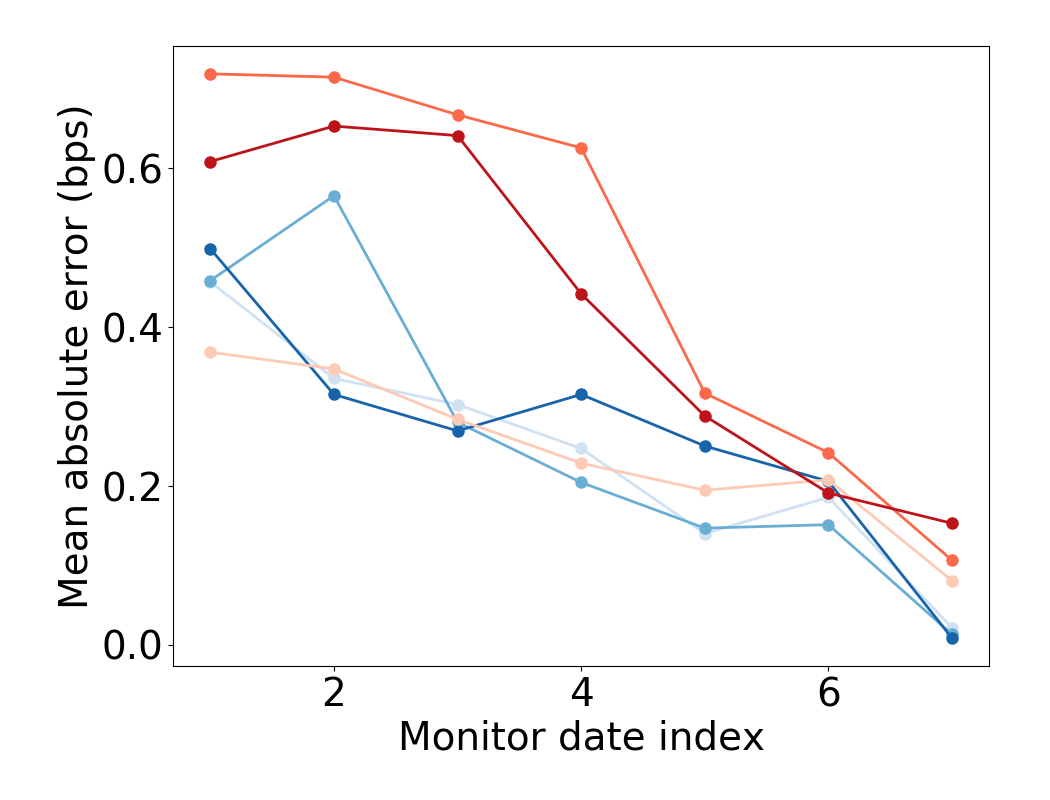}
  \caption{3Yx7Y Bermudan}
  \label{fig: error 2F 3y7y}
\end{subfigure}
\begin{subfigure}{.60\textwidth}
  \centering
  \includegraphics[width=1.0\linewidth]{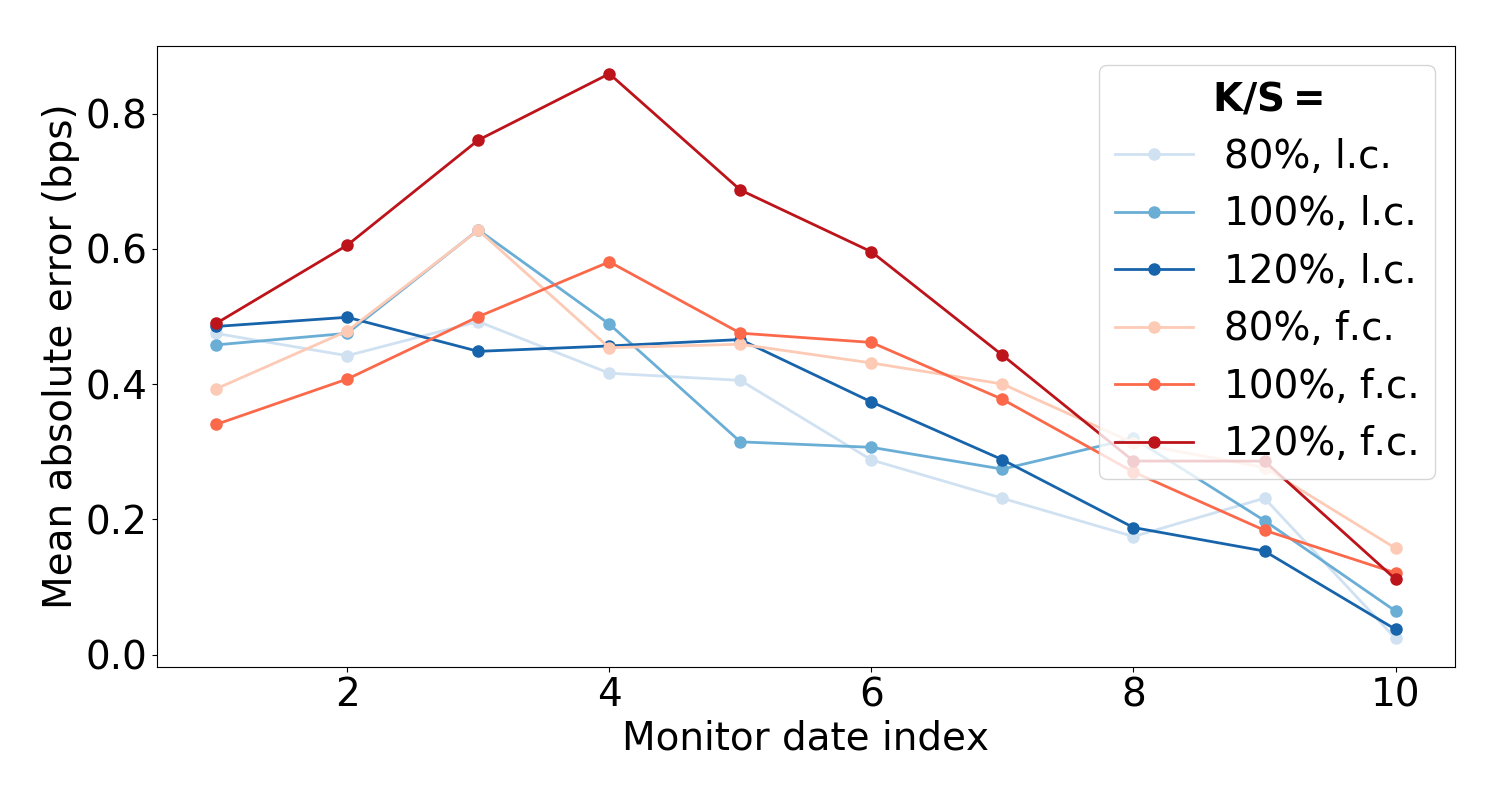}
  \caption{1Yx10Y Bermudan}
  \label{fig: error 2F 1y10y}
\end{subfigure}%
\caption{Accuracy of neural network fit per monitor date under a 2-factor model. Blue lines represent the locally-connected (l.c.) case and the red lines represent the fully-connected (f.c.) case. The legend in Figure (c) applies to all three graphs.}
\label{fig: Fitting errors 2F}
\end{figure}

\subsection{Performance semi-static hedge}
Finally, we consider the hedging problem of a vanilla swaption under the 1-factor model and a Bermudan swaption under the 2-factor model.

\subsubsection{1-factor swaption}
Here we compare the performance of a static versus a dynamic hedge in the 1-factor model. As an example, we take a $1Y\times5Y$  European receiver swaption at different levels of moneyness. The model set-up is similar to that in \cref{sec Num ex swaption}, using the same set of parameters reported in \cref{table:1F parameters}. In the static hedge case, the option contract writer aims to hedge the risk using a static portfolio of zero-coupon bond options and discount bonds. The replicating portfolio is composed using a neural network with 64 hidden nodes, optimized using 20,000 training-points generated through Monte Carlo sampling. The portfolio is composed at time-zero and kept until the expiry of the option at $t=1$ year. In the dynamic hedge case, the delta-hedging strategy is applied. The replicating portfolio is composed of units of the underlying forward-starting swap and investment in the money-market. The dynamic hedge involves the periodic rebalancing of the portfolio. The delta for a receiver swaption under the Hull-White model (see \cite{henrard2003explicit}) is given by
\begin{equation} \begin{aligned}
    \Delta(t)=\frac{\sum_{j=1}^M c_jP(t,T_j)\nu(t,T_j)\Phi(\kappa+\alpha_j)-P(t,T_0)\nu(t,T_0)\Phi(\kappa)}{\sum_{j=1}^M c_jP(t,T_j)\nu(t,T_j)-P(t,T_0)\nu(t,T_0)}
\end{aligned} \end{equation}
where $\kappa$ is the solution of
\begin{equation*} \begin{aligned}
    \sum_{j=1}^M c_j\frac{P(t,T_j)}{P(t,T_0)}\exp\left(-\frac{1}{2}\alpha_j^2-\alpha_j\kappa\right)=1
\end{aligned} \end{equation*}
and
\begin{equation*} \begin{aligned}
    \alpha_j^2:=\int_0^{T_0}\left(\nu(u,T_j)-\nu(u,T_0)\right)^2du
\end{aligned} \end{equation*}
where $\Phi$ denotes the CDF of a standard normal distribution, $c_j=\Delta T_jK$ for $j=1,\ldots,M-1$ and $c_M=1+\Delta T_MK$. The function $\nu(t,T)$ denotes the instantaneous volatility of a discount bond maturing at $T$, which under Hull-White is given by $\nu(t,T):=\frac{\sigma}{a}\left(1-e^{-a(T-t)}\right)$. We validated the analytic expression above with numerical approximations of the Delta obtained by bumping the yield curve. Within the simulation, the dynamic hedge portfolio is rebalanced on a daily basis between time-zero and expiry of the option. In this experiment that means it is updated on 255 instances at equidistant monitor dates.

The performance of both hedging strategies is reported in \cref{table: hedging errors}. The results are based on 10,000 risk-neutral Monte Carlo paths. The hedging error refers to the difference between the option's pay-off at expiry and the replicating portfolio's final value. The quantities are reported in basis points of the notional amount. The empirical distribution of the hedging error is shown in \cref{fig: Hedge error results}. We observe that overall the static hedge outperforms the dynamic hedge in terms of accuracy, even though it involves only a quarter (64 versus 255) of the trades. Although it is not visible in \cref{fig: HE static}, the static strategy does give rise to occasional outliers in terms of accuracy. These are associated with scenarios that reach or exceed the boundary of the training set. These errors are typically of a similar order of magnitude as the errors observed in the dynamic hedge. The impact of outliers can be reduced by increasing the training-set and thereby broadening the regression-domain.
\begin{table}
\centering
\begin{tabular}{cccc} 
    \toprule
    \textbf{Hedge error (bps)} & \textbf{K/S} & \textbf{Static hedge} & \textbf{Dyn. hedge} \\
    \midrule
    Mean  & 80\% &$-1.9\times10^{-2}$ & $0.38$\\
          & 100\% &$-2.2\times10^{-3}$& $0.61$ \\
          & 120\% &$-1.5\times10^{-2}$& $0.46$ \\\midrule
    St. dev.  & 80\% & $2.5$ & $9.1$\\
          & 100\% &$3.1\times10^{-2}$& $10.1$ \\
          & 120\% &$4.5\times10^{-2}$& $9.4$ \\\midrule
    95\%-percentile  & 80\% & $6.6\times10^{-2}$ & $15.7$ \\
          & 100\% & $1.2\times10^{-2}$ & $17.9$ \\
          & 120\% &$2.0\times10^{-2}$& $16.2$ \\\midrule
\end{tabular}
\caption{Hedging errors for static and dynamic hedging strategy for a $1Y\times5Y$ receiver swaption, based on $10^4$ MC paths. $S_{1Y\times5Y}\approx0.0305$.}
\label{table: hedging errors}
\end{table}

\begin{figure}
\centering
\begin{subfigure}{.5\textwidth}
  \centering
  \includegraphics[width=0.95\linewidth]{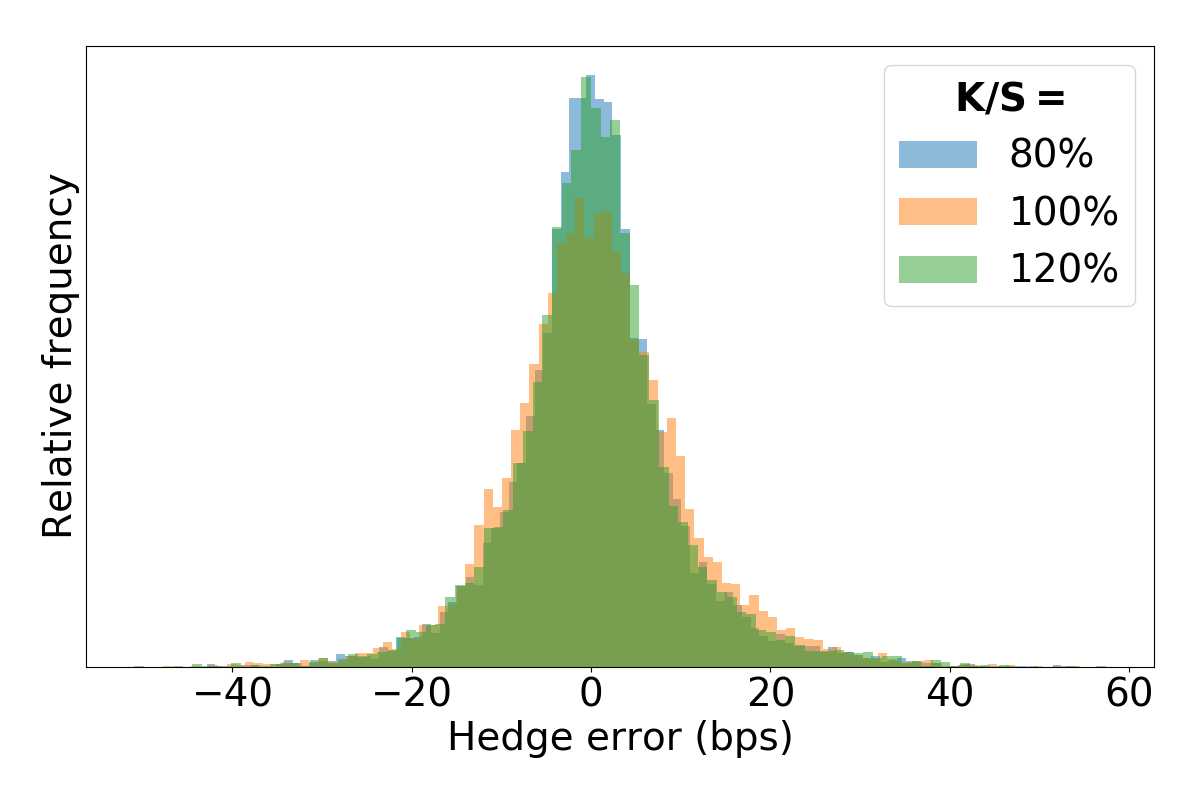}
  \caption{Hedge error dynamic strategy}
  \label{fig: HE dynamic}
\end{subfigure}%
\begin{subfigure}{.5\textwidth}
  \centering
  \includegraphics[width=0.95\linewidth]{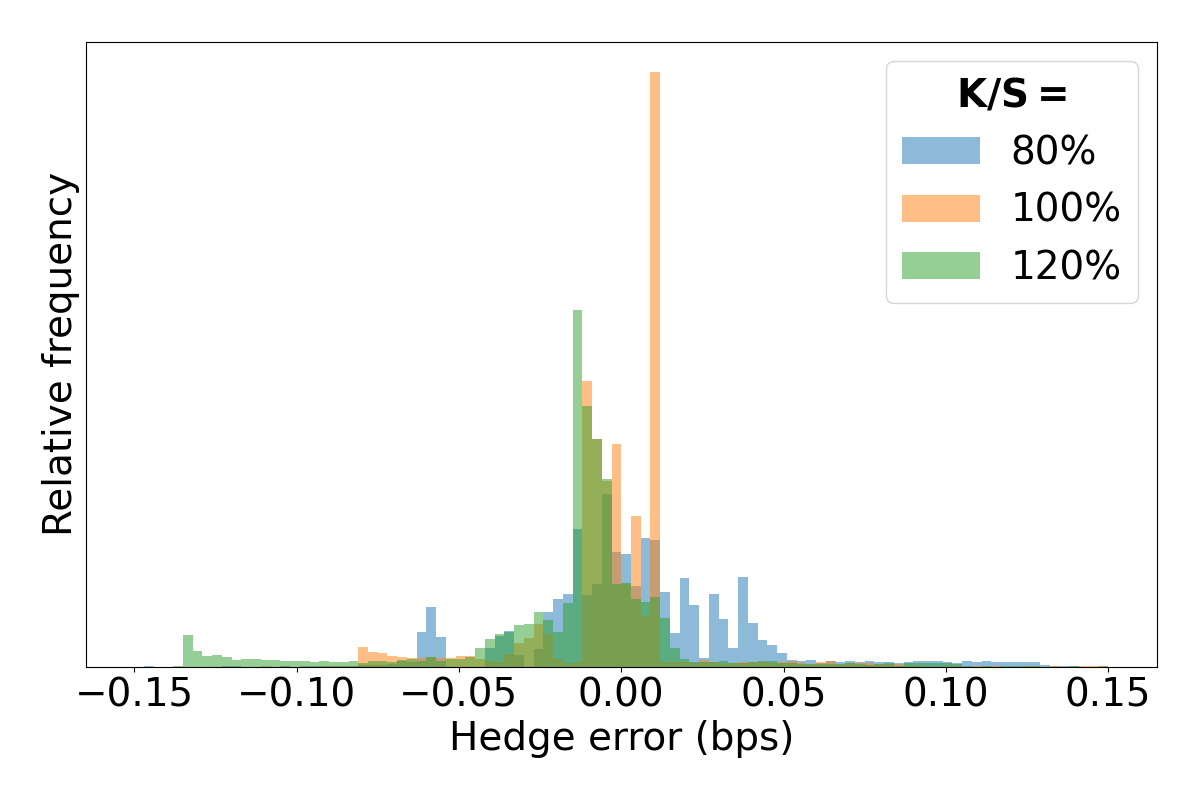}
  \caption{Hedge error static strategy}
  \label{fig: HE static}
\end{subfigure}
\caption{Hedge error distribution for a $1Y\times5Y$ receiver swaption, based on $10^4$ MC paths. $S_{1Y\times5Y}\approx0.0305$.}
\label{fig: Hedge error results}
\end{figure}

\subsubsection{2-factor Bermudan swaption}
Here we demonstrate the performance of the semi-static hedge for a $1Y\times5Y$ receiver Bermudan swaption under a 2-factor model. We compare the accuracy of the hedging strategy utilizing a locally-connected versus a fully-connected neural network. In the former, the replication portfolio consists of zero-coupon bonds and zero-coupon bond options. In the latter, the Bermudan is replicated with options written on hypothetical assets with a pay-off equal to the $\log$ of a zero-coupon bond (see \cref{sec: d-dim NN design}). The model set-up is similar to that in \cref{sec: num ex 2fac berm}, using the same set of parameters reported in \cref{table:2F parameters}. Both network are composed with 64 hidden nodes and optimized using 20,000 training-points generated through Monte Carlo sampling. The portfolio is set-up at time-zero and updated at each monitor-date of the Bermudan until it is either exercised or expired. We assume that the holder of the Bermudan swaption follows the exercise strategy implied by the algorithm, i.e. the option is exercised as soon as $\tilde C_m\left(T_m\right)\leq h_m\left(\mathbf{x}_{T_m}\right)$. When a monitor date $T_m$ is reached, the replication portfolio matures with a pay-off equal to $G_m\left(z_m(T_m)\right)$. In case the Bermudan is continued, the price to set up a new replication portfolio is given by $\tilde V\left(T_m\right)=B(T_m)\mathbb{E^Q}\left[\frac{G_{m+1}\left(z_{m+1}\right)}{B\left(T_{m+1}\right)}\Big|\mathcal{F}_{T_m}\right]$, which contributes $G_m\left(z_m(T_m)\right)-\tilde V\left(T_m\right)$ to the hedging error. In case the Bermudan is exercised, the holder will claim $\tilde V\left(T_m\right)=h_m\left(\mathbf{x}_{T_m}\right)$, which also contributes $G_m\left(z_m(T_m)\right)-\tilde V\left(T_m\right)$ to the hedging error. The total error of the semi-static hedge (HE) is therefore computed as
\begin{equation*} \begin{aligned}
    \text{HE} := \sum_{m=0}^{M-1}\left(G_m\left(z_m(T_m)\right)-\tilde V\left(T_m\right)\right)\mathbbm{1}_{\left\{\tilde\tau\leq T_m\right\}}
\end{aligned} \end{equation*}
where $\tilde V\left(T_m\right):=\max\left\{B(T_m)\mathbb{E^Q}\left[\frac{G_{m+1}\left(z_{m+1}\right)}{B\left(T_{m+1}\right)}\Big|\mathcal{F}_{T_m}\right], h_m\left(\mathbf{x}_{T_m}\right)\right\}$ denotes the direct estimator at date $T_m$ and $\tilde\tau$ denotes the stopping time, as defined in \ref{eq: stopping time}.

The performance of the strategies related to locally- and fully-connected neural networks is reported in \cref{table: hedging errors Bermudan}. The results are based on 10,000 risk-neutral Monte Carlo paths and reported in basis points of the notional amount. The empirical distribution of the hedging error is shown in \cref{fig: Hedge error bermudan}. We observe that both approaches yield an accuracy in the same order of magnitude, although the locally-connected case slightly outperforms the fully-connected case. This is in line with expectations, as the fitting performance of the locally-connected networks is generally higher. For similar reasons as the 1-factor case, the hedging experiments give rise to occasional outliers in terms of accuracy. These outliers can be in the order of several dozens of basis points. Again, the impact of outliers can be reduced by broadening the regression-domain.
\begin{table}
\centering
\begin{tabular}{cccc} 
    \toprule
    \textbf{Hedge error (bps)} & \textbf{K/S} & \textbf{Loc. conn. NN} & \textbf{Fully conn. NN} \\
    \midrule
    Mean  & 80\% &$3.2\times10^{-2}$ & $2.1\times 10^{-2}$\\
          & 100\% &$7.9\times10^{-2}$& $-5.5\times10^{-2}$ \\
          & 120\% &$-9.4\times10^{-2}$& $4.5\times10^{-2}$ \\\midrule
    St. dev.  & 80\% & $0.45$ & $0.55$\\
          & 100\% &$0.38$& $0.48$ \\
          & 120\% &$0.37$& $0.67$ \\\midrule
    95\%-percentile  & 80\% & $0.66$ & $0.69$ \\
          & 100\% & $0.56$ & $0.85$ \\
          & 120\% &$0.72$& $0.76$ \\\midrule
\end{tabular}
\caption{Hedging errors of the semi-static hedging strategy for a $1Y\times5Y$ receiver Bermudan swaption, based on $10^4$ MC paths. $S_{1Y\times5Y}\approx0.0305$.}
\label{table: hedging errors Bermudan}
\end{table}

\begin{figure}
\centering
\begin{subfigure}{.5\textwidth}
  \centering
  \includegraphics[width=0.95\linewidth]{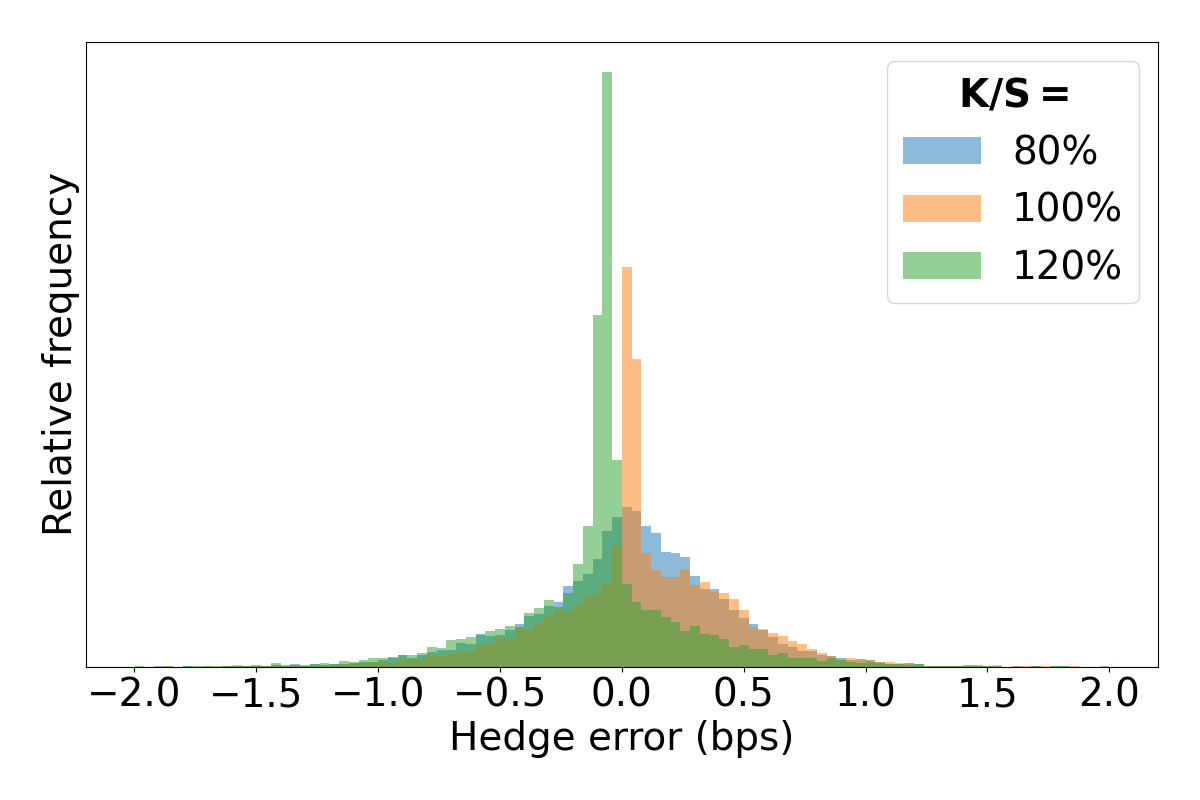}
  \caption{Hedge error locally-connected NN}
  \label{fig: HE local}
\end{subfigure}%
\begin{subfigure}{.5\textwidth}
  \centering
  \includegraphics[width=0.95\linewidth]{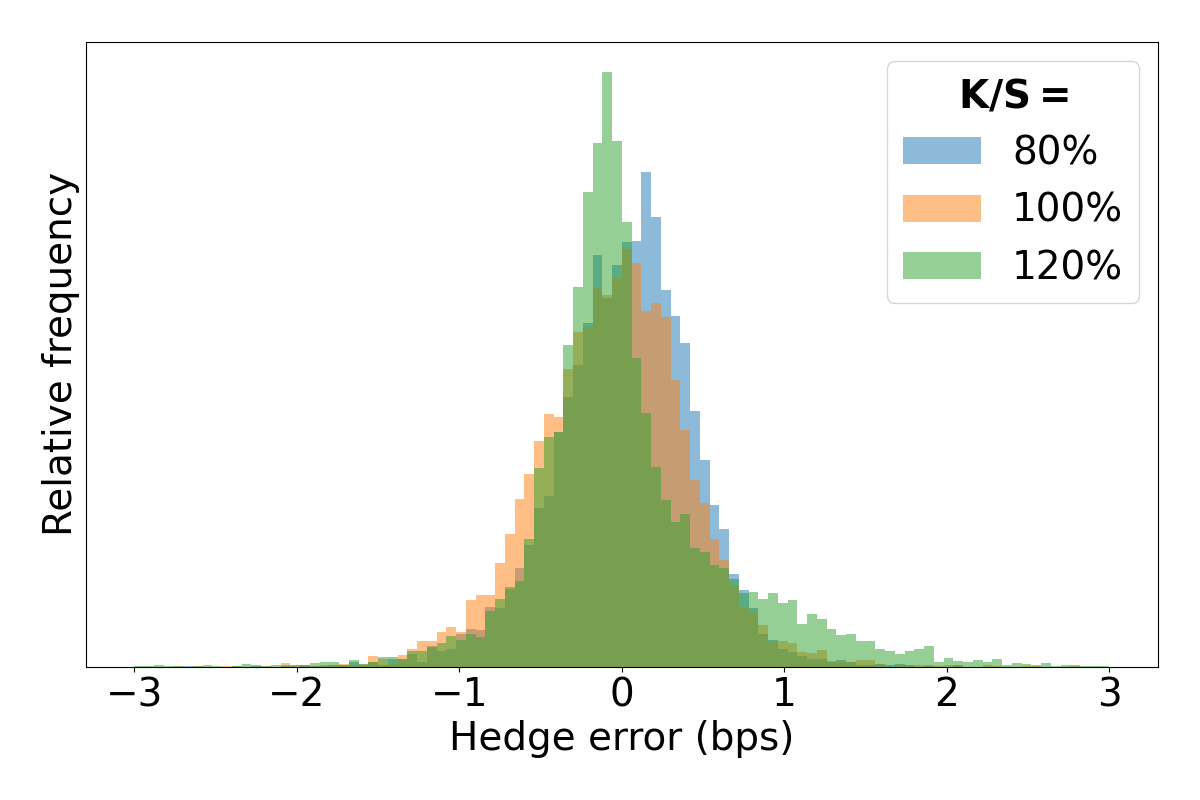}
  \caption{Hedge error fully-connected NN}
  \label{fig: HE connected}
\end{subfigure}
\caption{Hedge error distribution for a $1Y\times5Y$ receiver Bermudan swaption, based on $10^4$ MC paths. $S_{1Y\times5Y}\approx0.0305$.}
\label{fig: Hedge error bermudan}
\end{figure}

\section{Conclusion\label{sec: conclusion}}
This paper has featured a semi-static hedging algorithm for callable interest rate derivatives under an affine, multi-factor term-structure model. As such we present a novel contribution to the literature, where little work is done concerning the static replication of path-dependent interest rate derivatives. Taking Bermudan swaptions as example, we have demonstrated that these products can be replicated with an options portfolio written on a basket of discount bonds. The static portfolio composition is obtained by regressing the target option's value using a shallow, artificial neural network. For established regression methods, such as LSM, the choice of basis functions is arbitrary. In our approach, the chosen basis functions are motivated by their representation of the replicating portfolio's pay-off, yielding interpretable neural networks. Apart from a hedging application, the algorithm gives rise to a direct estimator of the contract price. An upper bound and lower bound estimate to this price valuation can be computed at a minimal additional computational cost. The LSM approach requires expensive nested simulations to obtain such bounds, which can be avoided here. An important contribution is that we prove strict theoretical error margins to the earlier mentioned price statistics. These margins are based on simple metrics, such as the mean absolute error of the neural network fit, which is easily measured. The performance of the semi-static hedge is justified by established benchmarks. Illustrative examples focussing on convergence, hedging and pricing are provided for a 1-factor and a 2-factor model, which are popular amongst practitioners.

As a look-out for further research, we consider applying the algorithm to the computation of counterparty credit risk-measures and various value-adjustments. These metrics typically rely on generating a forward value- and sensitivity-profiles of (exotic) derivative portfolios. We see the semi-static hedging approach combined with the simple error analysis as an effective tool to address the computational challenges associated with these risk measures.

\section*{Disclosure}
The opinions expressed in this work are solely those of the authors and do not represent in any way those of their current and past employers. No potential conflict of interest was reported by the authors.

\appendix

\bibliographystyle{siamplain}
\bibliography{references}

\begin{thebibliography}{10}

\bibitem{Ame2003}
{\sc F.~Ametrano and L.~Ballabio}, {\em Quantlib - a free/open-source library
  for quantitative finance}, 2003, \url{http://quantlib.org/}.

\bibitem{andersen2004primal}
{\sc L.~Andersen and M.~Broadie}, {\em Primal-dual simulation algorithm for
  pricing multidimensional american options}, Management Science, 50 (2004),
  pp.~1222--1234.

\bibitem{becker2020pricing}
{\sc S.~Becker, P.~Cheridito, and A.~Jentzen}, {\em Pricing and hedging
  american-style options with deep learning}, Journal of Risk and Financial
  Management, 13 (2020), p.~158.

\bibitem{beyna2013interest}
{\sc I.~Beyna}, {\em Interest rate derivatives: valuation, calibration and
  sensitivity analysis}, Springer Science \& Business Media, 2013.

\bibitem{bishop1995neural}
{\sc C.~M. Bishop et~al.}, {\em Neural networks for pattern recognition},
  Oxford university press, 1995.

\bibitem{boyle1988lattice}
{\sc P.~P. Boyle}, {\em A lattice framework for option pricing with two state
  variables}, Journal of Financial and Quantitative Analysis,  (1988),
  pp.~1--12.

\bibitem{breeden1978prices}
{\sc D.~T. Breeden and R.~H. Litzenberger}, {\em Prices of state-contingent
  claims implicit in option prices}, Journal of business,  (1978),
  pp.~621--651.

\bibitem{brennan1977valuation}
{\sc M.~J. Brennan and E.~S. Schwartz}, {\em The valuation of american put
  options}, The Journal of Finance, 32 (1977), pp.~449--462.

\bibitem{brigo2007interest}
{\sc D.~Brigo and F.~Mercurio}, {\em Interest rate models-theory and practice:
  with smile, inflation and credit}, Springer Science \& Business Media, 2007.

\bibitem{broadie1996american}
{\sc M.~Broadie and J.~Detemple}, {\em American option valuation: new bounds,
  approximations, and a comparison of existing methods}, The Review of
  Financial Studies, 9 (1996), pp.~1211--1250.

\bibitem{carr1994static}
{\sc P.~Carr and J.~Bowie}, {\em Static simplicity}, Risk, 7 (1994),
  pp.~45--50.

\bibitem{carr1999static}
{\sc P.~Carr, K.~Ellis, and V.~Gupta}, {\em Static hedging of exotic options},
  in Quantitative Analysis In Financial Markets: Collected Papers of the New
  York University Mathematical Finance Seminar, World Scientific, 1999,
  pp.~152--176.

\bibitem{carr2014static}
{\sc P.~Carr and L.~Wu}, {\em Static hedging of standard options}, Journal of
  Financial Econometrics, 12 (2014), pp.~3--46.

\bibitem{carriere1996valuation}
{\sc J.~F. Carriere et~al.}, {\em Valuation of the early-exercise price for
  options using simulations and nonparametric regression}, Insurance:
  mathematics and Economics, 19 (1996), pp.~19--30.

\bibitem{chollet2015keras}
{\sc F.~Chollet et~al.}, {\em Keras}.
\newblock \url{https://keras.io}, 2015.

\bibitem{chung2009static}
{\sc S.-L. Chung and P.-T. Shih}, {\em Static hedging and pricing american
  options}, Journal of Banking \& Finance, 33 (2009), pp.~2140--2149.

\bibitem{cox1979option}
{\sc J.~C. Cox, S.~A. Ross, and M.~Rubinstein}, {\em Option pricing: A
  simplified approach}, Journal of financial Economics, 7 (1979), pp.~229--263.

\bibitem{dai2000specification}
{\sc Q.~Dai and K.~J. Singleton}, {\em Specification analysis of affine term
  structure models}, The journal of finance, 55 (2000), pp.~1943--1978.

\bibitem{derman1995static}
{\sc E.~Derman, D.~Ergener, and I.~Kani}, {\em Static options replication},
  Journal of Derivatives, 2 (1995).

\bibitem{duffie1996yield}
{\sc D.~Duffie and R.~Kan}, {\em A yield-factor model of interest rates},
  Mathematical finance, 6 (1996), pp.~379--406.

\bibitem{feng2016efficient}
{\sc Q.~Feng, S.~Jain, P.~Karlsson, D.~Kandhai, and C.~W. Oosterlee}, {\em
  Efficient computation of exposure profiles on real-world and risk-neutral
  scenarios for bermudan swaptions}, Available at SSRN 2790874,  (2016).

\bibitem{filipovic2009term}
{\sc D.~Filipovic}, {\em Term-Structure Models. A Graduate Course.}, Springer,
  2009.

\bibitem{geman1995changes}
{\sc H.~Geman, N.~El~Karoui, and J.-C. Rochet}, {\em Changes of numeraire,
  changes of probability measure and option pricing}, Journal of Applied
  probability,  (1995), pp.~443--458.

\bibitem{glasserman2013monte}
{\sc P.~Glasserman}, {\em Monte Carlo methods in financial engineering},
  vol.~53, Springer Science \& Business Media, 2013.

\bibitem{glasserman2004simulation}
{\sc P.~Glasserman and B.~Yu}, {\em Simulation for american options: Regression
  now or regression later?}, in Monte Carlo and Quasi-Monte Carlo Methods 2002,
  Springer, 2004, pp.~213--226.

\bibitem{goodfellow2016deep}
{\sc I.~Goodfellow, Y.~Bengio, A.~Courville, and Y.~Bengio}, {\em Deep
  learning}, vol.~1, MIT press Cambridge, 2016.

\bibitem{gregory2015xva}
{\sc J.~Gregory}, {\em The xVA Challenge: counterparty credit risk, funding,
  collateral and capital}, John Wiley \& Sons, 2015.

\bibitem{haentjens2015adi}
{\sc T.~Haentjens and K.~J. in’t Hout}, {\em Adi schemes for pricing american
  options under the heston model}, Applied Mathematical Finance, 22 (2015),
  pp.~207--237.

\bibitem{hagan2005convexity}
{\sc P.~S. Hagan}, {\em Convexity conundrums: Pricing cms swaps, caps, and
  floors}, The Best of Wilmott,  (2005), p.~305.

\bibitem{harrison1979martingales}
{\sc J.~M. Harrison and D.~M. Kreps}, {\em Martingales and arbitrage in
  multiperiod securities markets}, Journal of Economic theory, 20 (1979),
  pp.~381--408.

\bibitem{haugh2004pricing}
{\sc M.~B. Haugh and L.~Kogan}, {\em Pricing american options: a duality
  approach}, Operations Research, 52 (2004), pp.~258--270.

\bibitem{henrard2003explicit}
{\sc M.~Henrard}, {\em Explicit bond option formula in heath--jarrow--morton
  one factor model}, International Journal of Theoretical and Applied Finance,
  6 (2003), pp.~57--72.

\bibitem{hornik1989multilayer}
{\sc K.~Hornik, M.~Stinchcombe, H.~White, et~al.}, {\em Multilayer feedforward
  networks are universal approximators.}, Neural networks, 2 (1989),
  pp.~359--366.

\bibitem{jamshidian1989exact}
{\sc F.~Jamshidian}, {\em An exact bond option formula}, The journal of
  Finance, 44 (1989), pp.~205--209.

\bibitem{joshi2016least}
{\sc M.~Joshi and O.~K. Kwon}, {\em Least squares monte carlo credit value
  adjustment with small and unidirectional bias}, International Journal of
  Theoretical and Applied Finance, 19 (2016), p.~1650048.

\bibitem{kingma2014adam}
{\sc D.~P. Kingma and J.~Ba}, {\em Adam: A method for stochastic optimization},
  arXiv preprint arXiv:1412.6980,  (2014).

\bibitem{kloeden2013numerical}
{\sc P.~E. Kloeden and E.~Platen}, {\em Numerical solution of stochastic
  differential equations}, vol.~23, Springer Science \& Business Media, 2013.

\bibitem{kusuoka2015least}
{\sc S.~Kusuoka and Y.~Morimoto}, {\em Least square regression methods for
  bermudan derivatives and systems of functions}, in Advances in Mathematical
  Economics Volume 19, Springer, 2015, pp.~57--89.

\bibitem{lokeshwar2019neural}
{\sc V.~Lokeshwar, V.~Bhardawaj, and S.~Jain}, {\em Neural network for pricing
  and universal static hedging of contingent claims}, Available at SSRN
  3491209,  (2019).

\bibitem{longstaff2001valuing}
{\sc F.~A. Longstaff and E.~S. Schwartz}, {\em Valuing american options by
  simulation: a simple least-squares approach}, The review of financial
  studies, 14 (2001), pp.~113--147.

\bibitem{mitchell1999finite}
{\sc A.~R. Mitchell and D.~F. Griffiths}, {\em Finite difference and related
  methods for differential equations}, Wiley, 1999.

\bibitem{musiela2005martingale}
{\sc M.~Musiela and M.~Rutkowski}, {\em Martingale methods in financial
  modelling}, Springer Finance, 2005.

\bibitem{pelsser2003pricing}
{\sc A.~Pelsser}, {\em Pricing and hedging guaranteed annuity options via
  static option replication}, Insurance: Mathematics and Economics, 33 (2003),
  pp.~283--296.

\bibitem{rogers2002monte}
{\sc L.~C. Rogers}, {\em Monte carlo valuation of american options},
  Mathematical Finance, 12 (2002), pp.~271--286.

\bibitem{shreve2004stochastic}
{\sc S.~E. Shreve}, {\em Stochastic calculus for finance II: Continuous-time
  models}, vol.~11, Springer Science \& Business Media, 2004.

\bibitem{tsitsiklis1999optimal}
{\sc J.~N. Tsitsiklis and B.~Van~Roy}, {\em Optimal stopping of markov
  processes: Hilbert space theory, approximation algorithms, and an application
  to pricing high-dimensional financial derivatives}, IEEE Transactions on
  Automatic Control, 44 (1999), pp.~1840--1851.

\bibitem{xiu2010numerical}
{\sc D.~Xiu}, {\em Numerical methods for stochastic computations: a spectral
  method approach}, Princeton university press, 2010.

\end{thebibliography}

\section{Evaluation of the conditional expectation\label{sec: eval cond exp}}
In this section we will explicitly compute the conditional expectations related to the continuation values. We will distinguish two approaches associated with the two proposed network structures, i.e. the locally connected case (suggestion 1) and the fully connected case (suggestion 2). 

For the ease of computation we will use a simplified, yet equivalent representation of the risk-factor dynamics discussed in \cref{sec: model formulation}. This concerns a linear shift of the canonical representation of the latent factors as presented in \cite{dai2000specification}. We write $\mathbf{x}_t:=\left(x_1(t),\ldots,x_n(t)\right)^\top$, where each component $x_i$ denotes a mean-reverting zero-mean process. The risk-neutral dynamics are assumed to satisfy
\begin{equation} \begin{aligned}
    d\begin{pmatrix}x_1(t)\\\vdots\\x_d(t)\end{pmatrix}=-\begin{pmatrix}a_1(t)x_1(t)\\\vdots\\a_d(t)x_d(t)\end{pmatrix}dt+\begin{pmatrix}\sigma_{11}(t)&\hdots&\sigma_{1d}(t)\\\vdots&\ddots&\vdots\\\sigma_{d1}(t)&\hdots&\sigma_{dd}(t)\end{pmatrix}d\mathbf{W}(t), \quad \begin{pmatrix}x_1(0)\\\vdots\\x_d(0)\end{pmatrix}=\begin{pmatrix}0\\\vdots\\0\end{pmatrix}\label{eqn: canonical repr}
\end{aligned} \end{equation}
where $\mathbf{W}$ denotes a standard $d$-dimensional Brownian motion with independent entries. By setting $\tilde\sigma_i(t):=\sqrt{\sum_{j=1}^d\sigma_{ij}^2(t)}$, the process above can be rewritten in terms of one-dimensional It\^o processes \cite{shreve2004stochastic} of the form
\begin{equation} \begin{aligned}
    dx_i(t)&=-a_i(t)x_i(t)dt+\tilde\sigma_i(t)d\tilde W_i(t),\qquad i=1,\ldots,d\label{eqn:SDE1}
\end{aligned} \end{equation}
where $\tilde W_1,\ldots,\tilde W_d$ denote a set of one-dimensional, correlated Brownian motions under the measure $\mathbb{Q}$. The instantaneous correlation is denoted by $\rho_{ij}$, such that $d\left<\tilde W_i,\tilde W_j\right>_t=\rho_{ij}(t)dt$.

\subsection{The continuation value with locally connected NN\label{sec: eval local NN}}
We consider the network $G_m(\cdot)$, which is trained to approximate $\tilde V(T_m)$. Let $t\in[T_{m-1},T_m)$. In order to obtain $\tilde V(t)$, we need to evaluate $\mathbb{E^Q}\left[e^{-\int_t^{T_m}r(u)du}G_m(\mathbf{x}_{T_m})\Big|\mathcal{F}_t\right]$. As $G_m(\cdot)$ represents the linear combination of the outcome of $q$ hidden nodes, we will focus on the conditional expectation of hidden node $i\in\{1,\ldots,q\}$. Our aim is then to compute the following
\begin{equation*} \begin{aligned}
    H_i(t):=\mathbb{E^Q}\left[e^{-\int_t^{T_m}r(u)du}\varphi(\mathbf{w}_i^\top \mathbf{P}(T_m)+b_i)\Big|\mathcal{F}_t\right]
\end{aligned} \end{equation*}
The map $\varphi:\mathbb{R}\to \mathbb{R}$ denotes the ReLU function defined as $\varphi(x)=\max\{x,0\}$. The weight-vector $\mathbf{w}_i$ (corresponding to hidden node $i$) and $\mathbf{P}(T_m)$ are defined as
\begin{equation*} \begin{aligned}
    \mathbf{w}_i=\begin{pmatrix}w_1^i\\ \vdots\\w_d^i\end{pmatrix},\qquad\mathbf{P}(T_m)=\begin{pmatrix} P(T_m,T_{m}+\delta_1)\\\vdots\\P(T_m,T_{m}+\delta_d)\end{pmatrix}
\end{aligned} \end{equation*}
with $T_m<T_{m}+\delta_1<\ldots<T_{m}+\delta_d\leq T_M$. Recall that as a characteristic of the affine term-structure model, the random variable $P(t,T)$ can be expressed as
\begin{equation*} \begin{aligned}
    P(t,T)=e^{A(t,T)-\sum_{i=1}^dB_i(t,T)x_i(t)}
\end{aligned} \end{equation*}
for deterministic functions $A$ and $B_i$, which are available in closed-form (see \cite{brigo2007interest}). By the structure of the network, the weight-vector is constraint to have only a single non-zero entry, which we will denote to have index $k$. Therefore we can rewrite
\begin{equation*} \begin{aligned}
    H_i(t)&=\mathbb{E^Q}\left[e^{-\int_t^{T_m}r(u)du}\max\left\{w^k_i P(T_m,T_m+\delta_k)+b_i,\;0\right\}\Big|\mathcal{F}_t\right]
\end{aligned} \end{equation*}
As we argued before, if $w_i^k$ and $b_i$ are both non-negative, $H_i(t)$ denotes the value of a forward contract. In that case we have
\begin{equation*} \begin{aligned}
    H_i(t)&=\mathbb{E^Q}\left[e^{-\int_t^{T_m}r(u)du}\left(w^k_i P(T_m,T_m+\delta_k)+b_i\right)\Big|\mathcal{F}_t\right]\\
    &=w^k_i\mathbb{E^Q}\left[e^{-\int_t^{T_m}r(u)du} \mathbb{E^Q}\left[e^{-\int_{T_m}^{T_m+\delta_k}r(u)du}\Big|\mathcal{F}_{T_m}\right]\bigg|\mathcal{F}_t\right]+b_i\mathbb{E^Q}\left[e^{-\int_t^{T_m}r(u)du}\Big|\mathcal{F}_t\right]\\
    &=w^k_iP(t,T_m+\delta_k)+b_iP(t,T_m)
\end{aligned} \end{equation*}

If on the other hand $b_i<0<w^k_i$ or $w^k_i<0<b_i$, we are dealing with a European call or put option respectively. Closed-form expressions for European bond options are available based on Black's formula and have been treated extensively in the literature; see for example \cite{musiela2005martingale}, \cite{filipovic2009term} or \cite{brigo2007interest}. In our case we have
\begin{equation*} \begin{aligned}
    H_i(t)=\begin{cases}
    w_i^kP(t,T_m+\delta_k)\Phi\left(d_+\right)+b_iP(t,T_m)\Phi\left(d_-\right)&\text{if }b_i<0<w^k_i\\
    -b_iP(t,T_m)\Phi\left(-d_-\right)-w_i^kP(t,T_m+\delta_k)\Phi\left(-d_+\right)&\text{if }w^k_i<0<b_i
    \end{cases}
\end{aligned} \end{equation*}
where $\Phi$ denotes the CDF of a standard normal distribution and we define
\begin{equation*} \begin{aligned}
    &d_\pm:=\frac{\log\left(-\frac{w_i^kP(t,T_m+\delta_k)}{b_iP(t,T_m)}\right)\pm\frac{1}{2}\Sigma(t,T_m)}{\sqrt{\Sigma(t,T_m)}}
\end{aligned} \end{equation*}
and
\begin{equation*} \begin{aligned}
    \Sigma(t,T_m):=\int_t^{T_m}\left\|\nu(u,T_{m}+\delta_k)-\nu(u,T_m)\right\|^2du
\end{aligned} \end{equation*}
In the expression above, the function $\nu(t,T)\in\mathbb{R}^d$ refers to the instantaneous volatility at time $t$ of a discount bond maturing at $T$. Under the dynamics of \ref{eqn: canonical repr}, $\nu$ is given by
\begin{equation} \begin{aligned}
    \nu(t,T)=\begin{pmatrix}\sum_{i=1}^dB_i(t,T)\sigma_{i1}(t)\\\vdots\\\sum_{i=1}^dB_i(t,T)\sigma_{id}(t)\end{pmatrix} \label{eqn: volatility ZCB}
\end{aligned} \end{equation}

\subsection{The continuation value with fully connected NN\label{sec: eval fully NN}}
Once again, we consider the network $G_m(\cdot)$, focus on the outcome of hidden node $i\in\{1,\ldots,q\}$ and let $t\in[T_{m-1},T_m)$. Now our aim is to evaluate the conditional expectation below, which by a change of num\'eraire argument can be rewritten as
\begin{equation} \begin{aligned}
    &\mathbb{E^Q}\left[e^{-\int_t^{T_m}r(u)du}\varphi(\mathbf{w}_i^\top\log \mathbf{P}(T_m)-b_i)\Big|\mathcal{F}_t\right]\nonumber\\
    &\qquad\qquad=P(t,T_m)\mathbb{E}^{T_m}\left[\left\{\mathbf{w}_i^\top\log \mathbf{P}(T_m)-b_i),\;0\right\}\Big|\mathcal{F}_t\right]\label{eqn:CondExp}
\end{aligned} \end{equation}
where the expectation on the right is taken under the $T_m$-forward measure, taking $P(t,T_m)$ as num\'eraire. The weight-vector $\mathbf{w}_i$ (corresponding to hidden node $i$) and $\log \mathbf{P}(T_m)$ are defined as
\begin{equation*} \begin{aligned}
    \mathbf{w}_i=\begin{pmatrix}w_1^i\\ \vdots\\w_d^i\end{pmatrix},\qquad\log\mathbf{P}(T_m)=\begin{pmatrix}\log P(T_m,T_{m}+\delta_1)\\\vdots\\\log P(T_m,T_{m}+\delta_d)\end{pmatrix}
\end{aligned} \end{equation*}
with $T_m<T_{m}+\delta_1<\ldots<T_{m}+\delta_d\leq T_M$. We set the input-dimension equal to the number of risk-factors (i.e. $d=n$). Therefore we can write
\begin{equation*} \begin{aligned}
    \mathbf{w}_i^\top\log \mathbf{P}(T_m)=&\,\sum_{j=1}^d w_j^i\log P(T_m,T_{m}+\delta_j)\\
    =&\,\sum_{j=1}^d w_j^i A(T_m,T_{m}+\delta_j)-\sum_{j=1}^dw_j^i\sum_{k=1}^dB_k(T_m,T_{m}+\delta_j)x_k(T_m)\\
    =&\,\begin{pmatrix}w_1^i&\hdots&w_d^i\end{pmatrix}\begin{pmatrix} A(T_m,T_{m}+\delta_1)\\\vdots\\ A(T_m,T_{m}+\delta_d)\end{pmatrix}\\
    &\,-\begin{pmatrix}w_1^i&\hdots&w_d^i\end{pmatrix}\begin{pmatrix}B_1(T_m,T_{m}+\delta_1)&\hdots&B_d(T_m,T_{m}+\delta_1)\\\vdots&\ddots&\vdots\\B_1(T_m,T_{m}+\delta_d)&\hdots&B_d(T_m,T_{m}+\delta_d)\end{pmatrix}\begin{pmatrix}x_1(T_m)\\\vdots\\x_d(T_m)\end{pmatrix}\\
    =&\,\mathbf{w}_i^\top \mathbf{A}(T_m)-\mathbf{w}_i^\top \mathbf{B}(T_m)\mathbf{x}_{T_m}
\end{aligned} \end{equation*}
where we implicitly defined 
\begin{equation*} \begin{aligned}
    \mathbf{A}(T_m)&:=\begin{pmatrix}A(T_m,T_{m}+\delta_1)\\\vdots\\ A(T_m,T_{m}+\delta_d)\end{pmatrix},\\
    \mathbf{B}(T_m)&:=\begin{pmatrix}B_1(T_m,T_{m}+\delta_1)&\hdots&B_d(T_m,T_{m}+\delta_1)\\\vdots&\ddots&\vdots\\B_1(T_m,T_{m}+\delta_d)&\hdots&B_d(T_m,T_{m}+\delta_d)\end{pmatrix}
\end{aligned} \end{equation*}
In order to compute the conditional expectation of \ref{eqn:CondExp}, a change of measure is required to obtain the dynamics of $x_1,\ldots,x_n$ under the $T_m-$forward measure. Consider the Radon-Nikodym derivative process \cite{beyna2013interest}, defined by
\begin{equation*} \begin{aligned}
    \frac{d\mathbb{Q}^{T_m}}{d\mathbb{Q}}\bigg|\mathcal{F}_t=\frac{B(t)}{B(T_m)}\frac{P(T_m,T_m)}{P(t,T_m)}=\exp\left\{-\int_t^{T_m}\nu(u,T_m)\cdot d\mathbf{W}(u)-\frac{1}{2}\int_t^{T_m}\left\|\nu(u,T_m)\right\|^2du\right\}
\end{aligned} \end{equation*}
where $\nu$ refers to to the instantaneous volatility of the num\'eraire, given in \ref{eqn: volatility ZCB}. The dynamics of the risk-factors under $\mathbb{Q}^{T_m}$ can be obtained by an application of Girsanov's theorem \cite{musiela2005martingale}. Denote by $\sigma_i(t):=\left(\sigma_{i1}(t),\ldots,\sigma_{id}(t)\right)$ the $i^{th}$ row of the volatility matrix of $\mathbf{x}_t$ and let $\tilde W_i^{T_m}$ be Brownian motions under $\mathbb{Q}^{T_m}$ then
\begin{equation} \begin{aligned}
    dx_i(t)&=-a_i(t)x_i(t)dt-\sigma_i(t)\cdot\nu(t,T_m)dt+\tilde\sigma_i(t)d\tilde W_i^{T_m}(t),\qquad i=1,\ldots,d\label{eqn:SDE2}
\end{aligned} \end{equation}
Let $\Theta_i(t,T_m)=\int_t^{T_m}\sigma_i(s)\cdot\nu(s,T_m)e^{-\int_{s}^{T_m}a_i(u)du}ds$, then the SDE above solves to
\begin{equation} \begin{aligned}
    x_i(T_m)&=x_i(t)e^{-\int_{t}^{T_m}a_i(u)du}-\Theta_i(t,T_m)+\int_t^{T_m}\tilde\sigma_i(s)e^{-\int_{s}^{T_m}a_i(u)du}d\tilde W_i^{T_m}(s),\quad i=1,\ldots,d \label{eqn:riskf}
\end{aligned} \end{equation}
It follows that as a property of the It\^o integral, the risk-factors $\left(x_1(T_m),\dots,x_n(T_m)\right)$ as presented in \ref{eqn:riskf}, conditional on $\mathcal{F}_t$, have a multivariate normal distribution under $\mathbb{Q}^{T_m}$. Their mean vector, and co-variance matrix are respectively given by
\begin{equation*}
\begin{gathered}
    \bm{\mu}:=\begin{pmatrix}\mu_1\\\vdots\\\mu_d\end{pmatrix}:=\begin{pmatrix}\mathbb{E}^{T_m}\left[x_1(T_m)|\mathcal{F}_t\right]\\\vdots\\\mathbb{E}^{T_m}\left[x_d(T_m)|\mathcal{F}_t\right]\end{pmatrix}=\begin{pmatrix}x_1(t)e^{-\int_{t}^{T_m}a_1(u)du}-\Theta_1(t,T_m)\\\vdots\\x_d(t)e^{-\int_{t}^{T_m}a_d(u)du}-\Theta_d(t,T_m)\end{pmatrix}\\
    \mathbf{C}:=\begin{pmatrix}c_{11}&\hdots&c_{1d}\\\vdots&\ddots&\vdots\\c_{d1}&\hdots&c_{dd}\end{pmatrix}:=\begin{pmatrix}\text{Cov}[x_1(T_m),x_1(T_m)|\mathcal{F}_t]&\hdots&\text{Cov}\left[x_1(T_m),x_d(T_m)|\mathcal{F}_t\right]\\\vdots&\ddots&\vdots\\\text{Cov}\left[x_d(T_m),x_1(T_m)|\mathcal{F}_t\right]&\hdots&\text{Cov}[x_d(T_m),x_d(T_m)|\mathcal{F}_t]\end{pmatrix}\\
    c_{ii}=\int_t^{T_m}\tilde\sigma_i^2(s)e^{-2\int_{s}^{T_m}a_i(u)du}ds\qquad \forall_{i\in\{1,\ldots,d\}}\\
    c_{ij}=\int_t^{T_m}\rho(s)\tilde\sigma_i(s)\tilde\sigma_j(s)e^{-\int_{s}^{T_m}\left(a_i(u)+a_j(u)\right)du}ds\qquad\forall_{i\neq j}
\end{gathered}
\end{equation*}
As a result it should be clear the the random variable $Y:=\mathbf{w}_i^\top\log \mathbf{P}(T_m)$ is normally distributed with mean and variance given respectively by
\begin{equation*} \begin{aligned}
    \mu_Y=\mathbf{w}_i^\top \mathbf{A}(T_m)-\mathbf{w}_i^\top \mathbf{B}(T_m)\bm{\mu}
\end{aligned} \end{equation*}
and variance
\begin{equation*} \begin{aligned}
    \sigma_Y^2=\mathbf{w}_i^\top\mathbf{B}(T_m)\mathbf{CB}(T_m)^\top\mathbf{w}_i
\end{aligned} \end{equation*}
As a result we can compute
\begin{equation} \begin{aligned}
    &\mathbb{E^Q}\left[e^{-\int_t^{T_m}r(u)du}\varphi(\mathbf{w}_i^\top\log \mathbf{P}(T_m)-b_i)\Big|\mathcal{F}_t\right]\nonumber=P(t,T_m)\mathbb{E}^{T_m}\left[\max(Y-b_i,0)\big|\mathcal{F}_t\right]
\end{aligned} \end{equation}
where the conditional expectation on the right hand side can be expressed in closed-form following a similar analysis as presented in \cite{musiela2005martingale}. Let $d_i:=\frac{\mu_Y-b_i}{\sigma_Y}$ and denote by $\xi\sim N(0,1)$ a standard normal random variable. Then it follows that
\begin{equation*} \begin{aligned}
    \mathbb{E}^{T_m}\left[\max(Y-b_i,0)\big|\mathcal{F}_t\right]&=\mathbb{E}^{T_m}\left[(Y-b_i)\mathbbm{1}_{\{Y>b_i\}}\big|\mathcal{F}_t\right]\\
    &=\mathbb{E}^{T_m}\left[(Y-\mu_Y)\mathbbm{1}_{\{Y>b_i\}}\right]+\left(\mu_Y-b_i\right)\mathbb{Q}^{T_m}\left[Y>b_i\big|\mathcal{F}_t\right]\\
    &=\sigma_Y\mathbb{E}^{T_m}\left[\frac{Y-\mu_Y}{\sigma_Y}\mathbbm{1}_{\left\{\frac{Y-\mu_Y}{\sigma_Y}>-d_i\right\}}\bigg|\mathcal{F}_t\right]\\
    &\quad+\left(\mu_Y-b_i\right)\mathbb{Q}^{T_m}\left[\frac{Y-\mu_Y}{\sigma_Y}>-d_i\bigg|\mathcal{F}_t\right]\\
    &=\sigma_Y\mathbb{E}\left[-\xi\mathbbm{1}_{\left\{-\xi<d_i\right\}}\right]+\left(\mu_Y-b_i\right)\mathbb{P}\left[\xi<d_i\right]\\
    &=\sigma_Y\phi(d_i)+\left(\mu_Y-b_i\right)\Phi(d_i)
\end{aligned} \end{equation*}
where $\phi$ denotes the standard normal density function and $\Phi$ the standard normal cumulative density function.

\section{Pre-processing the regression-data\label{sec: preprocessing the data}}
A procedure that significantly improves the fitting performance of the neural networks is the normalization of the training-data. Linear rescaling of the input to the optimizer is a common form of data pre-processing \cite{bishop1995neural}. In the case of a multivariate input, the variables might have typical values in different orders of magnitude, even though that does not reflect their relative influence on determining the outcome \cite{bishop1995neural}. Normalizing the scale avoids that the impact of certain input is prioritized over other input. Also, the transfer of the final weights in $G_{m+1}$ to the initialization of $G_m$ is more effective as the target variables are of roughly the same size at each time-step. In the default situation, the average continuation values would change in magnitude and the risk-factor distribution would grow with each passing of a monitor date.

Another argument for pre-processing the input is that large data values, typically induce large weights. Large weights can lead to exploding network outputs in the feed-forward process \cite{goodfellow2016deep}. Furthermore can it cause an unstable optimization of the network, as extreme gradients can be very sensitive to small perturbations in the data \cite{goodfellow2016deep}.

In practice we propose the following rescaling of the data. Denote by
\begin{equation*} \begin{aligned}
    \hat z(T_m):=\left\{\begin{pmatrix}z_1(T_m)\\\vdots\\z_d(T_m)\end{pmatrix}_{1},\ldots,\begin{pmatrix}z_1(T_m)\\\vdots\\z_d(T_m)\end{pmatrix}_{N}\right\},\quad \hat V(T_m):=\left\{\tilde V\left(T_m;x_{T_m}^1\right),\ldots,\tilde V\left(T_m;x_{T_m}^N\right)\right\}
\end{aligned} \end{equation*}
the training points for the in- and output of network $G_m$. Define the standard sample mean and standard deviations as
\begin{equation*} \begin{aligned}
    &\mu_{z_i}(T_m):=\frac{1}{N}\sum_{n=1}^N z_i^n(T_m),\qquad\mu_{V}(T_m):=\frac{1}{N}\sum_{n=1}^N \tilde V\left(T_m;x_{T_m}^n\right)\\
    &\sigma_{z_i}(T_m):=\frac{1}{N-1}\sum_{n=1}^N\left(z_i^n(T_m)-\mu_{z_i}\right)^2,\quad \sigma_{V}(T_m):=\frac{1}{N-1}\sum_{n=1}^N\left(\tilde V\left(T_m;x_{T_m}^n\right)-\mu_V(T_m)\right)^2
\end{aligned} \end{equation*}
We then perform a simple element-wise linear transformation to obtain the scaled data $\hat z^\dagger$ and $\hat V^\dagger$ given by
\begin{equation*} \begin{aligned}
    \hat z_i^\dagger(T_m) := \frac{\hat z_i(T_m)-\mu_{z_i}(T_m)}{\sigma_{z_i}(T_m)},\qquad \hat V^\dagger(T_m) := \frac{\hat V(T_m)}{\sigma_{V}(T_m)}
\end{aligned} \end{equation*}
With the transformations above in mind, it is important to adjust associated composition of the replicating portfolio accordingly. For the two network designs, this has the following implications:
\begin{description}
    \item[The locally connected NN case:] Consider the outcome of the $i^{th}$ hidden node $\nu_i$ and denote the input of the network as $\mathbf{z}$. Then $\nu_i=\varphi\left(w_i^kz_k+b_i\right)$, where $k$ is the index of the only non-zero entry of $\mathbf{w}_i$, the $i^{th}$ row of weight matrix $\mathbf{w}_1$. The transformation $\mathbf{z}\mapsto \frac{\mathbf{z}-\mu_\mathbf{z}}{\sigma_{\mathbf{z}}}$ implies that \begin{equation*} \begin{aligned}
        \nu_i\mapsto\varphi\left(w_i^k\frac{z_k-\mu_{z_k}}{\sigma_{z_k}}+b_i\right)=\varphi\left(\frac{w_i^k}{\sigma_{z_k}}z_k+\left(b_i-\frac{w_i^k\mu_{z_k}}{\sigma_{z_k}}\right)\right)
    \end{aligned} \end{equation*}
    As a consequence, in the analysis of \cref{sec: eval local NN}, the transformations $w_i^k\mapsto\frac{w_i^k}{\sigma_{z_k}}$ and $b_i\mapsto b_i-\frac{w_i^k\mu_{z_k}}{\sigma_{z_k}}$ should be taken into account. Additionally, the transformation $\mathbf{w}_2\mapsto \sigma_V\mathbf{w}_2$ is required to account for the scaling of $\hat V$.
    \item[The fully connected NN case:] Again consider the outcome of the $i^{th}$ hidden node $\nu_i$. This time the transformation $\mathbf{z}\mapsto \frac{\mathbf{z}-\mu_\mathbf{z}}{\sigma_{\mathbf{z}}}$ implies that
    \begin{equation*} \begin{aligned}
        \nu_i\mapsto \varphi\left(\mathbf{w}_i^\top \frac{\mathbf{z}-\mu_\mathbf{z}}{\sigma_{\mathbf{z}}}+b_i\right)=\varphi\left(\sum_{j=1}^d\frac{w_i^j}{\sigma_{z_j}}z_j+\left(b_i-\sum_{j=1}^d\frac{w_i^j\mu_{z_j}}{\sigma_{z_i}}\right)\right)
    \end{aligned} \end{equation*}
    As a consequence, in the analysis of \cref{sec: eval fully NN}, the transformations $\mathbf{w}_i\mapsto$ $\left(\frac{w_i^1}{\sigma_{z_1}},\ldots\right.$ $\left.,\frac{w_i^d}{\sigma_{z_d}}\right)^\top$ and $b_i\mapsto b_i-\sum_{j=1}^d\frac{w_i^j\mu_{z_j}}{\sigma_{z_i}}$ should be taken into account. And, again the transformation $\mathbf{w}_2\mapsto \sigma_V\mathbf{w}_2$ is required to account for the scaling of $\hat V$.
\end{description}

\section{Hyperparameter selection\label{sec: hyper-parameters}}
The accuracy of the neural network fitting procedure is dependent on the choice of several hyperparameters. For the numerical experiments reported in \cref{sec: numerical examples}, the hyperparameters have been selected based on a convergence analysis. We focused on the following:
\begin{itemize}
    \item \textbf{Hidden node count:} see \cref{fig: errors nodes}
    \item \textbf{Size training set:} see \cref{fig: errors tps}
    \item \textbf{Learning-rate:} see \cref{fig: errors lrs}
\end{itemize}
Several numerical experiments indicated that the batch-size did not have significant impact on the fitting accuracy and is therefore fixed at a default of 32. For the convergence analysis of the parameters listed above, we considered a $1Y\times10Y$ receiver Bermudan swaption with a fixed rate of $K=0.03$. Experiments are performed under the two-factor G2++ model, using the model specifications depicted in \cref{table:2F parameters}. The figures show the mean-absolute errors of the neural network fits per monitor date in basis points of the notional.

\begin{figure}
\centering
\begin{subfigure}{.49\textwidth}
  \centering
  \includegraphics[width=0.95\linewidth]{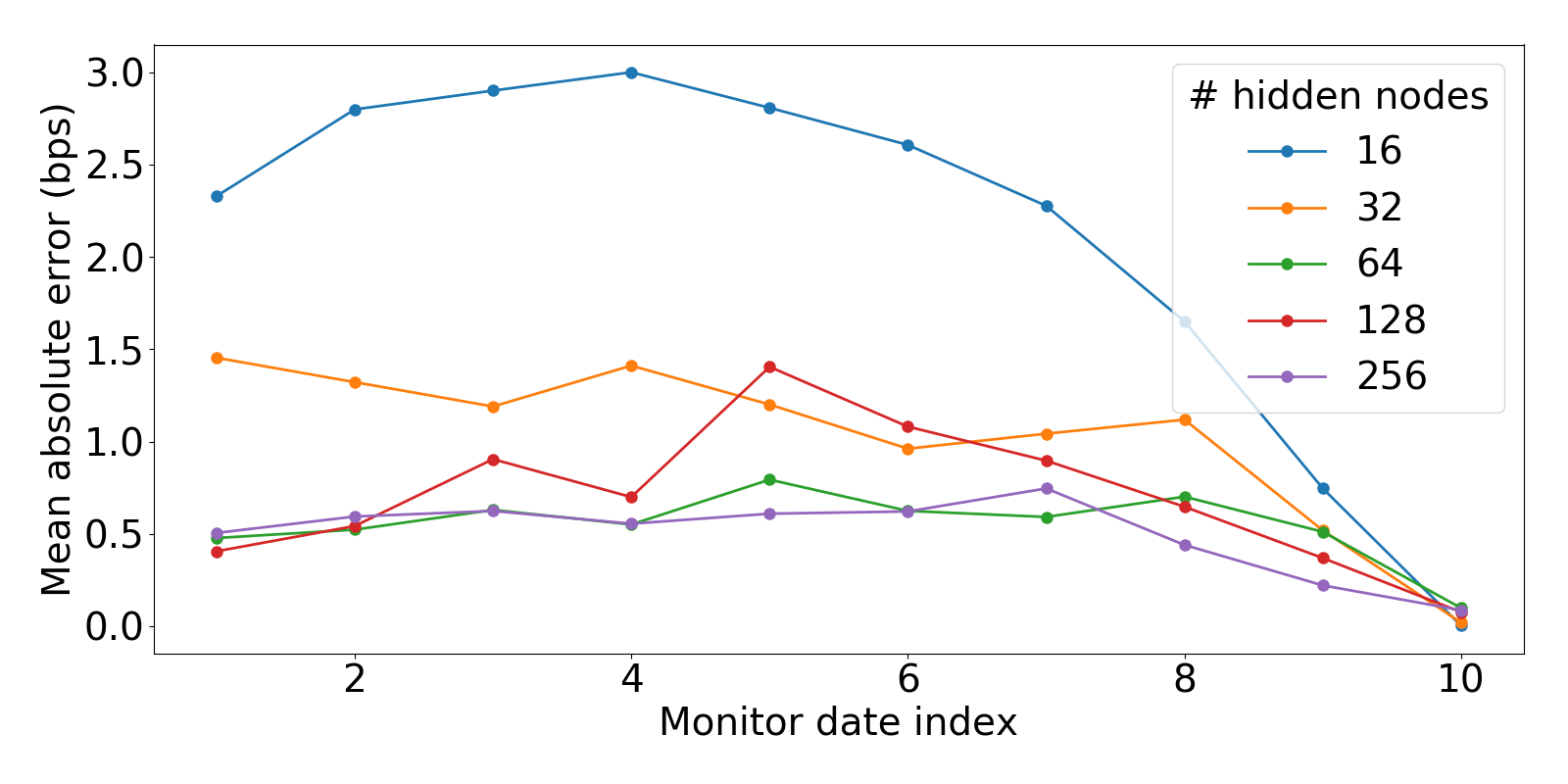}
  \caption{Locally-connected NN}
  \label{fig:loc nodes}
\end{subfigure}
\begin{subfigure}{.49\textwidth}
  \centering
  \includegraphics[width=0.95\linewidth]{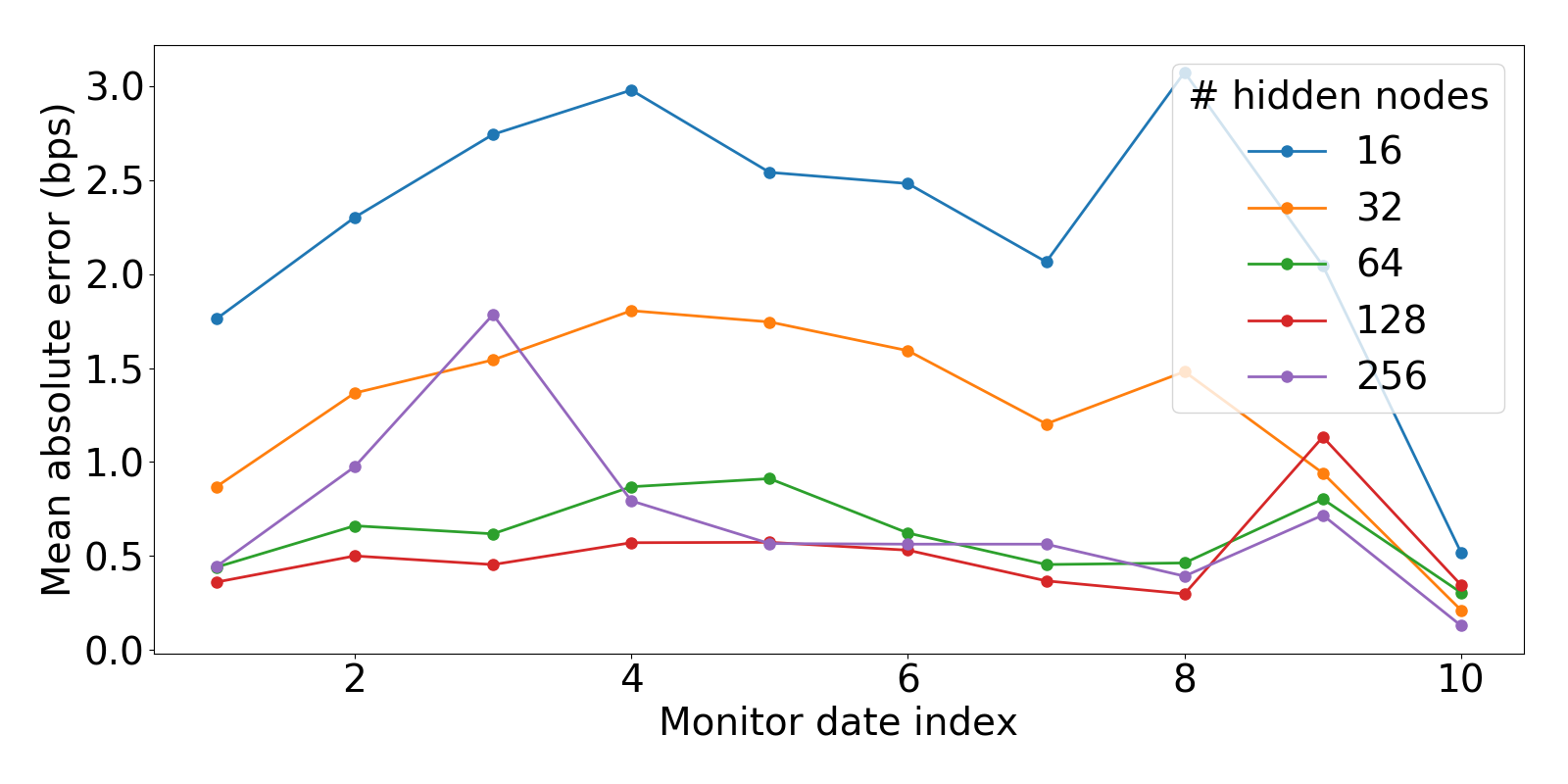}
  \caption{Fully-connected NN}
  \label{fig:con nodes}
\end{subfigure}
\caption{Impact hidden node count: Accuracy of the neural network fit per monitor date under a 2-factor model. \# training points = 5000. Learning-rate = 0.0002.}
\label{fig: errors nodes}
\end{figure}

\begin{figure}
\centering
\begin{subfigure}{.5\textwidth}
  \centering
  \includegraphics[width=0.95\linewidth]{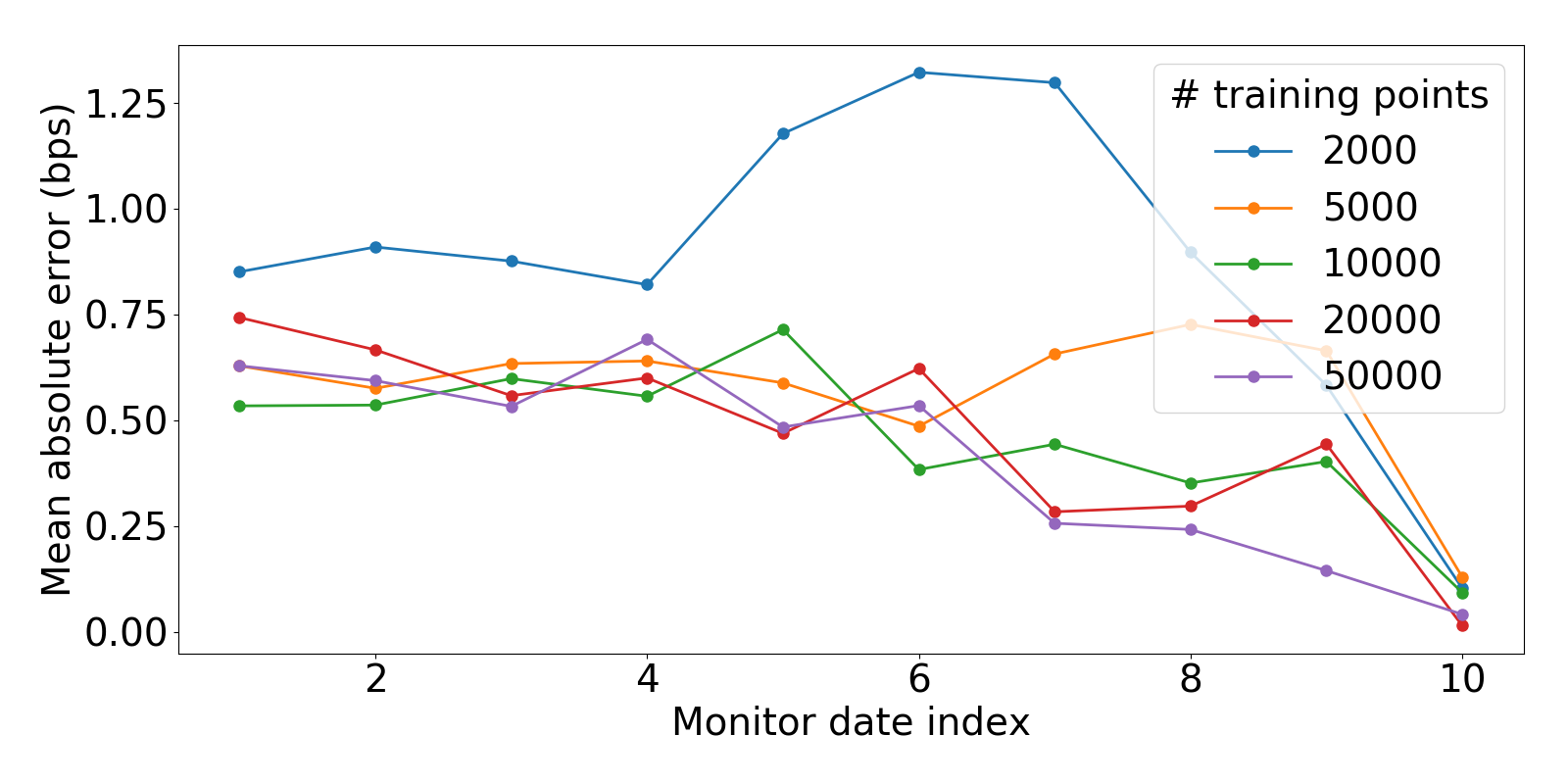}
  \caption{Locally-connected NN}
  \label{fig:loc tp}
\end{subfigure}%
\begin{subfigure}{.5\textwidth}
  \centering
  \includegraphics[width=0.95\linewidth]{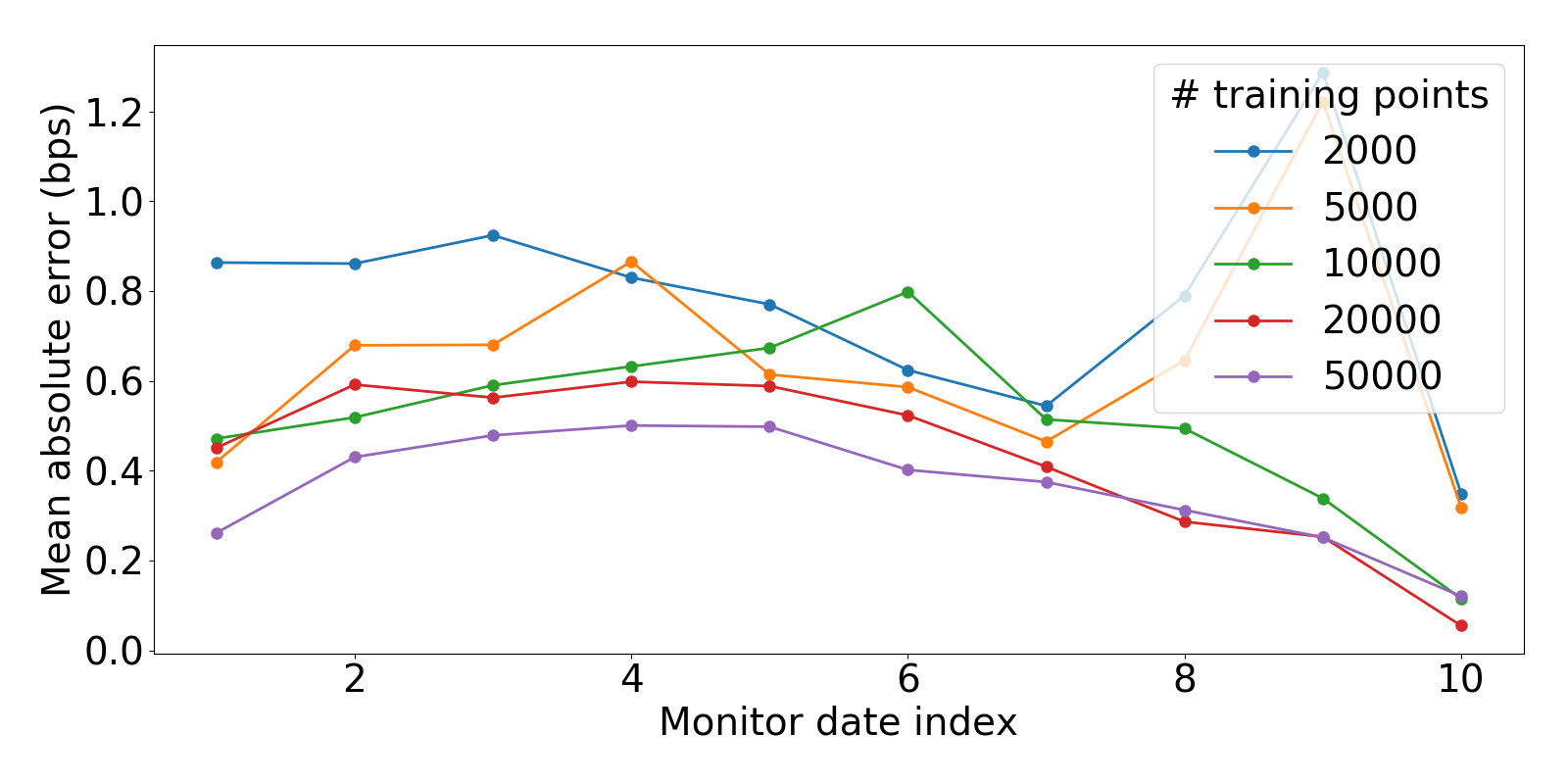}
  \caption{Fully-connected NN}
  \label{fig:con tp}
\end{subfigure}
\caption{Impact size training set: Accuracy of the neural network fit per monitor date under a 2-factor model. \# hidden nodes = 64. Learning-rate = 0.0002.}
\label{fig: errors tps}
\end{figure}

\begin{figure}
\centering
\begin{subfigure}{.5\textwidth}
  \centering
  \includegraphics[width=0.95\linewidth]{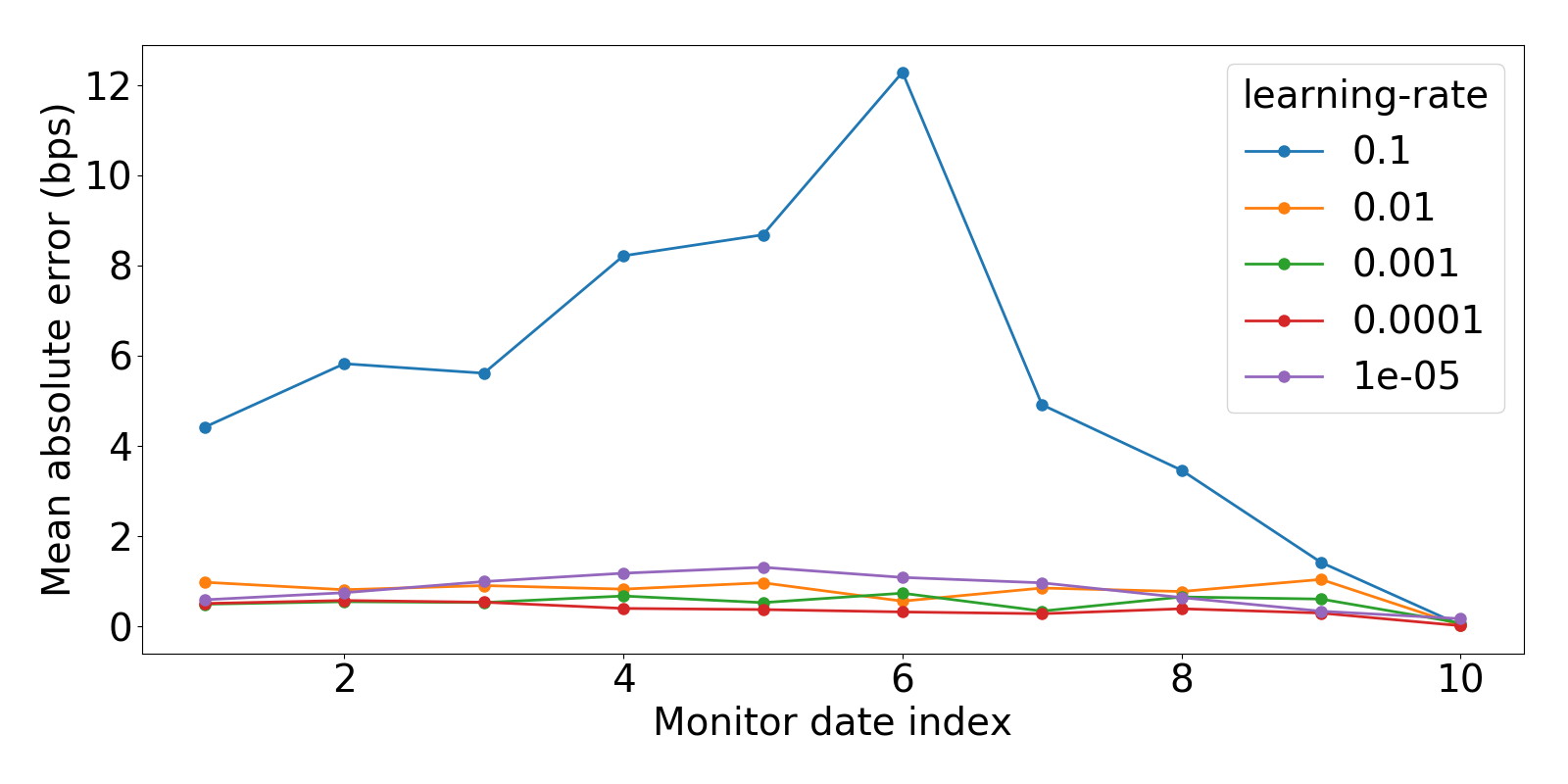}
  \caption{Locally-connected NN}
  \label{fig:loc lr}
\end{subfigure}%
\begin{subfigure}{.5\textwidth}
  \centering
  \includegraphics[width=0.95\linewidth]{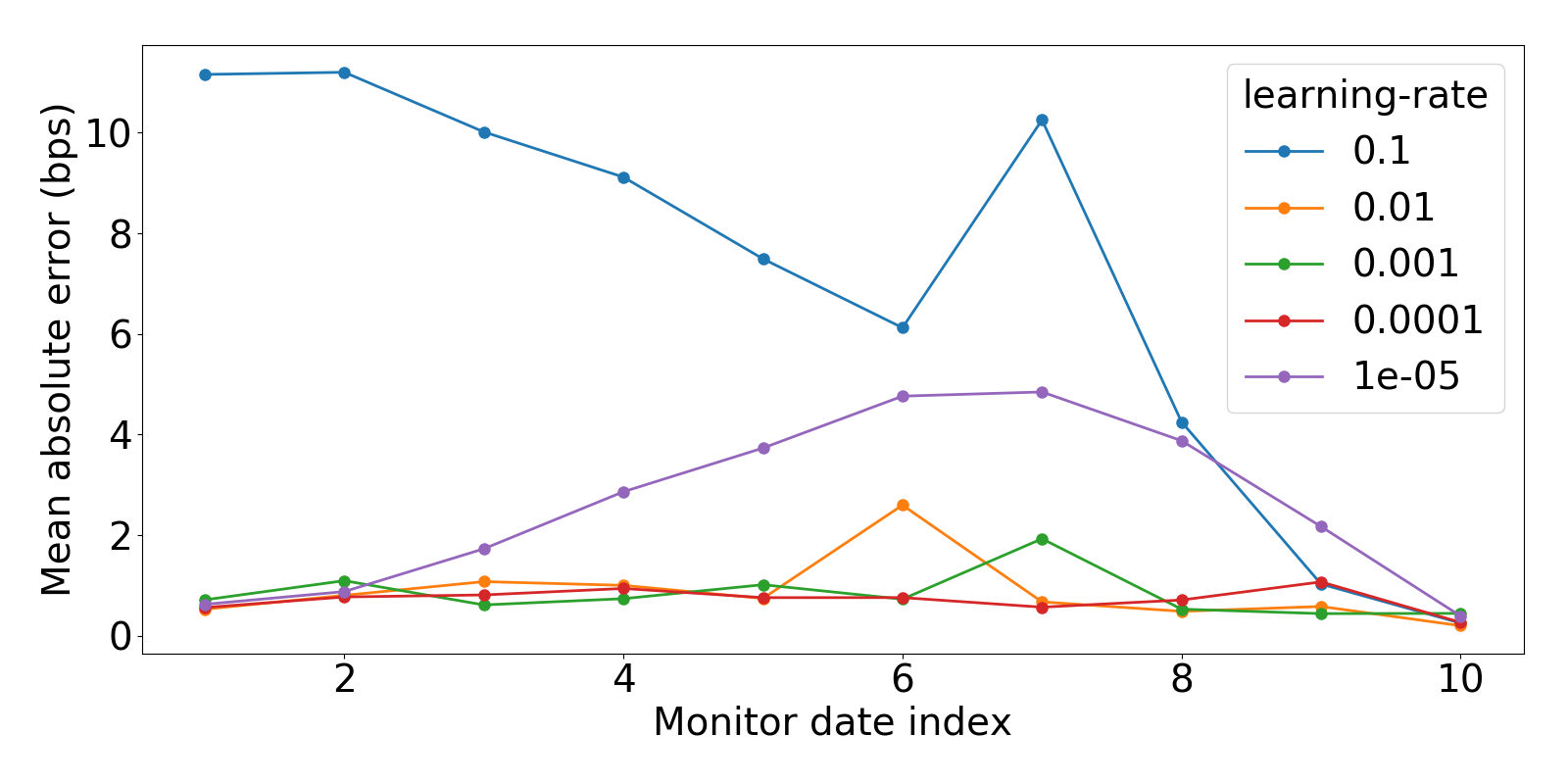}
  \caption{Fully-connected NN}
  \label{fig:con lr}
\end{subfigure}
\caption{Impact learning-rate: Accuracy of the neural network fit per monitor date under a 2-factor model. \# hidden nodes = 64. \# training-points = 10,000.}
\label{fig: errors lrs}
\end{figure}

\section{Proof of \cref*{theorem: MAE error}\label{proof: direct error}}
\begin{proof}
First we fix some notation.
\begin{itemize}
    \item Let $V_m:=V\left(T_m\right)$ denote the true price of the bermudan swaption at $T_m$ conditioned on the fact it is not yet exercised.
    \item Let $\tilde C_m:=B(T_m)\mathbb{E^Q}\left[\frac{G_{m+1}\left(z_{m+1}\right)}{B\left(T_{m+1}\right)}\bigg|\mathcal{F}_{T_m}\right]$ denote the 
    estimator of the continuation value at $T_m$.
    \item Let $\tilde V_m:=\max\left\{\tilde C_m, h_m(\mathbf{x}_{T_m})\right\}$ denote the estimator of $V_m$.
    \item Let $G_m:=G_m\left(z_m\right)$ denote the neural network approximation of $\tilde V_m$.
    \item Let $B_m:=B(T_{m})$ denote the num\'eraire at $T_m$.
    \item Let $h_m:=h_m\left(\mathbf{x}_{T_m}\right)$.
\end{itemize}
Let $T_m\in \{T_0,\ldots,T_{M-1}\}$. We will prove the theorem by induction on $m$. For the base case, note that at time zero we have
\begin{equation} \begin{aligned}
    \left|V(0)-\tilde V(0)\right|&=\left|\mathbb{E^Q}\left[\frac{V_0}{B_0}\bigg|\mathcal{F}_0\right]-\mathbb{E^Q}\left[\frac{G_0}{B_0}\bigg|\mathcal{F}_0\right]\right|\leq\mathbb{E^Q}\left[\left|\frac{V_0-G_0}{B_0}\right|\bigg|\mathcal{F}_0\right]
\end{aligned} \end{equation}
which is induced by Jensen's inequality. For the inductive step, assume that for some $m\in \{0,\ldots,M-1\}$ we have that
\begin{equation} \begin{aligned}
    \left|V(0)-\tilde V(0)\right|< \mathbb{E^Q}\left[\left|\frac{V_m-G_m}{B_m}\right|\bigg|\mathcal{F}_0\right] + m\cdot\varepsilon\label{eqn: proof1}
\end{aligned} \end{equation}
The expectation in (\ref{eqn: proof1}) can be rewritten using the triangular inequality
\begin{equation} \begin{aligned}
    \mathbb{E^Q}\left[\left|\frac{V_m-G_m}{B_m}\right|\bigg|\mathcal{F}_0\right]&=
    \mathbb{E^Q}\left[\left|\frac{V_m-\tilde V_m + \tilde V_m - G_m}{B_m}\right|\Bigg|\mathcal{F}_0\right]\\
    &\leq \mathbb{E^Q}\left[\left|\frac{V_m-\tilde V_m}{B_m}\right|\Bigg|\mathcal{F}_0\right] + \mathbb{E^Q}\left[\left|\frac{\tilde V_m-G_m}{B_m}\right|\Bigg|\mathcal{F}_0\right]
    \label{eqn: proof2}
\end{aligned} \end{equation}
The second term in \ref{eqn: proof2} is by assumption bounded by $\varepsilon$. Note that the first term in \ref{eqn: proof2} can be bounded as
\begin{equation*} \begin{aligned}
    \mathbb{E^Q}\left[\left|\frac{V_m-\tilde V_m}{B_m}\right|\Bigg|\mathcal{F}_0\right]&=\mathbb{E^Q}\left[\left|\frac{\max\left\{C_m, h_m\right\}-\max\left\{\tilde C_m, h_m\right\}}{B_m}\right|\Bigg|\mathcal{F}_0\right]\\
    &\leq \mathbb{E^Q}\left[\left|\frac{C_m-\tilde C_m}{B_m}\right|\Bigg|\mathcal{F}_0\right]\\
    &=\mathbb{E^Q}\left[\left|\mathbb{E^Q}\left[\frac{V_{m+1}}{B_{m+1}}\bigg|\mathcal{F}_{T_m}\right]-\mathbb{E^Q}\left[\frac{G_{m+1}}{B_{m+1}}\bigg|\mathcal{F}_{T_m}\right]\right|\Bigg|\mathcal{F}_0\right]\\
    &\leq\mathbb{E^Q}\left[\mathbb{E^Q}\left[\left|\frac{V_{m+1}-G_{m+1}}{B_{m+1}}\right|\bigg|\mathcal{F}_{T_m}\right]\Bigg|\mathcal{F}_0\right]\\
    &= \mathbb{E^Q}\left[\left|\frac{V_{m+1}-G_{m+1}}{B_{m+1}}\right|\bigg|\mathcal{F}_{0}\right]
\end{aligned} \end{equation*}
It follows that
\begin{equation*} \begin{aligned}
     \left|V(0)-\tilde V(0)\right|< \mathbb{E^Q}\left[\left|\frac{V_{m+1}-G_{m+1}}{B_{m+1}}\right|\bigg|\mathcal{F}_{0}\right] + (m+1)\cdot\varepsilon
\end{aligned} \end{equation*}
For the final step, note that if $m=M-1$ we have
\begin{equation*} \begin{aligned}
    \mathbb{E^Q}\left[\left|\frac{V_m-G_m}{B_m}\right|\bigg|\mathcal{F}_0\right]=\mathbb{E^Q}\left[\left|\frac{\max\left\{h_{M-1},0\right\}-G_{M-1}}{B_{M-1}}\right|\bigg|\mathcal{F}_0\right]< \varepsilon
\end{aligned} \end{equation*}
We conclude by induction on $m$ that $ \left|V(0)-\tilde V(0)\right|<  M\varepsilon$
\end{proof}

\section{Proof of \Cref*{theorem: lower bound}\label{proof: lower bound}}
\begin{proof}
We consider the following three events: $\{\tau = \tilde\tau\}$, $\{\tau < \Tilde\tau\}$ and $\{\tau>\tilde\tau\}$. Note that
\begin{equation*} \begin{aligned}
    V(0)-L(0)&=\mathbb{E^Q}\left[\frac{h_{\tau}(\mathbf{x}_\tau)}{B(\tau)}-\frac{h_{\tilde\tau}(\mathbf{x}_{\tilde\tau})}{B(\tilde\tau)}\bigg|\mathcal{F}_0\right]\\
    &=\mathbb{E^Q}\left[\left(\frac{h_{\tau}(\mathbf{x}_\tau)}{B(\tau)}-\frac{h_{\tilde\tau}(\mathbf{x}_{\tilde\tau})}{B(\tilde\tau)}\right)\mathbbm{1}_{\{\tau=\tilde\tau\}}\bigg|\mathcal{F}_0\right]+\mathbb{E^Q}\left[\left(\frac{h_{\tau}(\mathbf{x}_\tau)}{B(\tau)}-\frac{h_{\tilde\tau}(\mathbf{x}_{\tilde\tau})}{B(\tilde\tau)}\right)\mathbbm{1}_{\{\tau<\tilde\tau\}}\bigg|\mathcal{F}_0\right]\\
    &\quad+\mathbb{E^Q}\left[\left(\frac{h_{\tau}(\mathbf{x}_\tau)}{B(\tau)}-\frac{h_{\tilde\tau}(\mathbf{x}_{\tilde\tau})}{B(\tilde\tau)}\right)\mathbbm{1}_{\{\tau>\tilde\tau\}}\bigg|\mathcal{F}_0\right]\\
    &=E_1+E_2+E_3
\end{aligned} \end{equation*}
We will bound the three terms above one by one.\\
\textbf{Bounding $\mathbf{E_1}$:} Starting with the event $\{\tau = \tilde\tau\}$, we observe that we can write
\begin{equation*} \begin{aligned}
    E_1=\mathbb{E^Q}\left[\left(\frac{h_{\tau}(\mathbf{x}_\tau)}{B(\tau)}-\frac{h_{\tau}(\mathbf{x}_{\tau})}{B(\tau)}\right)\mathbbm{1}_{\{\tau=\tilde\tau\}}\bigg|\mathcal{F}_0\right]=0
\end{aligned} \end{equation*}
\textbf{Bounding $\mathbf{E_2}$:} We continue with the event $\{\tau < \Tilde\tau\}$. For this we will introduce two types of sub-events: $A_m:=\left\{\tau=T_m\land \tilde\tau>T_m\right\}$ and $B_m:=\left\{\tau\leq T_m\land \tilde\tau>T_m\right\}$, where $\land$ denotes the logical AND operator. Also define the difference-process $e_m:=\frac{\tilde V(T_m)}{B(T_m)}-\frac{h_{\tilde\tau}(\mathbf{x}_{\tilde\tau})}{B(\tilde\tau)}$. It should be clear that $\mathbbm{1}_{\{\tau<\tilde\tau\}}=\sum_{m=0}^{M-1}\mathbbm{1}_{A_m}$. Therefore it holds that
\begin{equation*} \begin{aligned}
    E_2=\sum_{m=0}^{M-1}\mathbb{E^Q}\left[\left(\frac{h_{\tau}(\mathbf{x}_\tau)}{B(\tau)}-\frac{h_{\tilde\tau}(\mathbf{x}_{\tilde\tau})}{B(\tilde\tau)}\right)\mathbbm{1}_{A_m}\bigg|\mathcal{F}_0\right]\leq\sum_{m=0}^{M-1}\mathbb{E^Q}\left[e_m\mathbbm{1}_{A_m}\Big|\mathcal{F}_0\right]
\end{aligned} \end{equation*}
where the inequality follows from the fact that the direct estimator has the property $\tilde V(T_m)=\max\{\tilde C_m,h_m\}\geq h_m$. Now we will show by induction that $E_2<(M-1)\varepsilon$. First, observe that $A_0\equiv B_0$. Second, note that for any $m\in\{0,\ldots,M-1\}$ we have that
\begin{equation} \begin{aligned}
    \mathbb{E^Q}\left[e_m\mathbbm{1}_{B_m}\Big|\mathcal{F}_0\right]&=\mathbb{E^Q}\left[\left(\mathbb{E^Q}\left[\frac{G_{m+1}\left(z_{m+1}\right)}{B\left(T_{m+1}\right)}\bigg|\mathcal{F}_{T_m}\right]-\frac{h_{\tilde\tau}(\mathbf{x}_{\tilde\tau})}{B(\tilde\tau)}\right)\mathbbm{1}_{B_m}\bigg|\mathcal{F}_0\right]\\
    &=\mathbb{E^Q}\left[\left(\frac{G_{m+1}\left(z_{m+1}\right)}{B\left(T_{m+1}\right)}-\frac{h_{\tilde\tau}(\mathbf{x}_{\tilde\tau})}{B(\tilde\tau)}\right)\mathbbm{1}_{B_m}\bigg|\mathcal{F}_0\right]\\
    &\leq\mathbb{E^Q}\left[\left|\frac{G_{m+1}\left(z_{m+1}\right)}{B\left(T_{m+1}\right)}-\frac{\tilde V(T_{m+1})}{B(T_{m+1})}\right|\mathbbm{1}_{B_m}\Bigg|\mathcal{F}_0\right]+\mathbb{E^Q}\left[e_{m+1}\mathbbm{1}_{B_m}\Big|\mathcal{F}_0\right]\label{eqn: prf LB 1}
\end{aligned} \end{equation}
The first equality follows from the fact that $\tilde V(T_m)=\tilde C_m$ in the event $\tilde\tau>T_m$. The second equality follows from the tower rule in combination with the fact that $\mathbbm{1}_{B_m}$ is $\mathcal{F}_{T_m}-$measurable. The final inequality follows from an application of the triangle inequality.
The first term in \ref{eqn: prf LB 1} is by assumption bounded by $\varepsilon$. The second term in \ref{eqn: prf LB 1} can be rewritten by observing that $\mathbbm{1}_{B_m}:=\mathbbm{1}_{B_m^1}+\mathbbm{1}_{B_m^2}:=\mathbbm{1}_{\left\{\tau\leq T_m\land \tilde\tau=T_{m+1}\right\}}+\mathbbm{1}_{\left\{\tau\leq T_{m}\land \tilde\tau>T_{m+1}\right\}}$. We have that
\begin{equation*} \begin{aligned}
    \mathbb{E^Q}\left[e_{m+1}\mathbbm{1}_{B_m^1}\Big|\mathcal{F}_0\right]=\mathbb{E^Q}\left[\left(\frac{h_{m+1}(\mathbf{x}_{T_{m+1}})}{B(T_{m+1})}-\frac{h_{m+1}(\mathbf{x}_{T_{m+1}})}{B(T_{m+1})}\right)\mathbbm{1}_{B_m^1}\bigg|\mathcal{F}_0\right]=0
\end{aligned} \end{equation*}
Furthermore we have that $\mathbbm{1}_{B_m^2}+\mathbbm{1}_{A_{m+1}}=\mathbbm{1}_{B_{m+1}}$. Therefore we can infer that
\begin{equation*} \begin{aligned}
    \mathbb{E^Q}\left[e_m\mathbbm{1}_{B_m}\Big|\mathcal{F}_0\right]+\mathbb{E^Q}\left[e_{m+1}\mathbbm{1}_{A_{m+1}}\Big|\mathcal{F}_0\right]&<\varepsilon+\mathbb{E^Q}\left[e_{m+1}\mathbbm{1}_{B_m^2}\Big|\mathcal{F}_0\right]+\mathbb{E^Q}\left[e_{m+1}\mathbbm{1}_{A_{m+1}}\Big|\mathcal{F}_0\right]\\
    &=\varepsilon+\mathbb{E^Q}\left[e_{m+1}\mathbbm{1}_{B_{m+1}}\Big|\mathcal{F}_0\right]
\end{aligned} \end{equation*}
Together with the fact that $A_0\equiv B_0$, we conclude by induction on $m$ that
\begin{equation*} \begin{aligned}
    E_2&\leq\mathbb{E^Q}\left[e_0\mathbbm{1}_{B_0}\Big|\mathcal{F}_0\right]+\sum_{m=1}^{M-1}\mathbb{E^Q}\left[e_m\mathbbm{1}_{A_m}\Big|\mathcal{F}_0\right]\\
    &<\varepsilon+\mathbb{E^Q}\left[e_1\mathbbm{1}_{B_1}\Big|\mathcal{F}_0\right]+\sum_{m=2}^{M-1}\mathbb{E^Q}\left[e_m\mathbbm{1}_{A_m}\Big|\mathcal{F}_0\right]\\
    &\qquad\qquad\qquad\qquad\quad\vdots\\
    &<(M-1)\varepsilon+\mathbb{E^Q}\left[e_{M-1}\mathbbm{1}_{B_{M-1}}\Big|\mathcal{F}_0\right]=(M-1)\varepsilon
\end{aligned} \end{equation*}
 \textbf{Bounding $\mathbf{E_3}$:} We finalise the proof by considering the third event $\{\tau > \Tilde\tau\}$. In a similar fashion as before, we introduce two types of sub-events: $A_m:=\left\{\tilde\tau=T_m\land \tau>T_m\right\}$ and $B_m:=\left\{\tilde\tau\leq T_m\land \tau>T_m\right\}$. Also again define a difference-process, this time given by $e_m:=\frac{h_{\tau}(\mathbf{x}_{\tau})}{B(\tau)}-\frac{\tilde V(T_m)}{B(T_m)}$. It should be clear that $\mathbbm{1}_{\{\tau>\tilde\tau\}}=\sum_{m=0}^{M-1}\mathbbm{1}_{A_m}$. Therefore it holds that
\begin{equation*} \begin{aligned}
    E_3=\sum_{m=0}^{M-1}\mathbb{E^Q}\left[\left(\frac{h_{\tau}(\mathbf{x}_\tau)}{B(\tau)}-\frac{h_{\tilde\tau}(\mathbf{x}_{\tilde\tau})}{B(\tilde\tau)}\right)\mathbbm{1}_{A_m}\bigg|\mathcal{F}_0\right]=\sum_{m=0}^{M-1}\mathbb{E^Q}\left[e_m\mathbbm{1}_{A_m}\Big|\mathcal{F}_0\right]
\end{aligned} \end{equation*}
where the second equality follows from the fact that the direct estimator has the property $\tilde V(\tilde\tau)=h_{\tilde\tau}$. Now we will show by induction that $E_3<(M-1)\varepsilon$. Note that for any $m\in\{0,\ldots,M-1\}$ we have that
\begin{equation} \begin{aligned}
    \mathbb{E^Q}\left[e_m\mathbbm{1}_{B_m}\Big|\mathcal{F}_0\right]&\leq\mathbb{E^Q}\left[\left(\frac{h_{\tau}(\mathbf{x}_{\tau})}{B(\tau)}-\mathbb{E^Q}\left[\frac{G_{m+1}\left(z_{m+1}\right)}{B\left(T_{m+1}\right)}\bigg|\mathcal{F}_{T_m}\right]\right)\mathbbm{1}_{B_m}\bigg|\mathcal{F}_0\right]\\
    &=\mathbb{E^Q}\left[\left(\frac{h_{\tau}(\mathbf{x}_{\tau})}{B(\tau)}-\frac{G_{m+1}\left(z_{m+1}\right)}{B\left(T_{m+1}\right)}\right)\mathbbm{1}_{B_m}\bigg|\mathcal{F}_0\right]\\
    &\leq\mathbb{E^Q}\left[\left|\frac{\tilde V(T_{m+1})}{B(T_{m+1})}-\frac{G_{m+1}\left(z_{m+1}\right)}{B\left(T_{m+1}\right)}\right|\mathbbm{1}_{B_m}\Bigg|\mathcal{F}_0\right]+\mathbb{E^Q}\left[e_{m+1}\mathbbm{1}_{B_m}\Big|\mathcal{F}_0\right]\label{eqn: prf LB 3}
\end{aligned} \end{equation}
The first inequality follows from the fact that $\tilde V(T_m)=\max\{\tilde C_m,h_m\}\geq \tilde C_m$. The subsequent equality follows from the tower rule in combination with the fact that $\mathbbm{1}_{B_m}$ is $\mathcal{F}_{T_m}-$measurable. The final inequality follows from an application of the triangle inequality. The first term in \ref{eqn: prf LB 3} is by assumption bounded by $\varepsilon$. The second term in \ref{eqn: prf LB 3} can be rewritten by observing that $\mathbbm{1}_{B_m}:=\mathbbm{1}_{B_m^1}+\mathbbm{1}_{B_m^2}:=\mathbbm{1}_{\left\{\tilde\tau\leq T_m\land \tau=T_{m+1}\right\}}+\mathbbm{1}_{\left\{\tilde\tau\leq T_{m}\land \tau>T_{m+1}\right\}}$. We have that
\begin{equation*} \begin{aligned}
    \mathbb{E^Q}\left[e_{m+1}\mathbbm{1}_{B_m^1}\Big|\mathcal{F}_0\right]=\mathbb{E^Q}\left[\left(\frac{h_{m+1}(\mathbf{x}_{T_{m+1}})}{B(T_{m+1})}-\frac{\Tilde V(T_{m+1})}{B(T_{m+1})}\right)\mathbbm{1}_{B_m^1}\bigg|\mathcal{F}_0\right]\leq0
\end{aligned} \end{equation*}
where the inequality follows from the fact that $\tilde V(T_{m+1})=\max\{\tilde C_{m+1},h_{m+1}\}\geq h_{m+1}$.
Furthermore we have that $\mathbbm{1}_{B_m^2}+\mathbbm{1}_{A_{m+1}}=\mathbbm{1}_{B_{m+1}}$. Therefore we can once again infer that
\begin{equation*} \begin{aligned}
    \mathbb{E^Q}\left[e_m\mathbbm{1}_{B_m}\Big|\mathcal{F}_0\right]+\mathbb{E^Q}\left[e_{m+1}\mathbbm{1}_{A_{m+1}}\Big|\mathcal{F}_0\right]&<\varepsilon+\mathbb{E^Q}\left[e_{m+1}\mathbbm{1}_{B_m^2}\Big|\mathcal{F}_0\right]+\mathbb{E^Q}\left[e_{m+1}\mathbbm{1}_{A_{m+1}}\Big|\mathcal{F}_0\right]\\
    &=\varepsilon+\mathbb{E^Q}\left[e_{m+1}\mathbbm{1}_{B_{m+1}}\Big|\mathcal{F}_0\right]
\end{aligned} \end{equation*}
Together with the fact that $A_0\equiv B_0$, we again conclude by induction on $m$ that
\begin{equation*} \begin{aligned}
    E_3&\leq\mathbb{E^Q}\left[e_0\mathbbm{1}_{B_0}\Big|\mathcal{F}_0\right]+\sum_{m=1}^{M-1}\mathbb{E^Q}\left[e_m\mathbbm{1}_{A_m}\Big|\mathcal{F}_0\right]\\
    &<\varepsilon+\mathbb{E^Q}\left[e_1\mathbbm{1}_{B_1}\Big|\mathcal{F}_0\right]+\sum_{m=2}^{M-1}\mathbb{E^Q}\left[e_m\mathbbm{1}_{A_m}\Big|\mathcal{F}_0\right]\\
    &\qquad\qquad\qquad\qquad\quad\vdots\\
    &<(M-1)\varepsilon+\mathbb{E^Q}\left[e_{M-1}\mathbbm{1}_{B_{M-1}}\Big|\mathcal{F}_0\right]=(M-1)\varepsilon
\end{aligned} \end{equation*}
\textbf{Conclusion:} We hence find that
\begin{equation*} \begin{aligned}
    V(0)-L(0)=E_1+E_2+E_3<0+(M-1)\varepsilon+(M-1)\varepsilon=2(M-1)\varepsilon
\end{aligned} \end{equation*}
\end{proof}

\section{Proof of \Cref*{theorem: upper bound}\label{proof: upper bound}}
\begin{proof}
The discounted true price process is a supermartingale under $\mathbb{Q}$. Therefore we have that $\frac{V(t)}{B(t)}=Y_t+Z_t$ for a Martingale $Y_t$ and a predictable process $Z_t$, which starts at zero (i.e. $Z_0=0$) and is strictly decreasing. Define a difference process on $\mathcal{T}$, given by $e_{T_m}=\frac{V(T_m)-G_m(z_{m})}{B(T_m)}$. We can rewrite Martingale $M_t$ as defined in \ref{eqn:martingale2} in terms of $e_t$ as follows:
\begin{equation*} \begin{aligned}
    M_{T_m}&=\frac{G_0(z_{0})}{B(T_0)}+\sum_{j=1}^{m} \left(\frac{G_j(z_{j})}{B(T_j)}-\mathbb{E^Q}\left[\frac{G_j(z_{j})}{B(T_j)}\bigg|\mathcal{F}_{T_{j-1}}\right]\right)\\
    &=Y_{T_m}-e_{T_0}-\sum_{j=1}^{m}\left(e_{T_j}-\mathbb{E^Q}\left[e_{T_j}\big|\mathcal{F}_{T_{j-1}}\right]\right)
\end{aligned} \end{equation*}
Substituting the expression for $M_t$ into the definition of $U(0)$ yields
\begin{equation*} \begin{aligned}
    U(0)&= M_0+\mathbb{E^Q}\left[\max_{T_m\in\mathcal{T}_f}\left\{\frac{h_m(\mathbf{x}_{T_m})}{B(T_m)}-M_{T_m}\right\}\bigg|\mathcal{F}_0\right]\\
    &=\mathbb{E^Q}\left[\frac{G_0(z_0)}{B(T_0)}\bigg|\mathcal{F}_{0}\right]+\mathbb{E^Q}\left[\max_{m\in\{0,\ldots,M-1\}}\left\{\frac{h_m(\mathbf{x}_{T_m})}{B(T_m)}-Y_{T_m}
    +e_{T_0}\vphantom{\sum_{j=1}^{m}\Bigg|}\right.\right.\\
    &\qquad\qquad\qquad\qquad\qquad\qquad\qquad+\left.\left.\vphantom{\sum_{j=1}^{m}\Bigg|}\sum_{j=1}^{m}\left(e_{T_j}-\mathbb{E^Q}\left[e_{T_j}\big|\mathcal{F}_{T_{j-1}}\right]\right)\right\}\Bigg|\mathcal{F}_0\right]\\
    &\leq \mathbb{E^Q}\left[\frac{V(T_0)}{B(T_0)}\bigg|\mathcal{F}_{0}\right]+\mathbb{E^Q}\left[\max_{m\in\{0,\ldots,M-1\}}\left\{\sum_{j=1}^{m}\left(e_{T_j}-\mathbb{E^Q}\left[e_{T_j}\big|\mathcal{F}_{T_{j-1}}\right]\right)\right\}\Bigg|\mathcal{F}_0\right]
\end{aligned} \end{equation*}
The last step follows by merging $\mathbb{E^Q}\left[e_{T_0}\big|\mathcal{F}_{0}\right]$ with $M_0$ and by noting that $\frac{h_m(\mathbf{x}_{T_m})}{B(T_m)}-Y_{T_m}\leq\frac{V(T_m)}{B(T_m)}-Y_{T_m}=Z_{T_m}\leq0$. The remaining inequality is not easy to bound \cite{andersen2004primal}. However, by taking the absolute values of the difference process, we can obtain a loose bound as follows 
\begin{equation*} \begin{aligned}
    U(0)&\leq V(0)+\mathbb{E^Q}\left[\max_{m\in\{0,\ldots,M-1\}}\left\{\sum_{j=1}^{m}\left|e_{T_j}\right|+\sum_{j=1}^m\left|\mathbb{E^Q}\left[e_{T_j}\big|\mathcal{F}_{T_{j-1}}\right]\right|\right\}\Bigg|\mathcal{F}_0\right]\\
    &\leq V(0)+\mathbb{E^Q}\left[\sum_{j=1}^{M-1}\left|e_{T_j}\right|+\sum_{j=1}^{M-1}\left|\mathbb{E^Q}\left[e_{T_j}\big|\mathcal{F}_{T_{j-1}}\right]\right|\Bigg|\mathcal{F}_0\right]\\
    &\leq V(0)+2\sum_{j=1}^{M-1}\mathbb{E^Q}\left[\left|e_{T_j}\right|\Big|\mathcal{F}_{0}\right]
\end{aligned} \end{equation*}
Note that as a consequence of \cref{theorem: MAE error} have that $\mathbb{E^Q}\left[\left|e_{T_m}\right|\Big|\mathcal{F}_{0}\right]<(M-m)\varepsilon$. It follows
\begin{equation*} \begin{aligned}
    \left|U(0)-V(0)\right|<2\sum_{m=1}^{M-1}(M-m)\varepsilon=M(M-1)\varepsilon
\end{aligned} \end{equation*}
This concludes the proof.
\end{proof}

\end{document}


\maketitle

\section{A detailed example}

Here we include some equations and theorem-like environments to show
how these are labeled in a supplement and can be referenced from the
main text.
Consider the following equation:
\begin{equation}
  \label{eq:suppa}
  a^2 + b^2 = c^2.
\end{equation}
You can also reference equations such as \cref{eq:matrices,eq:bb} 
from the main article in this supplement.

\lipsum[100-101]

\begin{theorem}
  An example theorem.
\end{theorem}

\lipsum[102]
 
\begin{lemma}
  An example lemma.
\end{lemma}

\lipsum[103-105]

Here is an example citation: \cite{KoMa14}.

\section[Proof of Thm]{Proof of \cref{thm:bigthm}}
\label{sec:proof}

\lipsum[106-112]

\section{Additional experimental results}
\Cref{tab:foo} shows additional
supporting evidence. 

\begin{table}[htbp]
{\footnotesize
  \caption{Example table}  \label{tab:foo}
\begin{center}
  \begin{tabular}{|c|c|c|} \hline
   Species & \bf Mean & \bf Std.~Dev. \\ \hline
    1 & 3.4 & 1.2 \\
    2 & 5.4 & 0.6 \\ \hline
  \end{tabular}
\end{center}
}
\end{table}

\usepackage[numbers]{natbib}
\bibliographystyle{siamplain}
\bibliography{references}